\documentclass[11pt]{article}

\usepackage{times}
\usepackage{color, xcolor, colortbl}
\usepackage{graphicx}
\usepackage{epstopdf}
\ifpdf
\DeclareGraphicsExtensions{.eps,.pdf,.png,.jpg}
\else
\DeclareGraphicsExtensions{.eps}
\fi

\usepackage{geometry}
\usepackage{amsmath,amssymb, amsthm}
\usepackage{algorithm}
\usepackage{algorithmic}
\usepackage{bm}
\usepackage[caption=false]{subfig}
\usepackage[title]{appendix}
\usepackage{multirow}
\usepackage{braket}
\usepackage[english]{babel}
\usepackage{hyperref}
\usepackage[capitalize]{cleveref}
\usepackage[square,numbers]{natbib}
\usepackage{adjustbox}
\usepackage{xspace}
\usepackage{comment}
\usepackage{booktabs}
\usepackage{makecell}
\usepackage{threeparttable}

\newcommand{\bvec}[1]{\mathbf{#1}}

\newcommand{\valpha}{ {\bm{\alpha}} }
\newcommand{\vbeta}{ {\bm{\beta}} }
\newcommand{\va}{\bvec{a}}
\newcommand{\vb}{\bvec{b}}

\newcommand{\vk}{\bvec{k}}

\newcommand{\vq}{\bvec{q}}
\newcommand{\vr}{\bvec{r}}

\newcommand{\vv}{\bvec{v}}

\newcommand{\vx}{\bvec{x}}
\newcommand{\vy}{\bvec{y}}
\newcommand{\vz}{\bvec{z}}

\newcommand{\vG}{\bvec{G}}

\newcommand{\conj}[1]{\overline{#1}}

\newcommand{\I}{\mathrm{i}}

\newcommand{\mc}[1]{\mathcal{#1}}

\newcommand{\wt}[1]{\widetilde{#1}}

\newcommand{\abs}[1]{\left\lvert#1\right\rvert}

\newcommand{\ud}{\,\mathrm{d}}
\newcommand{\Or}{\mathcal{O}}

\newcommand{\RR}{\mathbb{R}}

\newtheorem{thm}{\protect\theoremname}
\theoremstyle{plain}

\theoremstyle{plain}
\newtheorem{rem}[thm]{\protect\remarkname}
\theoremstyle{plain}

\theoremstyle{plain}
\newtheorem{cor}[thm]{\protect\corollaryname}

\providecommand{\definitionname}{Definition}
\providecommand{\assumptionname}{Assumption}
\providecommand{\corollaryname}{Corollary}
\providecommand{\lemmaname}{Lemma}
\providecommand{\propositionname}{Proposition}
\providecommand{\remarkname}{Remark}
\providecommand{\theoremname}{Theorem}

\newcommand{\xsum}{\mathop{\sum\nolimits'}}

\crefname{equation}{Eq.}{Eqs.}
\crefformat{equation}{Eq.~#2(#1)#3}

\newcommand{\XX}[1]{\textcolor{blue}{[XX: #1]}}

\newcommand{\REV}[1]{\textcolor{black}{#1}}

\numberwithin{equation}{section}
\numberwithin{figure}{section}
\numberwithin{table}{section}
\numberwithin{thm}{section}

\title{Unified analysis of finite-size error for periodic Hartree-Fock and second order M{\o}ller-Plesset perturbation theory}

\usepackage{fancyhdr}
\usepackage{authblk}

\author[1]{Xin Xing}
\author[2]{Xiaoxu Li}
\author[3]{Lin Lin \thanks{Mailing address: 817 Evans Hall, University Dr, Berkeley, CA 94720}}
\affil[1]{Department of Mathematics, University of California, Berkeley, CA 94720}
\affil[2]{Department of Mathematics, University of California,
					Berkeley, CA 94720, USA, and 
					College of Education for the Future, Beijing Normal University at Zhuhai, Guangdong 519087, P.R.China.}
\affil[3]{Department of Mathematics, University of California, Berkeley, and Applied Mathematics and Computational Research
					Division, Lawrence Berkeley National Laboratory, Berkeley, CA 94720}

\date{}

\providecommand{\keywords}[1]{\textbf{\textit{Keywords:}} #1}
\providecommand{\msc}[1]{\textbf{\textit{2010 MSC:}} #1}
\begin{document}

\maketitle 

%
%

\begin{abstract}
Despite decades of practice, finite-size errors in many widely used electronic structure theories for periodic systems remain poorly understood. 
For periodic systems using a general Monkhorst-Pack grid, there has been no \REV{comprehensive and} rigorous analysis of the finite-size error in the Hartree-Fock theory (HF) and the second order M{\o}ller-Plesset perturbation theory (MP2), which are the simplest wavefunction based method, and the simplest post-Hartree-Fock method, respectively.
Such calculations can be viewed as a multi-dimensional integral discretized with certain trapezoidal rules. 
Due to the Coulomb singularity, the integrand has many points of discontinuity in general, and standard error analysis based on the Euler-Maclaurin formula gives overly pessimistic results. 
The lack of analytic understanding of finite-size errors also impedes the development of effective finite-size correction schemes.
We propose a unified \REV{analysis} to obtain  sharp convergence rates of  finite-size errors for the periodic HF and MP2 theories. 
Our main technical advancement is a generalization of the result of [Lyness, 1976] for obtaining sharp convergence rates of the trapezoidal rule for a class of non-smooth integrands. 
Our result is applicable to three-dimensional bulk systems as well as low dimensional systems (such as nanowires and 2D materials). Our unified analysis also allows us to prove the effectiveness of the Madelung-constant correction to the Fock exchange energy, and the effectiveness of a recently proposed staggered mesh method for periodic MP2 calculations \REV{[Xing, Li, Lin, J. Chem. Theory Comput. 2021]}.
Our analysis connects the effectiveness of the staggered mesh method with integrands with removable singularities, and suggests a new staggered mesh method for reducing finite-size errors of periodic HF calculations.
\end{abstract}
\noindent\msc{81Q99, 65G99, 65D32}

\noindent\keywords{finite-size error, periodic systems, Hartree-Fock, MP2 perturbation theory}

\section{Introduction}
Accurate estimate of ground state energies of periodic systems (e.g., crystals, nanotubes, nanowires, 2D materials, and surfaces) is of immense importance in quantum physics, chemistry, and materials science.
The simplest wavefunction based electronic structure theory is the Hartree-Fock (HF) theory, and the simplest post-HF wavefunction based method is the second order M{\o}ller-Plesset perturbation theories (MP2) (see e.g., \cite{SzaboOstlund1989,ShavittBartlett2009}).
The HF and MP2 theories are also ingredients in other electronic structure theories, such as in constructing accurate exchange-correlation energy functionals in Kohn-Sham density functional theory (DFT)~\cite{KohnSham1965,Martin2008}.  HF calculations for periodic systems have been routinely performed for decades. Despite the relatively large computational cost,  post-HF methods such as MP2 and coupled cluster (CC) theories have also been increasingly routinely performed for periodic systems~\cite{MarsmanGruneisPaierKresse2009,GruneisMarsmanKresse2010, MuellerPaulus2012,SchaeferRambergerKresse2017,McClainSunChanEtAl2017,GruberLiaoTsatsoulisEtAl2018}, thanks to the improvement of numerical algorithms and the increase of computational power.

For periodic systems, \REV{a fundamental physical quantity} is the energy per unit cell in the thermodynamic limit (TDL). 
The error between the energy computed from a finite-sized system and the exact value is called the finite-size error.
Due to the steep increase of the computational cost with respect to the system size (in particular for post-HF methods), reaching convergence in a brute force fashion is often beyond reach, and corrections to finite-size errors must be applied.
In order to develop finite-size correction schemes, accurate understanding of the scaling of the finite-size error is needed.

To the best of our knowledge, the finite-size errors of the Fock exchange energy  and the MP2 correlation energy have not been rigorously analyzed. (The correlation energy is defined to be the difference between the ground state energy of the post-HF theory and the HF ground state energy.
Throughout the paper, the exchange energy and the MP2 energy stand for the Fock exchange energy and the MP2 correlation energy, respectively.)
In a nutshell, let $\Omega$ be the unit cell of a periodic system, and $\Omega^*$ be the first Brillouin zone (BZ). 
The  exchange and the MP2  energies in the TDL can both be compactly written as the following integrals over $\Omega^*$: 
\begin{align}
E_\text{X}^{\text{TDL}} 
&= 
\int_{\Omega^*} \ud \vk_i \int_{\Omega^*} \ud \vk_j  F_\text{X}(\vk_i, \vk_j),
\label{eqn:exchange_tdl}
\\
E_\text{MP2}^{\text{TDL}} 
&= 
\int_{\Omega^*} \ud \vk_i \int_{\Omega^*} \ud \vk_j  \int_{\Omega^*} \ud \vk_a   F_\text{MP2}(\vk_i, \vk_j, \vk_a),
\label{eqn:mp2_tdl}
\end{align}
where, following the convention in quantum chemistry, $\vk_i, \vk_j$ ($\vk_a, \vk_b$) are crystal momentum vectors associated with the occupied (virtual) bands,  respectively.
Here $\Omega^*\subset \RR^d$ with $d=3$ for 3D periodic systems, and $d=1,2$ for quasi-1D and quasi-2D systems.
 

In  numerical calculations, the Brillouin zone $\Omega^*$ is first discretized by a uniform mesh $\mathcal{K}$ (called the Monkhorst-Pack mesh~\cite{MonkhorstPack1976}) with $N_\vk$ points in total. 
The exchange and the MP2 energies in \cref{eqn:exchange_tdl} and \cref{eqn:mp2_tdl} are then approximated by \textit{trapezoidal  rules} (see \Cref{sec:integral_form} for the precise definition of trapezoidal rules in the current context) that sample $\vk_i$, $\vk_j$, and $\vk_a$ on $\mathcal{K}$ as 
\begin{align}
E_\text{X}(N_\vk) 
& = \frac{\abs{\Omega^*}^2}{N_\vk^2}\sum_{\vk_i,\vk_j \in \mathcal{K}} F_\text{X}(\vk_i, \vk_j),
\label{eqn:exchange_n}\\
E_\text{MP2}(N_\vk) 
& = \frac{\abs{\Omega^*}^3}{N_\vk^3}\sum_{\vk_i,\vk_j, \vk_a\in \mathcal{K}} F_\text{MP2}(\vk_i, \vk_j, \vk_a).
\label{eqn:mp2_n}
\end{align}
Hence the general form of the energy in the TDL and the numerical scheme can be written  as
\begin{align}
E_{*}^{\text{TDL}} 
&= 
\int_{(\Omega^*)^s} F_{*}(\vx_1, \cdots, \vx_s) \ud \vx_1 \cdots \ud \vx_s,
\label{eqn:TDL_general}\\
E_{*}(N_\vk) 
& = \left(\frac{\abs{\Omega^*}}{N_\vk}\right)^s \sum_{\vx_1, \cdots, \vx_s \in (\mathcal{K})^s} F_{*}(\vx_1, \cdots, \vx_s),
\end{align}
where the subscript $*$ can be ``X'' or ``MP2'' and each $\vx_i\in \RR^d$.
We have $s=2$ for the exchange energy, and $s=3$ for the MP2  energy, respectively.
Eq.~\eqref{eqn:TDL_general} can be a high dimensional \REV{integral}.
For instance, MP2 calculations for 3D periodic systems require the evaluation of a $9$-dimensional integral.

The finite-size error, i.e., $E_*^\text{TDL} - E_*(N_\vk)$, thus can be interpreted as the error of the numerical quadrature. 
At first glance, it may seem that due to periodicity of the integrand $F_*$, the quadrature error should readily follow from standard numerical analysis of trapezoidal rules for periodic functions (e.g., \cite{JavedTrefethen2014,TrefethenWeideman2014}).
However, due to the subtle nature of the Coulomb singularity, the integrand is generally discontinuous at certain points.
As a result, standard error analysis based on the Euler-Maclaurin formula gives overly pessimistic results: not only the convergence rate is not sharp, direct application of the Euler-Maclaurin formula fails to demonstrate the convergence of the exchange and MP2 energy calculations towards the thermodynamic limit (see \Cref{sec:standard_em_quad}).

\REV{
        Besides the energy, many physical observables can be similarly represented as integrals over the first Brillouin zone. 
        While the Monkhorst-Pack mesh is perhaps the most widely used method for discretizing the Brillouin zone, other choices are also available such as the tetrahedron method \cite{blochl1994improved}. We refer readers to \cite{CancesEhrlacherGontierEtAl2020} for a more detailed discussion.
        Following a similar approach as developed in this paper, the finite-size errors in these more general settings may be analyzed as well. 
}

\vspace{1em}
\noindent\textbf{Contributions:}

In this paper, we establish the first \REV{comprehensive and} rigorous analysis of the quadrature errors in exchange and MP2 energy calculations \REV{for insulating systems with a direct gap and without topological obstructions~\cite{BrouderPanatiCalandraEtAl2007,MonacoPanatiPisanteEtAl2018}}. 
Our convergence rates are sharp for general systems and match numerical observations.
The key component of our analysis is a new Euler-Maclaurin type of formula in \cref{thm:em_n_fraction}, which generalizes the classical result by Lyness~\cite{Lyness1976}, and can predict the sharp convergence rate of the trapezoidal quadrature error for a class of non-smooth functions (including the integrand $F_{*}$ as special cases).
Using this formula, we can rigorously analyze the finite-size errors of the standard as well as a number of improved methods for exchange and MP2 energy calculations.
The main results are summarized in \cref{tab:main_result1} for 3D systems and \cref{tab:main_result2} for low-dimensional (quasi-1D and quasi-2D) systems. 

The staggered mesh method for exchange energy calculations (corresponding to the intersection of ``Staggered mesh method'' and ``Madelung-Exchange'' in \cref{tab:main_result1} and \cref{tab:main_result2}) is a new scheme. 
It is worth noting that the staggered mesh method only requires the computation of orbitals and orbital energies on an additional Monkhorst-Pack grid.
In electronic structure calculations, this only requires a set of non-self-consistent calculations, and the additional cost can be negligible.

Due to the appearance of a logarithmic dependence in \cref{thm:em_n_fraction}, the asymptotic scaling of the finite-size error $\delta E$ with respect to $N_\vk$ always involves a multiplicative $\ln N_\vk$ term.
For brevity of  notation, throughout the paper we slightly abuse the big-O notation when expressing the finite-size error of energies, i.e., $\delta E = \Or(N_\vk^{-\alpha})$ means that $\abs{\delta E }\leqslant C N_\vk^{-\alpha} \ln N_\vk$ for some constant $C$ when $N_\vk$ is sufficiently large.


\begin{table}[htbp]
        \centering
        \caption{Theoretical estimate of the quadrature errors in the calculations of the exchange energy, the Madelung-corrected exchange energy, and the MP2 energy for 3D periodic systems. 
        When combined with the staggered mesh method, the Madelung constant correction requires minor modifications as detailed in \Cref{subsec:staggered_ex}.
        The integrand $F_{*}$ of special systems has removable discontinuities, as detailed in \Cref{sec:removable}.
}
        \label{tab:main_result1}
        \begin{tabular}{rccc}
                \toprule 
                & Exchange & Madelung-Exchange & MP2   \\
                \midrule 
                Standard method  & $ \Or(N_\vk^{-\frac13})$         &   $\Or(N_\vk^{-1})$    &  $\Or(N_\vk^{-1})$  
                \\[0.5cm]
                \makecell[r]{Staggered mesh method\\ for general systems}  & $\Or(N_\vk^{-\frac13})$    &   $\Or(N_\vk^{-1})$  &  $\Or(N_\vk^{-1})$
                \\[0.5cm]
                \makecell[r]{Staggered mesh method\\ for special systems\tnote{d}} & $\Or(N_\vk^{-\frac13})$ &   $\Or(N_\vk^{-\frac53})$     &  $\Or(N_\vk^{-\frac53})$ \\
                \bottomrule 
        \end{tabular}
\end{table}

\begin{table}[htbp]
        \centering
        \caption{Theoretical estimate of  the quadrature errors  in the calculations of the Madelung-corrected exchange energy and the MP2 energy for quasi-1D and quasi-2D systems. 
        We consider a specific model for  low-dimensional systems where the Madelung constant correction is added to the exchange energy calculation by default.
        Notation `SA' means that the error decays super-algebraically, i.e., faster than $N_\vk^{-s}$ with any $s> 0$. 
        Modifications of the Madelung constant correction is needed for the staggered mesh method as detailed in \cref{appendix_low_dim}. 
        The integrand $F_{*}$ of special systems has removable discontinuities, as detailed in \Cref{sec:removable}.
        }
        \label{tab:main_result2}
        \begin{tabular}{r|cc|cc}
                \toprule 
                      &  \multicolumn{2}{c|}{Quasi-2D} & \multicolumn{2}{c}{Quasi-1D} \\
                                                                                                        &  Madelung-Exchange &  MP2   & Madelung-Exchange & MP2   \\
           \midrule 
				Standard method     &      $\Or(N_\vk^{-1})$   &  $\Or(N_\vk^{-1})$ &   $\Or(N_\vk^{-1})$  &  $\Or(N_\vk^{-1})$              
				\\[0.5cm]
				\makecell[r]{Staggered mesh method\\ for general systems}&      $ \Or(N_\vk^{-1})$   &  $\Or(N_\vk^{-1})$ &   SA  &  SA   
				\\[0.5cm]
				\makecell[r]{Staggered mesh method\\ for special systems}&      $\Or(N_\vk^{-2})$   &  $\Or(N_\vk^{-2})$ &   SA  &  SA        \\
             \bottomrule 
        \end{tabular}
%
\end{table}

\vspace{2em}
\noindent\textbf{Main idea:}

As the first step of our unified approach for finite-size error analysis, we
reformulate the exchange and MP2 energy calculations into the quadrature forms
in \cref{eqn:exchange_n} and \cref{eqn:mp2_n}, and obtain the explicit
representations of $F_\text{X}$ and $F_\text{MP2}$.  The finite-size error
analysis then becomes the classical numerical analysis problem of estimating the quadrature error of a trapezoidal rule for certain special integrands.

First we show that both $F_\text{X}(\vk_i,\vk_j)$ and $F_\text{MP2}(\vk_i, \vk_j,
\vk_a)$ are periodic with respect to each variable over $\Omega^*$, but
are discontinuous at points where $\vk_j - \vk_i = \bm{0}$ for $F_\text{X}$ and
$\vk_a - \vk_i = \bm{0}$ or $\vk_a - \vk_j  = \bm{0}$ for $F_\text{MP2}$ if
restricting $\vk_i,\vk_j,\vk_a$ to $\Omega^*$.  
We identify that the finite-size error is dominated by the quadrature error for a class of non-smooth functions.  Specifically, the quadrature error of $F_\text{X}$ is
dominated by  that of a non-smooth component of the form
\[
\dfrac{f(\vq)}{|\vq|^2}\ \text{with}\ f(\vq)  = \Or(|\vq|^{a}),
\]
where $f(\vq)$ is a generic smooth function compactly supported in $\Omega^*$, with $a = 0$ in the standard exchange energy calculation and $a = 2$ in the Madelung-corrected case.  Here $\vq$ denotes the minimum image of $\vk_j - \vk_i$ in $\Omega^*$ and the trapezoidal rule for this component is over $\vq \in \Omega^*$ with an $N_\vk$-sized MP mesh $\mathcal{K}_\vq$ induced by the definition of $\vq$ with $\vk_i,\vk_j \in \mathcal{K}$.  Similarly, the quadrature error of $F_\text{MP2}$ is dominated by that of the non-smooth  components of the forms
\begin{equation*}
\begin{split}
&
\dfrac{f(\vq_1)}{|\vq_1|^2}\ \text{with}\ f(\vq_1) = \Or(|\vq_1|^2), \qquad  \dfrac{f(\vq_1)}{|\vq_1|^4}\ \text{with}\ f(\vq_1) = \Or(|\vq_1|^4),
\\
& 
\dfrac{f_1(\vq_1, \vq_2)}{|\vq_1|^2}\dfrac{f_2(\vq_1, \vq_2)}{|\vq_2|^2}\ \text{with}\  f_1(\vq_1,\vq_2) = \Or(|\vq_1|^2), f_2(\vq_1,\vq_2) = \Or(|\vq_2|^2), 
\end{split}
\end{equation*}
where $\vq_1$ and $\vq_2$ are the minimum images of $\vk_i - \vk_a$ and $\vk_j - \vk_a$ in $\Omega^*$, respectively, and share the same MP mesh $\mathcal{K}_\vq$ in the corresponding trapezoidal rules. Here, $f, f_1, f_2$ denote generic smooth functions compactly supported in $\Omega^*$ or $\Omega^*\times \Omega^*$. 

The remaining problem is to analyze the quadrature errors for the non-smooth functions above. 
For a general function $g(\vx)$ compactly supported in a hypercube $V\subset \mathbb{R}^d$, the quadrature error of $\int_V g(\vx)\ud\vx$ can be analyzed using the standard Euler-Maclaurin formula: if $g(\vx)$ has continuous derivatives up to order $s$,  the quadrature error with a uniform mesh of size $m^d$ scales as $\Or(m^{-s})$. However, the non-smooth terms we identified above may be discontinuous and have unbounded first-order derivatives.
Thus,  direct application of the standard Euler-Maclaurin formula predicts that the quadrature errors in both the exchange and MP2 energy calculations do not decay at all with respect to $N_{\vk}$! 

This overly pessimistic estimate above can however be significantly improved. 
The key technical step of our analysis is to generalize a classical result by Lyness~\cite{Lyness1976} on the quadrature error for homogeneous functions, and obtain a special Euler-Maclaurin type of formula in \cref{thm:em_n_fraction} that works for non-smooth functions in the general form (which we refer to as the ``fractional form'')
\begin{equation}\label{eqn:fractional_form}
g(\vx_1, \vx_2, \ldots, \vx_n) = 
\dfrac{f_1(\vx_1, \vx_2, \ldots, \vx_n)}{ (\vx_1^T M \vx_1)^{p_1}} 
\dfrac{f_2(\vx_1, \vx_2, \ldots, \vx_n)}{ (\vx_2^T M \vx_2)^{p_2}} 
\cdots 
\dfrac{f_n(\vx_1, \vx_2, \ldots, \vx_n)}{  (\vx_n^T M \vx_n)^{p_n}},
\end{equation}
where $M$ is a symmetric positive definite matrix. 
For each $i = 1,\ldots, n$, let $f_i$ be smooth and scale as $\Or(|\vx_i|^{a_i})$ near $\vx_i = \bm{0}$ so that $f_i/(\vx_i^T M \vx_i)^{p_i} = \Or(|\vx_i|^{\gamma_i})$  with $\gamma_i = a_i - 2p_i$. 
Based on this special formula and further assuming  $g$ to be compactly supported in $V^{\times n}$ with a hypercube $V$, we prove in \cref{cor:em_n_fraction} that the quadrature error for $\int_{V^{\times n}} g\ud\vx_1\ldots\ud\vx_n$ using an $m^{nd}$-sized uniform mesh scales as $\mathcal{O}(m^{-(d+\min_i \gamma_i)}\ln m)$. 
Applying this result to our concerned non-smooth terms above, we obtain the $\Or(N_\vk^{-\frac13})$ and $\Or(N_\vk^{-1})$ quadrature error estimates in the exchange and Madelung-corrected exchange energies, and $\Or(N_\vk^{-1})$ quadrature error estimate in MP2 energy. 

For low-dimensional systems and special systems with removable discontinuities in $F_\text{X}$ and $F_\text{MP2}$, similar application of the special Euler-Maclaurin formula can be used to obtain the corresponding quadrature error estimates.

The Madelung constant correction is commonly used in practice to  reduce the $\Or(N_\vk^{-\frac13})$ finite-size error in the exchange energy calculation for 3D periodic systems, and is also used directly in the model Hamiltonian with a shifted Ewald kernel for periodic systems \cite{FraserFoulkesRajagopalEtAl1996}. 
This correction is originally introduced to remove the artificial interactions between particles and its periodic images in the supercell model (this model is equivalent to a special choice of the MP mesh $\mathcal{K}$). 
Numerical observations as well as heuristic arguments suggest that the finite-size error of the corrected scheme scales as $\Or(N_\vk^{-1})$.  
By analyzing the quadrature error, we rigorously prove this error scaling and justify the effectiveness of the correction (see \cref{tab:main_result1} and \cref{tab:main_result2}). 
The key, also an interesting finding, is the close connection between the Madelung constant correction and a quadrature technique called the singularity subtraction method. 
Specifically, the dominant $\Or(N_\vk^{-\frac13})$ error turns out to come from the leading non-smooth term of $F_\text{X}(\vk_i, \vk_j)$ in the form $\frac{C}{|\vq|^2}$ with a constant $C$ and $\vq = \vk_j - \vk_i$. 
We show that the correction is equivalent to first subtracting $\frac{C}{|\vq|^2}$ from $F_\text{X}(\vk_i, \vk_j)$, applying the trapezoidal rule over the remainder, and then adding back the contribution of the subtracted term.
As a result, the leading non-smooth term of $F_\text{X}$ is integrated exactly while the remainder, with improved smoothness condition, can be shown to have $\Or(N_\vk^{-1})$ quadrature error.
 

As shown by our analysis, the discontinuity of $F_\text{X}(\vk_i,\vk_j)$  and $F_\text{MP2}(\vk_i,\vk_j,\vk_a)$ is the main cause that leads to $\Or(N_\vk^{-1})$ quadrature error in the Madelung-corrected exchange energy, as well as MP2 energy calculations. 
The two functions are discontinuous at $\vk_j - \vk_i= \bm{0}$, and $\vk_a - \vk_i = \bm{0}$ or $\vk_a  - \vk_j = \bm{0}$, respectively. 
However, in many special systems, the points of discontinuity in $F_\text{X}$ and $F_\text{MP2}$ may become removable, i.e., by properly defining their function values at discontinuous points, $F_\text{X}$ and $F_\text{MP2}$ can become continuous. 
Unfortunately, the standard methods for calculating the exchange and MP2 energies (\cref{eqn:exchange_n} and \cref{eqn:mp2_n}) use the same mesh $\mathcal{K}$ for $\vk_i$, $\vk_j$, and $\vk_a$.
Therefore certain quadrature nodes are always placed at the points of discontinuity, and the resulting quadrature error remains $\Or(N_\vk^{-1})$. 
Inspired by this observation, we previously proposed the staggered mesh method for computing the MP2 energy~\cite{XingLiLin2021}, which uses one mesh for occupied orbitals $\vk_i,\vk_j$ and another staggered mesh for unoccupied orbitals $\vk_a$ to avoid sampling these discontinuous points (i.e., $\vk_a - \vk_i = \bm{0}$ or $\vk_a - \vk_j = \bm{0}$).
In this paper, we generalize the staggered mesh method for the exchange energy calculation. 
As listed in \cref{tab:main_result1} and \cref{tab:main_result2}, we demonstrate that the quadrature error of the staggered mesh method can be $o(N_\vk^{-1})$ for both the Madelung-corrected exchange and MP2 energy calculations, when the integrand discontinuities are removable. 
Especially  for quasi-1D systems, the integrand $F_\text{X}$ and $F_\text{MP2}$ can always be improved to become smooth functions, and the quadrature error decays super-algebraically.

\vspace{1em}
\noindent\textbf{Related works:}

There have been many works on the heuristic understanding of finite-size errors in electronic structure calculations (e.g., \cite{MakovPayne1995} for analyzing the finite-size error of the electrostatic interaction in periodic, aperiodic and charged systems), as well as numerical schemes to correct finite-size errors in various contexts.
It is worth noting that many finite-size correction schemes originate from the context of quantum Monte Carlo (QMC) calculations (e.g., \cite{FraserFoulkesRajagopalEtAl1996,ChiesaCeperleyMartinEtAl2006,FoulkesMitasNeedsEtAl2001,DrummondNeedsSorouriEtAl2008,HolzmannClayMoralesEtAl2016}).
Some correction methods rely on truncating the Coulomb operator in the
real space (e.g., \cite{spencer08,SundararamanArias2013}). 
However, these methods are designed for Fock exchange energy calculations, and rely on decay properties of the single-particle density matrix in the real space (for gapped systems).
In particular, such truncated Coulomb operator should not be used in MP2 calculations. In this paper, we focus on periodic HF and MP2 calculations in the reciprocal space using the standard Coulomb operator, as well as an arbitrary Monkhorst-Pack grid. To our knowledge, there is no rigorous analysis of the finite-size error in this context.
 
A generic way to correct the finite-size errors is to perform a power-law extrapolation~\cite{MarsmanGruneisPaierKresse2009,BoothGruneisKresseAlavi2013, McClainSunChanEtAl2017, MihmYangShepherd2020}. 
It fits the energies from several calculations with different values of $N_{\vk}$ using a power function of $N_\vk$ to estimate $E_*^\text{TDL}$. 
This approach is simple and often effective, but does not provide understanding of the finite-size errors from first principles.
Furthermore, the precise form of the power-law extrapolation is often debatable at least in the pre-asymptotic regime (see e.g., \cite{FreysoldtNeugebauerVan2009}).
For QMC, MP2, and coupled cluster (CC) calculations, another common tool is to analyze the structure factor, and the corresponding correction scheme is called structure factor interpolation method \cite{ChiesaCeperleyMartinEtAl2006,LiaoGrueneis2016,GruberLiaoTsatsoulisEtAl2018}.
By analyzing the structure factor in MP2/CC calculations, it has been proposed that the finite-size error should scale as  $\Or(N_\vk^{-1})$, and is due to the omission of terms related to the singularity of the Coulomb kernel~\cite{LiaoGrueneis2016,GruberLiaoTsatsoulisEtAl2018}.
The corresponding correction scheme interpolates the structure factor, and then approximates the missing term via extrapolation. 
\REV{According to our analysis, this missing term contributes to a portion of the quadrature error  related to volume elements containing the Coulomb singularity, which is $\Or(N_\vk^{-1})$.
Our analysis also indicates that the remaining volume elements not containing the Coulomb singularity also have  significant contribution to the quadrature error, which is also $\Or(N_\vk^{-1})$.
(See \cref{rem:error_splitting} and \cref{thm:em_n_fraction} for more detailed explanations.)}
Hence the structure factor interpolation scheme cannot generally improve the asymptotic scaling of the finite-size error. 
Another finite-size correction scheme is the twist averaging method, which has been used for QMC calculations \cite{LinZongCeperley2001,FoulkesMitasNeedsEtAl2001}, and also recently in MP2/CC calculations~\cite{GruberLiaoTsatsoulisEtAl2018, MihmMcIsaacShepherd2019}.
The twist averaging method calculates the average of the energies using a set of shifted $\vk$-point meshes of the same size, which can reduce the fluctuation as well as the magnitude of the finite-size error as $N_\vk\rightarrow\infty$.
In particular, after twist averaging, the finite-size error can decay more smoothly with respect to $N_{\vk}$, which improves the effectiveness of power-law extrapolation~\cite{LinZongCeperley2001, MihmMcIsaacShepherd2019}.

The slow convergence of the exchange energy that scales as $\Or(N_\vk^{-\frac13})$ for 3D periodic systems is due to the integrable singularity of the integrand $F_\text{X}$ in the Brillouin zone. The correction using the Madelung constant~\cite{FraserFoulkesRajagopalEtAl1996,DrummondNeedsSorouriEtAl2008,McClainSunChanEtAl2017} removes the leading contribution, and the finite-size error is observed to become $\Or(N_{\vk}^{-1})$. However, there has not been rigorous proof of this statement. 
The Madelung constant only depends on the geometry of the unit cell $\Omega$ and hence can be efficiently pre-computed~\cite{FraserFoulkesRajagopalEtAl1996,DaboKozinskySingh-MillerEtAl2008,MakovPayne1995}.
An alternative strategy is to choose a suitable auxiliary function to remove the leading singular term~\cite{GygiBaldereschi1986,CarrierRohraGorling2007}. 
We prove that the finite-size error of both correction techniques is $\Or(N_{\vk}^{-1})$, and hence they are equivalent up to the leading order of the error.

\vspace{1em}
\noindent\textbf{Paper Organization:}

\REV{
        \cref{sec:background} introduces the background information of the problem and  notations used in the paper. 
        \cref{sec:exchange} and \cref{sec:mp2} provide the finite-size error analysis for the Fock exchange and the MP2 energy calculations, respectively. 
        These three sections contain the main message of this paper for readers with a broad background. 
        \cref{sec:madelung} and \cref{sec:removable} then extend the finite-size error analysis for two correction schemes: the Madelung constant correction 
        and the staggered mesh method. 
        The main numerical analysis result that estimates the quadrature error of trapezoidal rules for a general class of non-smooth integrands in the fractional form \cref{eqn:fractional_form}
        is described in \cref{sec:em}.
        These sections may be skipped on a first reading.
}

\section{Background}\label{sec:background}
Unless otherwise stated, throughout the paper, the system is assumed to extend along all three dimensions. Let $\Omega$ be the unit cell, $\abs{\Omega}$ be its volume, and $\Omega^*$ be the associated BZ. 
Denote the Bravais lattice and its associated reciprocal lattice by $\mathbb{L}$  and $\mathbb{L}^*$, respectively. 
We use a uniform mesh $\mathcal{K}$ for $\vk$-point sampling in $\Omega^*$ (which may or may not include the $\Gamma$ point, i.e., the point of origin in $\Omega^*$; see \cref{fig:Kq} for an illustration), also referred to as an Monkhorst-Pack (MP) mesh, and denote $N_\vk$ as the number of $\vk$ points in the mesh.
For a mean-field calculation with $\mathcal{K}$, 
each molecular orbital (also called band orbital), characterized by the $\vk$-point and the band index $n$, is written as
\[
\psi_{n\vk}(\vr) = \dfrac{1}{\sqrt{N_\vk}} e^{\I \vk\cdot\vr} u_{n\vk}(\vr)=
\frac{1}{\abs{\Omega}\sqrt{N_{\vk}}} \sum_{\mathbf{G}\in\mathbb{L}^*} \hat{u}_{n\vk}(\mathbf{G}) e^{\I (\vk+\mathbf{G}) \cdot \mathbf{r}},
\]
and is associated with an orbital energy $\varepsilon_{n\vk}$.
The pair product is defined as 
\[
\varrho_{n'\vk',n\vk}(\vr)=\conj{u}_{n'\vk'}(\vr)  u_{n\vk}(\vr) =\frac{1}{\abs{\Omega}} \sum_{\mathbf{G}\in\mathbb{L}^*} \hat{\varrho}_{n'\vk',n\vk}(\mathbf{G}) e^{\I \mathbf{G} \cdot \mathbf{r}},
\]
and a two-electron repulsion integral (ERI) is then computed as 
\[
\braket{n_1\vk_1,n_2\vk_2|n_3\vk_3,n_4\vk_4}
=  \frac{1}{\abs{\Omega}N_\vk} \xsum_{\vG\in\mathbb{L}^*}
\frac{4\pi}{\abs{\vq+\vG}^2}  
\hat{\varrho}_{n_1\vk_1,n_3\vk_3}(\mathbf{G}) \hat{\varrho}_{n_2\vk_2,n_4\vk_4}(\vG_{\vk_1,\vk_2}^{\vk_3,\vk_4}-\mathbf{G}),
\]
where $\vq = \vk_3 - \vk_1$, $\vG_{\vk_1,\vk_2}^{\vk_3,\vk_4} = \vk_1+\vk_2-\vk_3-\vk_4$, and $\xsum_{\vG\in\mathbb{L}^*}$ excludes the possible term with $\vq+\vG=\bm{0}$. 
Such an ERI can be non-zero only when $\vG_{\vk_1,\vk_2}^{\vk_3,\vk_4} \in \mathbb{L}^*$, corresponding to crystal momentum conservation. 

Below, band indices $i,j$ ($a,b$) always refer to the occupied (virtual) bands, respectively.
All analysis is performed in the spin-restricted setting, and can be straightforwardly generalized when the spin degree of freedom is taken into account explicitly. 
As detailed in \cref{appendix}, the finite-size error in the kinetic energy, the Hartree energy, and the energy due to external potentials all decay super-algebraically if assuming all orbitals can be evaluated exactly at any $\vk$. 
In the following discussion, we will not consider the finite-size errors of these three types of energies. 
The exchange energy and the MP2  energy per unit cell are computed respectively as
\begin{align}
E_\text{X}(N_\vk) 
& = 
-\dfrac{1}{N_\vk} \sum_{ij}\sum_{\vk_i\vk_j\in\mathcal{K}} 
\braket{i\vk_i,j \vk_j |j \vk_j, i\vk_i},
\label{eqn:exchange_eri}
\\
E_\text{MP2}(N_\vk)
&=\frac{1}{N_\vk}
\sum_{ijab}
\sum_{\vk_i \vk_j \vk_a\in\mathcal{K}}
\dfrac{
\big(
2\braket{i\vk_i,j \vk_j |a\vk_a, b\vk_b} 
-
\braket{i\vk_i,j \vk_j |b\vk_b, a\vk_a} 
\big)
}{\varepsilon_{i\vk_i,j\vk_j}^{a\vk_a,b\vk_b}}
\braket{a\vk_a, b\vk_b|i\vk_i,j \vk_j},
\label{eqn:mp2_eri}
\end{align}
with $\varepsilon_{i\vk_i,j\vk_j}^{a\vk_a,b\vk_b} = \varepsilon_{i\vk_i} + \varepsilon_{j\vk_j} - \varepsilon_{a\vk_a} - \varepsilon_{b\vk_b}$. 
For each set of $(\vk_i, \vk_j,\vk_a)$ in $E_\text{MP2}(N_\vk)$, $\vk_b$ is the unique point in $\mathcal{K}$ satisfying $ \vG_{\vk_i,\vk_j}^{\vk_a,\vk_b}\in\mathbb{L}^*$.  
When $N_\vk$ goes to infinity and $\mathcal{K}$ converges to $\Omega^*$, these two energies converge to their exact values in the TDL, denoted by $E_\text{X}^\text{TDL}$ and $E_\text{MP2}^\text{TDL}$. 
In this paper, we adopt a uniform approach from numerical quadrature perspective to describe the asymptotic scaling of the finite-size errors in the two energy calculations, i.e., $E_*^\text{TDL} - E_*(N_\vk)$ v.s.\ $N_\vk$, where the subscript $*$ can be ``X'' or ``MP2''.

\subsection{Integral form of the energy in TDL}\label{sec:integral_form}
For each ERI in the energy calculations above, three momentum vectors are sampled over $\mathcal{K}$ (thus over $\Omega^*$ in the TDL) while the remaining one is determined by crystal momentum conservation. 
Note that such an ERI, say $\braket{n_1\vk_1,n_2\vk_2|n_3\vk_3,n_4\vk_4}$, is invariant  if we shift any $\vk$ to $\vk + \vG$ with any $\vG\in\mathbb{L}^*$. 
For each set of  $(\vk_1, \vk_2, \vk_3)$ with $\vq = \vk_3 - \vk_1$, we could shift $\vk_4$ by some $\vG$ vector so that $\vG_{\vk_1,\vk_2}^{\vk_3,\vk_4} = \bm{0}	$ or equivalently $\vk_4 = \vk_2 - \vq$. 
Then, a properly scaled ERI below can be treated as a function of $\vk_1, \vk_2,$ and $\vq$ as 
\begin{align}
N_\vk \braket{n_1\vk_1,n_2\vk_2|n_3\vk_3,n_4\vk_4}
&= 
\frac{4\pi}{\abs{\Omega}} \xsum_{\vG\in\mathbb{L}^*}
\frac{1}{\abs{\vq+\vG}^2}  
\hat{\varrho}_{n_1\vk_1,n_3(\vk_1 + \vq)}(\mathbf{G}) \hat{\varrho}_{n_2\vk_2,n_4(\vk_2 - \vq)}(-\mathbf{G})
\nonumber \\
& = R_{n_1n_2n_3n_4}(\vk_1, \vk_2, \vq),
\label{eqn:eri_scaled}
\end{align}
with band indices $n_1, n_2, n_3,$ and $n_4$ as parameters.
We also define the orbital energy fraction term in MP2 energy calculation  as a function of $\vk_1,\vk_2$ and $\vq$ as
\begin{equation}\label{eqn:fraction_energy}
\dfrac{1}{\varepsilon_{n_1\vk_1, n_2\vk_2}^{n_3\vk_3, n_4\vk_4}} = 
\dfrac{1}{\varepsilon_{n_1\vk_1, n_2\vk_2}^{n_3(\vk_1+\vq), n_4(\vk_2-\vq)}}
=E_{n_1n_2n_3n_4}(\vk_1, \vk_2, \vq).
\end{equation}
Using these two basic notations, the energy calculations in \cref{eqn:exchange_eri} and \cref{eqn:mp2_eri} could be reformulated as 
\begin{itemize}
        \item Exchange energy
        \begin{align}
        E_\text{X}(N_\vk) 
        & = 
        -\dfrac{1}{N_\vk^2}\sum_{\vk_i\vk_j\in\mathcal{K}} \left( \sum_{ij}F_\text{X}^{ij}(\vk_i, \vk_j)\right),
        \label{eqn:exchange_N}
        \\
        F_\text{X}^{ij}(\vk_i, \vk_j) 
        & = 
        R_{ijji}(\vk_i, \vk_j, \vk_j - \vk_i).
        \nonumber
        \end{align}
        
        \item MP2 energy 
        \begin{align}
        E_\text{MP2}(N_\vk)
        &=\frac{1}{N_\vk^3}
        \sum_{\vk_i \vk_j \vk_a\in\mathcal{K}} 
        \left(
        \sum_{ijab}F_\text{MP2,d}^{ijab}(\vk_i, \vk_j, \vk_a) + F_\text{MP2,x}^{ijab}(\vk_i, \vk_j, \vk_a)
        \right),
        \label{eqn:mp2_N}
        \\
        F_\text{MP2,d}^{ijab}(\vk_i, \vk_j, \vk_a)  
        & = 
        2R_{ijab}(\vk_i, \vk_j, \vk_a -\vk_i)R_{abij}(\vk_a, \vk_b, \vk_i -\vk_a)E_{ijab}(\vk_i,\vk_j,\vk_a-\vk_i),
        \nonumber\\
        F_\text{MP2,x}^{ijab}(\vk_i, \vk_j, \vk_a)  
        & = 
        - R_{ijba}(\vk_i, \vk_j, \vk_b - \vk_i)R_{abij}(\vk_a, \vk_b, \vk_i -\vk_a)E_{ijab}(\vk_i,\vk_j,\vk_a-\vk_i),
        \nonumber
        \end{align}
        with $\vk_b = \vk_i + \vk_j - \vk_a$ from $\braket{i\vk_i, j\vk_j|a\vk_a,b \vk_b}$. 
        The summations over $F_\text{MP2,d}^{ijab}$ and $F_\text{MP2,x}^{ijab}$ are referred to as the direct and exchange terms of the MP2 energy, respectively.
\end{itemize}

In the TDL, $\mathcal{K}$ converges to $\Omega^*$ and the summation $\frac{1}{N_\vk}\sum_{\vk \in \mathcal{K}}$ converges to the integral $\frac{1}{|\Omega^*|}\int_{\Omega^*}\ud\vk$. 
The two energies then converge to a double and a triple integrals over $\Omega^*$, respectively, as
\begin{align}
E_\text{X}^\text{TDL}
& = 
-\dfrac{1}{|\Omega^*|^2} \iint_{\Omega^*\times \Omega^*}\ud \vk_i \ud \vk_j 
\left(
\sum_{ij} F_\text{X}^{ij}(\vk_i, \vk_j)
\right),
\label{eqn:exchange_TDL}\\
E_\text{MP2}^\text{TDL}
& =
\dfrac{1}{|\Omega^*|^3}\iiint_{\Omega^*\times \Omega^*\times\Omega^*}\ud \vk_i \ud \vk_j \ud \vk_a
\left(
\sum_{ijab} F_\text{MP2,d}^{ijab}(\vk_i, \vk_j, \vk_a) + F_\text{MP2,x}^{ijab}(\vk_i, \vk_j, \vk_a) 
\right).
\label{eqn:mp2_TDL}
\end{align} 
By this formulation, due to the periodicity of the integrands with respect to each $\vk$ variable, numerical calculations of the exchange and the MP2 energies in \cref{eqn:exchange_N} and \cref{eqn:mp2_N} can be interpreted as applying a trapezoidal quadrature rule to approximate the corresponding integrals \cref{eqn:exchange_TDL} and \cref{eqn:mp2_TDL} using a uniform mesh $\mathcal{K}$ in $\Omega^*$. 
The finite-size errors can thus be decomposed into the error of the numerical quadrature and the error of  the integrand evaluation. 

In this paper, we focus on systems with a direct gap, i.e., $\varepsilon_{i\vk_i} + \varepsilon_{j\vk_j} - \varepsilon_{a\vk_a} - \varepsilon_{b\vk_b} \leqslant -2 \varepsilon_g < 0$ for all
$i,j,a,b,\vk_i,\vk_j,\vk_a,\vk_b$. 
\REV{
We assume that the mean-field orbital energies $\{\varepsilon_{n\vk}\}$ and orbitals $\{u_{n\vk}\}$ are exact for any $n$ and $\vk \in\Omega^*$, 
and that a finite number of virtual bands are used for the energy calculations in both the finite and the TDL cases. 
In addition, we assume that the $\varepsilon_{n\vk}$, $\psi_{n\vk}$, and $u_{n\vk}$ are smooth with respect to $\vk$ for any fixed band $n$ 
and thus $\hat{\varrho}_{n'\vk',n\vk}(\mathbf{G})$ is also smooth with respect to $\vk$, $\vk'$ for fixed any $n, n',$ and $\vG$. 
For systems free of topological obstructions~\cite{BrouderPanatiCalandraEtAl2007,MonacoPanatiPisanteEtAl2018}, these conditions can be replaced by weaker conditions using techniques based on Green's functions. We find that such a treatment introduces a considerable overhead to the presentation. Moreover, this issue is orthogonal to the study of the quadrature error below.  Therefore we adopt the assumptions stated above, and postpone a complete treatment of the problem without assuming the smoothness of $\varepsilon_{n\vk},\psi_{n\vk}$ to a future work.
}
With these assumptions, the numerical evaluation of all the integrands in \eqref{eqn:exchange_N} and \eqref{eqn:mp2_N} is exact so that we could focus on the quadrature error only.

We use the term ``trapezoidal rule'' to refer to a general class of quadrature rules over a hypercube that has equal quadrature weights and has quadrature nodes on a uniform mesh. 
Specifically, for a general function $g$ over a hypercube $V$, a trapezoidal rule with a uniform mesh $\mathcal{X}$ is denoted as 
\[
\mathcal{Q}_{V}(g, \mathcal{X}) =  \dfrac{|V|}{|\mathcal{X}|}\sum_{\vx_i \in \mathcal{X}}g(\vx_i),
\]
and its quadrature error is denoted as
\[
\mathcal{E}_{V}(g, \mathcal{X})
= 
\mathcal{I}_{V}(g)
- 
\mathcal{Q}_{V}(g, \mathcal{X})
= \int_V \ud\vx g(\vx) - \dfrac{|V|}{|\mathcal{X}|}\sum_{\vx_i \in \mathcal{X}}g(\vx_i),
\]
where $\mathcal{I}$ is the integral operator.

With a finite MP mesh $\mathcal{K}$, the finite-size error problem now reduces to describing the asymptotic scalings of the
quadrature errors below with respect to $N_\vk$, 
\begin{align*}
& E_\text{X}^\text{TDL} - E_\text{X}(N_\vk)
= -\dfrac{1}{|\Omega^*|^2}\mathcal{E}_{\Omega^*\times \Omega^*}\left(\sum_{ij} F_\text{X}^{ij}(\vk_i, \vk_j), \mathcal{K}\times\mathcal{K}\right),
\\
& E_\text{MP2}^\text{TDL} - E_\text{MP2}(N_\vk)
=
\dfrac{1}{|\Omega^*|^3}
\mathcal{E}_{(\Omega^*)^{\times 3}}\left(\sum_{ijab} F_\text{MP2,d}^{ijab}(\vk_i, \vk_j, \vk_a) + F_\text{MP2,x}^{ijab}(\vk_i, \vk_j, \vk_a), \mathcal{K}^{\times 3}\right).
\end{align*}

\subsection{Basic properties of the integrands}
The convergence rate of a trapezoidal rule generally depends on the smoothness of the integrand and its behavior at the boundary. 
All integrands in the energy calculations \cref{eqn:exchange_N} and \cref{eqn:mp2_N} are built upon basic functions $R_{n_1n_2n_3n_4}(\vk_1,\vk_2,\vq)$  and $E_{n_1n_2n_3n_4}(\vk_1,\vk_2,\vq)$ in \cref{eqn:eri_scaled} and \cref{eqn:fraction_energy}.

We first note that $E_{n_1n_2n_3n_4}(\vk_1,\vk_2,\vq)$ is periodic with respect to $\vk_1$, $\vk_2$ and $\vq$ over $\Omega^*$ due to the fact that orbital energy $\varepsilon_{n\vk}$ with a fixed band $n$ is periodic with $\vk$. 
The MP2 energy calculation only involves this function with $n_1, n_2$ being occupied orbitals and $n_3, n_4$ being virtual orbitals. 
In this case, $E_{n_1n_2n_3n_4}(\vk_1,\vk_2,\vq)$ is negative and smooth with respect to $\vk_1$, $\vk_2$, and $\vq$ based on our assumption that  $\varepsilon_{n\vk}$ is smooth with respect to $\vk$, and the system has a positive gap.

Since $\psi_{n\vk}(\vr) = \psi_{n(\vk+ \vG')}(\vr)$ with any $\vG' \in \mathbb{L}^*$, we have $u_{n(\vk+\vG')}(\vr) = e^{-\I\vG'\cdot\vr}u_{n\vk}(\vr)$. 
Then the multiplication of the two pair products in \cref{eqn:eri_scaled} of $R_{n_1n_2n_3n_4}(\vk_1,\vk_2,\vq)$ can be written as a function of $\vk_1,\vk_2, \vq+\vG$ as
\begin{align*}
        \frac{4\pi}{|\Omega|}\hat{\varrho}_{n_1\vk_1,n_3(\vk_1 + \vq)}(\mathbf{G}) \hat{\varrho}_{n_2\vk_2,n_4(\vk_2 - \vq)}(-\mathbf{G}) 
        &  = 
        \frac{4\pi}{|\Omega|}\hat{\varrho}_{n_1\vk_1,n_3(\vk_1 + \vq + \vG)}(\bm{0}) \hat{\varrho}_{n_2\vk_2,n_4(\vk_2 - \vq - \vG)}(\bm{0}) 
        \nonumber 
        \\
        &= r_{n_1n_2n_3n_4}(\vk_1, \vk_2, \vq + \vG),
        \label{eqn:r_function}
\end{align*}
where, by its definition, $r_{n_1n_2n_3n_4}(\vk_1, \vk_2, \vq)$ is smooth with respect to $\vk_1, \vk_2, \vq$ and periodic with respect to $\vk_1, \vk_2$ over $\Omega^*$.
Using this new notation, $R_{n_1n_2n_3n_4}(\vk_1, \vk_2, \vq)$ can be written in a more concise form as
\begin{equation}\label{eqn:R_function}
        R_{n_1n_2n_3n_4}(\vk_1, \vk_2, \vq)
        = \sum_{\vG\in\mathbb{L}^*}\dfrac{r_{n_1n_2n_3n_4}(\vk_1, \vk_2, \vq+\vG)}{|\vq + \vG|^2}.
\end{equation}
In this continuous formulation, $\xsum_\vG$ is replaced by regular summation $\sum_\vG$. When $\vq = -\vG$, the summation term associated with $\vq+\vG$ is indeterminate and is set to $0$ in the numerical evaluation of the function. 

By the orthonormality of the orbitals, i.e., $\hat{\varrho}_{n\vk,n'\vk}(\bm{0}) = \delta_{nn'}$, we can expand 
$r_{n_1n_2n_3n_4}$ near $\vq = \bm{0}$ as
\begin{equation}\label{eqn:expanion_rijab}
        r_{n_1n_2n_3n_4}(\vk_1, \vk_2, \vq) =
        \dfrac{4\pi}{|\Omega|}\delta_{n_1n_3}\delta_{n_2n_4} + (\delta_{n_1n_3} + \delta_{n_2n_4})\vv^T \vq + \Or(|\vq|^2).
\end{equation}
Therefore $R_{n_1n_2n_3n_4}$ is periodic with respect to $\vk_1, \vk_2,$ and $\vq$ over $\Omega^*$ and is smooth everywhere except at $\vq \in \mathbb{L}^*$. 
This non-smoothness comes from the summation term with $\vq + \vG = \bm{0}$, i.e., the singularity of the Coulomb kernel in the reciprocal space.

Combining the above discussions over $R_{n_1n_2n_3n_4}$ and $E_{n_1n_2n_3n_4}$ with the definitions of integrands in \cref{eqn:exchange_N} and \cref{eqn:mp2_N}, we obtain some basic properties of the three integrands as
\begin{itemize}
        \item $F_\text{X}^{ij}(\vk_i, \vk_j)$ is periodic with respect to $\vk_i,\vk_j$ in $\Omega^*$ and smooth everywhere except at $\vk_j - \vk_i \in \mathbb{L}^*$. 
        \item $F_\text{MP2,d}^{ijab}(\vk_i, \vk_j,\vk_a)$ is periodic  with respect to  $\vk_i,\vk_j, \vk_a$ in $\Omega^*$ and smooth everywhere except at $\vk_a - \vk_i \in \mathbb{L}^*$. 
        \item $F_\text{MP2,x}^{ijab}(\vk_i, \vk_j,\vk_a)$ is periodic  with respect to  $\vk_i,\vk_j, \vk_a$ in $\Omega^*$ and smooth everywhere except at $\vk_j - \vk_a \in \mathbb{L}^*$ or $\vk_i - \vk_a \in \mathbb{L}^*$. 
\end{itemize}

\subsection{Standard Euler-Maclaurin formula}\label{sec:standard_em_quad}
Consider a hypercube $V\subset \mathbb{R}^d$ of edge length $L$ and an $m^d$-sized uniform mesh $\mathcal{X}$ in $V$. 
Here, $m$ denotes the number of \REV{subintervals} along each dimension in $\mathcal{X}$. 
For a generic function $g$ that has continuous derivatives up to $l$-th order, its quadrature error can be explicitly described by the standard Euler-Maclaurin formula (see \Cref{thm:em_general} for the full description) as 
\begin{equation}\label{eqn:EM_general}
\mathcal{E}_V(g, \mathcal{X}) = \sum_{s=1}^{l-1} \dfrac{L^s}{m^s}\sum_{|\vbeta|=s} c_\vbeta \int_V g^{(\vbeta)}(\vx)\ud\vx + \Or(m^{-l}),
\end{equation}
where $\vbeta$ is a $d$-dimensional multi-index with $|\vbeta| =  \sum_i \beta_i$, $g^{(\vbeta)}(\vx)$ is the derivative of $g$ of order $\vbeta$, and $c_\vbeta$ is some constant.
When $g(\vx)$ and its derivatives up to $(l-2)$th order satisfy the periodic boundary condition on $\partial V$, all the integrals of $g^{(\vbeta)}(\vx)$ above vanish and the quadrature error scales as $\Or(m^{-l})$. 
Further, if $g(\vx)$ is also smooth (i.e., $l = \infty$) and all its derivatives satisfy the periodic boundary condition, the quadrature error decays super-algebraically, i.e., faster than $m^{-l}$ with any $l > 0$. 
These statements are summarized in \cref{cor:euler_maclaurin}. 
Note that smooth functions that are periodic with $V$ or compactly supported in $V$ (i.e., the function support is a subset of $V$ and separated from $\partial V$) satisfy the latter condition and have super-algebraically decaying 
quadrature error. 

As shown earlier, all integrands in the exchange and MP2 energy calculations are periodic but discontinuous at certain points.
The standard Euler-Maclaurin formula \cref{eqn:EM_general} cannot be directly applied, but is still a major tool used to analyze the quadrature errors of these non-smooth integrands in this paper. 
First, it turns out that these integrands can all be properly split into some smooth and non-smooth terms, where the smooth terms have super-algebraically decaying quadrature error by \cref{cor:euler_maclaurin} and the non-smooth terms all belong to a special class of functions in fractional form \cref{eqn:fractional_form}. 
The problem is thus simplified to analyzing the dominant quadrature error caused by such non-smooth fractional-form terms.
This idea of non-smoothness extraction is detailed in \Cref{subsec:extraction}. 
Second, a trapezoidal rule with $m^d$-sized uniform mesh is equivalent to uniformly partitioning the integration domain into $m^d$ subdomains and then applying a single-point quadrature rule to each subdomain. 
The standard Euler-Maclaurin formula can \REV{only} be applied separately to the quadrature in each subdomain where the integrand is smooth.
\REV{
Non-smooth terms with discontinuous points in certain subdomains need to be treated separately, and this gives a \textit{partial Euler-Maclaurin formula}.}
This \REV{is the key idea} for analyzing the quadrature error of these special fractional-form terms above and is detailed in \Cref{sec:em}.

\section{Quadrature error of Fock exchange energy}\label{sec:exchange}
Analyzing  the quadrature error
\[
\mathcal{E}_{\Omega^*\times \Omega^*}\left(\sum_{ij} F_\text{X}^{ij}(\vk_i, \vk_j), \mathcal{K}\times\mathcal{K}\right) = \sum_{ij}\mathcal{E}_{\Omega^*\times \Omega^*}( F_\text{X}^{ij}(\vk_i, \vk_j), \mathcal{K}\times\mathcal{K})
\]
in the exchange energy calculation is a classical numerical analysis problem. 
In the following analysis, we use the notation ``$f \lesssim g$'' between two $N_\vk$-dependent quantities $f$ and $g$ to mean that there exists a constant $C$ independent of $N_\vk$ such that $f \leqslant Cg$ for sufficiently large $N_\vk$. 
For each pair of band indices $i,j$, the quadrature error of $F_\text{X}^{ij}(\vk_i, \vk_j)$ can be estimated by the following three steps.

\subsection{Change of variables}
Note that $F_\text{X}^{ij}(\vk_i, \vk_j) = R_{ijji}(\vk_i, \vk_j, \vk_j - \vk_i)$ is periodic over $\Omega^*$ and is discontinuous at $\vk_j - \vk_i \in \mathbb{L}^*$. 
To isolate the discontinuity to one variable for later analysis, we define $\vq = \vk_j - \vk_i$, corresponding to the change of variable $\vk_j \rightarrow \vk_i + \vq$, and define 
\[
\wt{F}_\text{X}^{ij}(\vk_i, \vq) = F_\text{X}^{ij}(\vk_i, \vk_i + \vq) = R_{ijji}(\vk_i, \vk_i + \vq, \vq). 
\]
which is periodic with respect to $\vq$ over $\Omega^*$ as well.
Based on the periodicity of $F_\text{X}^{ij}(\vk_i, \vk_j) $, we have 
\[
\int_{\Omega^*}\ud\vk_i\int_{\Omega^*}\ud\vk_j F_\text{X}^{ij}(\vk_i, \vk_j) = 
\int_{\Omega^*}\ud\vk_i\int_{\Omega^* + \vk_i}\ud\vk_j F_\text{X}^{ij}(\vk_i, \vk_j) = 
\int_{\Omega^*}\ud\vk_i\int_{\Omega^*}\ud\vq F_\text{X}^{ij}(\vk_i, \vk_i + \vq),
\]
where the second equality applies $\vk_j \rightarrow\vk_i + \vq$. 
The same change of variable converts the trapezoidal rule for $F_\text{X}^{ij}(\vk_i, \vk_j)$ to
\[
\dfrac{|\Omega^*|^2}{N_\vk^2}\sum_{\vk_i\vk_j \in \mathcal{K}}F_\text{X}^{ij}(\vk_i, \vk_j) = 
\dfrac{|\Omega^*|^2}{N_\vk^2}\sum_{\vk_i \in \mathcal{K}}\sum_{\vq \in \mathcal{K} - \vk_i}F_\text{X}^{ij}(\vk_i, \vk_i + \vq) = 
\dfrac{|\Omega^*|^2}{N_\vk^2}\sum_{\vk_i \in \mathcal{K}}\sum_{\vq \in \mathcal{K}_\vq}F_\text{X}^{ij}(\vk_i,\vk_i + \vq) ,
\]
where the second equality changes $\vq$ to its minimum image in $\Omega^*$ by periodicity of the integrand, and $\mathcal{K}_\vq$ is an MP mesh that contains all the minimum images of $\vk_j - \vk_i$ in $\Omega^*$ with $\vk_i,\vk_j \in \mathcal{K}$. 
The new mesh $\mathcal{K}_\vq$ is of the same size as $\mathcal{K}$ and contains the $\Gamma$ point, i.e., $\vq = \bm{0}$. 
As illustrated in \cref{fig:Kq}, we note that $\mathcal{K}$ could be an arbitrary MP mesh in $\Omega^*$ in practical calculations while the induced $\mathcal{K}_\vq$ from the above change of variable is always $\Gamma$-centered. 
\begin{figure}[htbp]
        \centering
        \includegraphics[width=0.8\textwidth]{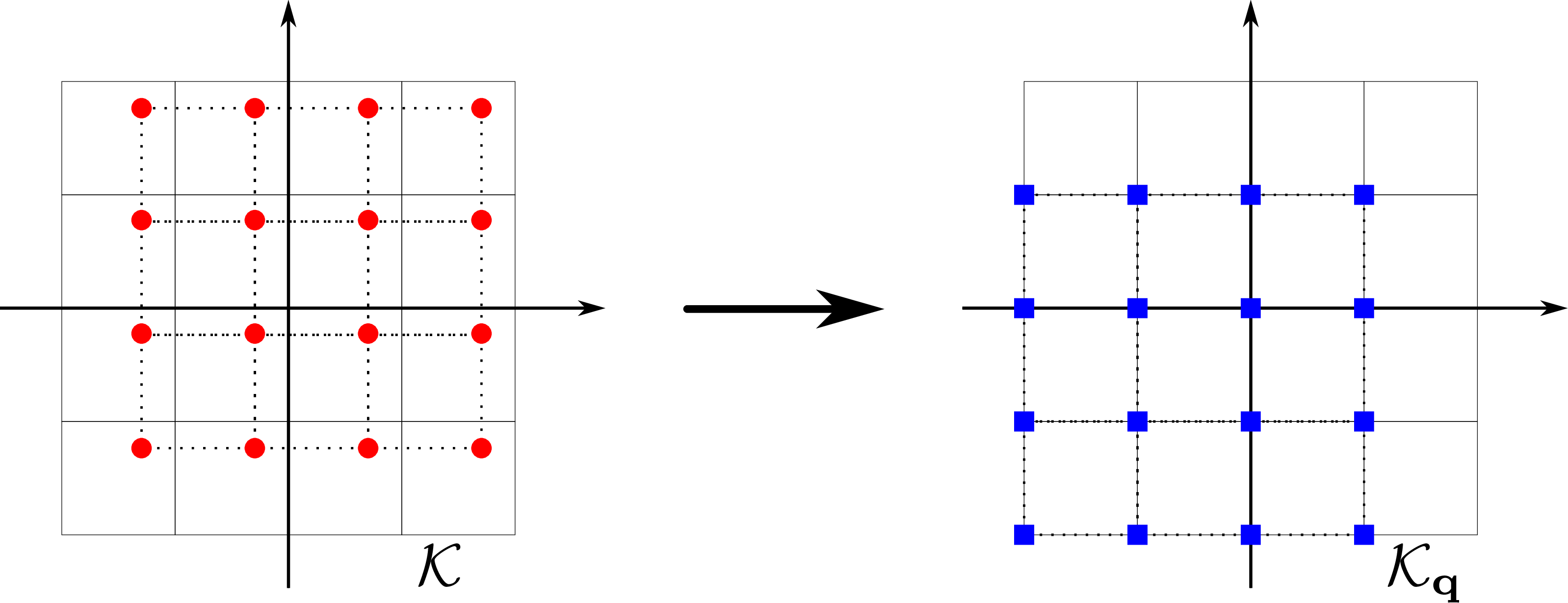}
        \caption{2D illustration 
                of MP mesh $\mathcal{K}$ and $\mathcal{K}_\vq$. The MP mesh $\mathcal{K}_\vq$ induced by any $\mathcal{K}$ always contains $\vq = \bm{0}$. 
        }
        \label{fig:Kq}
\end{figure}

From these two equations above, $\wt{F}_\text{X}^{ij}(\vk_i, \vq)$ satisfies 
\begin{equation}\label{eqn:exchange_changeVar}
\mathcal{E}_{\Omega^*\times \Omega^*}(F_\text{X}^{ij}(\vk_i, \vk_j), \mathcal{K}\times \mathcal{K})  
= \mathcal{E}_{\Omega^*\times \Omega^*}(\wt{F}_\text{X}^{ij}(\vk_i, \vq), \mathcal{K}\times \mathcal{K}_\vq). 
\end{equation}
It is thus equivalent to study the quadrature error for $\wt{F}_\text{X}^{ij}$ with uniform mesh $\mathcal{K}\times \mathcal{K}_\vq$ in $\Omega^*\times \Omega^*$. 
As can be noted, $\wt{F}_\text{X}^{ij}(\vk_i, \vq)$ is periodic with respect to $\vk_i, \vq$ over $\Omega^*$ and is smooth everywhere except at $\vq \in \mathbb{L}^*$.

\subsection{Extraction of non-smoothness}\label{subsec:extraction}
In the integration domain $\Omega^*\times\Omega^*$,  $\wt{F}_\text{X}^{ij}(\vk_i, \vq)$ is only non-smooth at $\vq = \bm{0}$, due to the second term in the following splitting, 
\[
\wt{F}_\text{X}^{ij}(\vk_i, \vq)
  = 
 \left(
\sum_{\bm{0}\neq \vG\in\mathbb{L}^*}\dfrac{r_{ijji}(\vk_i, \vk_i+\vq, \vq+\vG)}{|\vq + \vG|^2}
\right)
+ 
\dfrac{r_{ijji}(\vk_i, \vk_i+\vq, \vq)}{|\vq|^2}. 
\]
To extract this non-smooth term, consider a \textit{localizer} $H(\vq)$ that is smooth, radial, and compactly supported in $\Omega^*$ (more precisely, the support of $H(\vq)$ is in $\Omega^*$ and separated from $\partial\Omega^*$) and equals identity in an open domain containing $\vq = \bm{0}$. 
A simple example for $H(\vq)$ with $\Omega^* = [-\frac12, \frac12]^3$ is 
\begin{equation}\label{eqn:localizer}
        H(\vq) = 
        \left\{
        \begin{array}{ll}
                1 & |\vq| \leqslant  0.1 
                \\
                \dfrac{e^{-\frac{1}{0.4-x}}}{e^{-\frac{1}{x-0.1}} + e^{-\frac{1}{0.4-x}}}       & 0.1 < |\vq| < 0.4
                \\
                0 & |\vq| \geqslant 0.4
        \end{array}
        \right.
        ,
\end{equation}
where $0.1$ and $0.4$ are arbitrarily chosen, and $H(\vq)$ against $|\vq|$ is plotted in \cref{fig:Hq}.
\begin{figure}[htbp]
        \centering
        \includegraphics[width=0.4\textwidth]{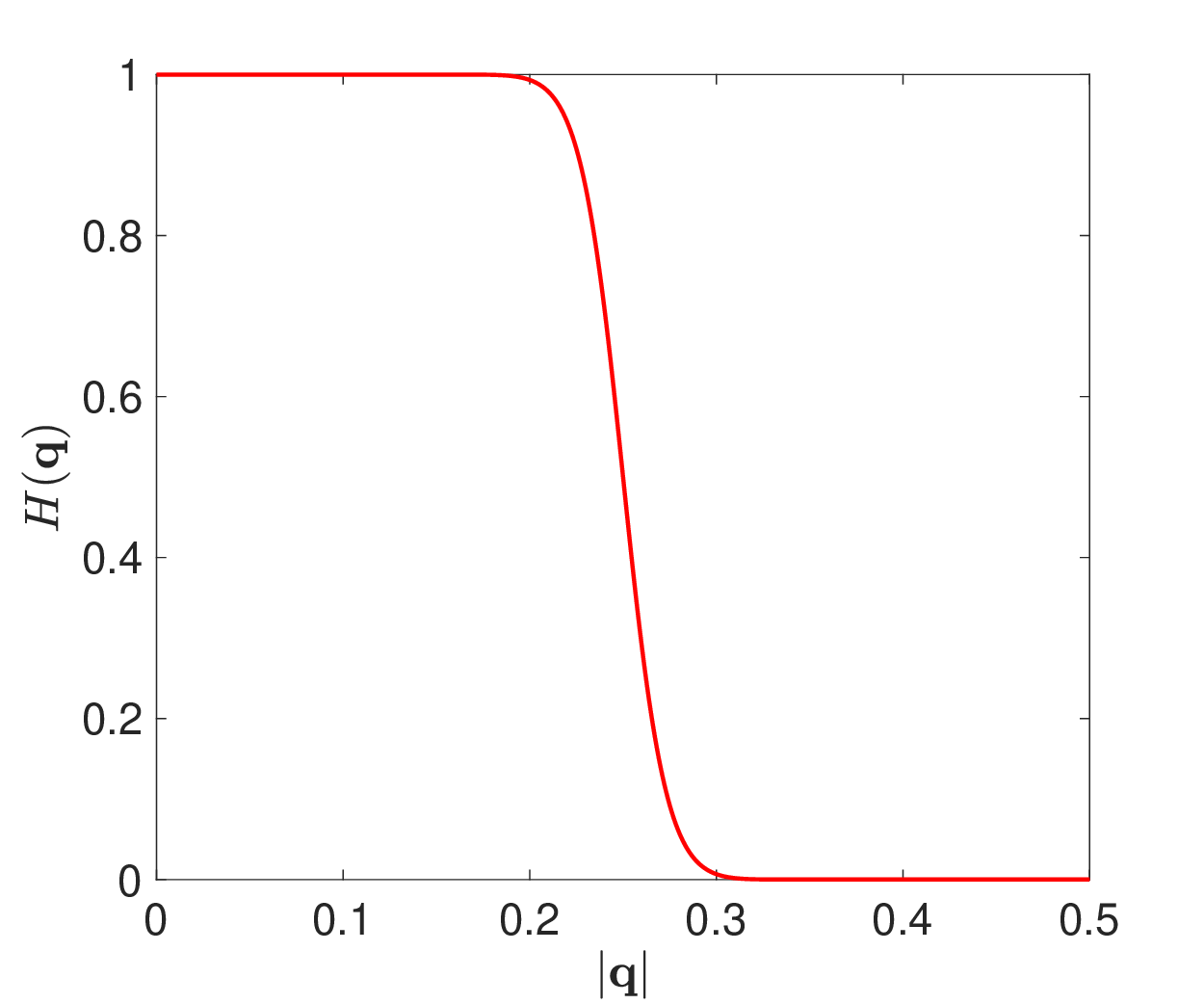}
        \caption{A radial function example of localizer $H(\vq)$ v.s.\ $|\vq|$.}
        \label{fig:Hq}
\end{figure}

Then define
\begin{equation}\label{eqn:local_F}
G^{ij}(\vk_i, \vq) = \dfrac{r_{ijji}(\vk_i, \vk_i+\vq, \vq)}{|\vq|^2} H(\vq),
\end{equation}
which is compactly supported with respect to $\vq$ in $\Omega^*$, periodic with respect to $\vk_i$, and smooth everywhere except at $\vq = \bm{0}$. 
When restricting $\vq$ to $\Omega^*$, $G^{ij}(\vk_i, \vq)$ equals to the non-smooth part of $\wt{F}_\text{X}^{ij}(\vk_i, \vq)$ in the neighborhood of $\vq = \bm{0}$, 
and thus $\wt{F}_\text{X}^{ij} - G^{ij}$ is smooth with respect to $\vq \in \Omega^*$. 
Meanwhile, due to the compactness of $G^{ij}$, $\wt{F}_\text{X}^{ij} - G^{ij}$ and all of its derivatives also satisfy the periodic boundary condition with respect to $\vq$ on $\partial \Omega^*$ (for brevity, we also call it  a periodic function over $\Omega^*$).
Since $\wt{F}_\text{X}^{ij} - G^{ij}$ is smooth and periodic with respect to both $\vk_i, \vq$ in $\Omega^*$, the quadrature error for $\wt{F}_\text{X}^{ij}$ can be split and estimated as 
\begin{align}
\mathcal{E}_{\Omega^*\times \Omega^*}\left(\wt{F}^{ij}_\text{X}, \mathcal{K}\times \mathcal{K}_\vq\right)
& =
\mathcal{E}_{\Omega^*\times \Omega^*}\left(\wt{F}^{ij}_\text{X} - G^{ij}, \mathcal{K}\times \mathcal{K}_\vq\right)
+
\mathcal{E}_{\Omega^*\times \Omega^*}\left(G^{ij}, \mathcal{K}\times \mathcal{K}_\vq\right)
\label{eqn:Fij_splitting}\\
& \lesssim 
\mathcal{E}_{\Omega^*\times \Omega^*}\left(G^{ij}, \mathcal{K}\times \mathcal{K}_\vq\right), 
\nonumber
\end{align}
where the second estimate uses the fact that the quadrature error for $\wt{F}^{ij}_\text{X} - G^{ij}$ decays super-algebraically according to the standard Euler-Maclaurin formula (see \cref{cor:euler_maclaurin}). 
The overall  error is thus dominated by the quadrature error for the extracted non-smooth term $G^{ij}(\vk_i, \vq)$. 

\subsection{Quadrature error for $G^{ij}$}
The trapezoidal quadrature rule and its error  for $G^{ij}(\vk_i, \vq)$ over the two variables $\vk_i, \vq\in\Omega^*$ can be further split into two parts as, 
\begin{align*}
\mathcal{E}_{\Omega^*\times \Omega^*}\left(G^{ij}, \mathcal{K}\times \mathcal{K}_\vq\right)
& = 
\left(
\mathcal{I}_{\Omega^*\times \Omega^*}(G^{ij})
-
\dfrac{|\Omega^*|}{N_\vk}\sum_{\vk_i \in \mathcal{K}}\int_{\Omega^*}\ud \vq G^{ij}(\vk_i, \vq) 
\right)
+ 
\\
& \hspace{8em}
\dfrac{|\Omega^*|}{N_\vk}\sum_{\vk_i \in \mathcal{K}}
\left(
\int_{\Omega^*}\ud \vq G^{ij}(\vk_i, \vq)
-
\dfrac{|\Omega^*|}{N_\vk}\sum_{\vq \in \mathcal{K}_\vq} G^{ij}(\vk_i, \vq) 
\right)
\\
& = 
\mathcal{E}_{\Omega^*}\left( \int_{\Omega^*}\ud \vq G^{ij}(\cdot , \vq), \mathcal{K}\right)
+ 
\dfrac{|\Omega^*|}{N_\vk}\sum_{\vk_i \in \mathcal{K}}
\mathcal{E}_{\Omega^*}\left( G^{ij}(\vk_i, \cdot), \mathcal{K}_\vq\right). 
\end{align*}
Since $r_{ijji}(\vk_i, \vk_i+\vq, \vq)$ in $G^{ij}(\vk_i, \vq)$ in \cref{eqn:local_F}  is smooth and periodic with respect to $\vk_i$, it can be proved that the partial integral $\int_{\Omega^*}\ud \vq G^{ij}(\vk_i , \vq)$ is a smooth, periodic function of $\vk_i$ using the dominated convergence theorem. 
The first part of the quadrature error thus also decays super-algebraically, and we have 
\begin{equation}\label{eqn:exchange_partial_q}
\mathcal{E}_{\Omega^*\times \Omega^*}\left(
G^{ij}, \mathcal{K}\times \mathcal{K}_\vq\right) 
\lesssim 
\dfrac{|\Omega^*|}{N_\vk}\sum_{\vk_i \in \mathcal{K}}
\mathcal{E}_{\Omega^*}\left( G^{ij}(\vk_i, \cdot), \mathcal{K}_\vq\right)
\lesssim
\max_{\vk_i\in \Omega^*}\mathcal{E}_{\Omega^*}\left( G^{ij}(\vk_i, \cdot), \mathcal{K}_\vq\right). 
\end{equation}

We firsts check the quadrature error $\mathcal{E}_{\Omega^*}\left( G^{ij}(\vk_i, \cdot), \mathcal{K}_\vq\right)$ with any given $\vk_i$. 
Fixing $\vk_i$ and restricting $\vq$ in $\Omega^*$, $G^{ij}(\vk_i, \vq)$ is compactly supported in $\Omega^*$ and in the fractional form \cref{eqn:fractional_form}.
Due to the denominator $|\vq|^2$, $G^{ij}(\vk_i, \vq)$ has an isolated point of discontinuity at $\vq = \bm{0}$, and the standard Euler-Maclaurin formula cannot be applied. 
Instead, \cref{thm:em_n_fraction} provides a special Euler-Maclaurin formula for functions in such a fractional form. 
\cref{cor:em_n_fraction} further describes the quadrature error when the integrand and its derivatives also satisfy the periodic boundary condition.

For brevity, we always assume $\Omega^* = [-\frac12, \frac12]^3$ in the following discussion. 
This assumption can be lifted in the general case by mapping $\Omega^*$ and all related variables to $[-\frac12, \frac12]^3$ using an affine transformation,  changing the denominator $|\vq|^2$ in $\wt{F}^{ij}_\text{X}$ and $G^{ij}$ to $\vq^T M \vq$ with a symmetric positive definite matrix $M$. 
Note that the cubic symmetry of $[-\frac12, \frac12]^3$ plays an additional role in removable discontinuity of $F_\text{X}$ in \Cref{sec:removable} but is not exploited in the following general error analysis of the exchange and MP2 energy calculations. 

By the expansion of $r_{ijji}(\vk_i, \vk_i+\vq,\vq) = \frac{4\pi}{|\Omega|}\hat{\varrho}_{i\vk_i,j(\vk_i+\vq)}(\bm{0})\hat{\varrho}_{j(\vk_i+\vq),i\vk_i}(\bm{0})$  in \cref{eqn:expanion_rijab}, we have 
\[
H(\vq) r_{ijji}(\vk_i, \vk_i+\vq, \vq) = \delta_{ij} + \delta_{ij}\vv^T_{i\vk_i}\vq + \Or(|\vq|^2).
\]
The integrand $G^{ij}(\vk_i, \vq)$ thus fits \cref{thm:em_n_fraction} with $n=1$ and $\gamma_\text{min} =  -2$ when $i = j$ and with $n=1$ and $\gamma_\text{min} = 0$ when $i\neq j$, and is also compactly supported in $\Omega^*$.
Thus, \cref{cor:em_n_fraction} shows that the quadrature error of $G^{ij}(\vk_i, \vq)$ over $\vq \in \Omega^*$ for any fixed $\vk_i$ scales as 
\begin{equation}\label{eqn:error_single_k}
\mathcal{E}_{\Omega^*}\left(G^{ij}(\vk_i, \cdot), \mathcal{K}_\vq\right)
= 
\left\{
\begin{array}{ll}
\Or(N_\vk^{-\frac13}) &  i=j\\
\Or(N_\vk^{-1}) & i\neq j
\end{array}
\right. ,
\end{equation}
where $N_\vk = m^3$ for 3D periodic systems. 
According to \cref{thm:em_n_fraction}, the prefactor of $\Or(N_\vk^{-\frac13})$ above can be controlled by  the upper bounds of  $|H(\vq)r_{ijji}(\vk_i,\vk_i+\vq,\vq)|$ and $|\frac{\partial^{|\valpha|}}{\partial\vq^{\valpha}} H(\vq)r_{ijji}(\vk_i,\vk_i+\vq,\vq)|$ with $|\valpha| = 1 $ for $\vq \in \Omega^*$, and similarly for the prefactor of $\Or(N_\vk^{-1})$. 
Since the numerator $H(\vq)r_{ijji}(\vk_i,\vk_i+\vq,\vq)$ is smooth with $\vk_i$ and $\vq$, its function values and derivatives with respect to $\vq$ have a uniform $\Or(1)$ upperbound that is independent of $\vk_i$. 
Thus, the prefactors of the asymptotic scalings in \cref{eqn:error_single_k} for any fixed $\vk_i$ can be independent of $\vk_i$, and we obtain
\begin{equation}\label{eqn:error_max_k}
\max_{\vk_i\in\Omega^*} \mathcal{E}_{\Omega^*}\left( G^{ij}(\vk_i, \cdot), \mathcal{K}_\vq\right) 
= 
\left\{
\begin{array}{ll}
\Or(N_\vk^{-\frac13}) &  i=j\\
\Or(N_\vk^{-1}) & i\neq j
\end{array}
\right. .
\end{equation}

Combining all the analysis above, \cref{thm:exchange_error} concludes that the quadrature error in the exchange energy calculation scales as $\Or(N_\vk^{-\frac13})$. 
This $\Or(N_\vk^{-\frac13})$ finite-size error  is well known in quantum chemistry but is mostly explained by physical intuitions.
To our best knowledge, \cref{thm:exchange_error} gives the first rigorous proof of this error scaling.
  
\begin{thm}[Fock exchange energy for 3D periodic systems]
\label{thm:exchange_error}
The finite-size error in the exchange energy calculation satisfies 
\[
E_\textup{x}^\textup{TDL} - E_\textup{x}(N_\vk)
= -\dfrac{1}{|\Omega^*|^2}\sum_{ij}\mathcal{E}_{\Omega^*\times \Omega^*}\left( F_\textup{x}^{ij}(\vk_i, \vk_j), \mathcal{K}\times\mathcal{K}\right)
= \Or(N_\vk^{-\frac13}).
\]
\end{thm}

\begin{proof}
        According to the change of variable in \cref{eqn:exchange_changeVar}, the extraction of non-smoothness in \cref{eqn:Fij_splitting}, and the quadrature error estimate for the extracted non-smooth term in \cref{eqn:exchange_partial_q} and \cref{eqn:error_max_k}, we can estimate the overall quadrature error in the exchange energy calculation as 
        \begin{align*}
        \sum_{ij}\mathcal{E}_{\Omega^*\times \Omega^*}\left(F^{ij}_\text{X}(\vk_i, \vk_j), \mathcal{K}\times \mathcal{K}\right)
        & = 
        \sum_{ij}\mathcal{E}_{\Omega^*\times \Omega^*}\left(\wt{F}^{ij}_\text{X}(\vk_i, \vq), \mathcal{K}\times \mathcal{K}_\vq\right)
        \\
        & \lesssim 
        \sum_{ij}\mathcal{E}_{\Omega^*\times \Omega^*}\left(G^{ij}(\vk_i, \vq), \mathcal{K}\times \mathcal{K}_\vq\right)
        \\
        & \lesssim 
        \sum_{ij}
        \max_{\vk_i \in \Omega^*}
        \mathcal{E}_{\Omega^*}\left( G^{ij}(\vk_i, \cdot), \mathcal{K}_\vq\right)
        \\
        & = \Or(N_\vk^{-\frac13}).
        \end{align*}

\end{proof}

\section{Quadrature error of MP2 energy}\label{sec:mp2}
The  quadrature error in MP2 energy calculation is split into the direct and the exchange terms, associated with integrands 
$F_\text{MP2,d}^{ijab}(\vk_i, \vk_j, \vk_a)$ and $F_\text{MP2,x}^{ijab}(\vk_i, \vk_j, \vk_a)$, as
\[
E_\text{MP2}^\text{TDL} - E_\text{MP2}(N_\vk)
=
\dfrac{1}{|\Omega^*|^3}
\mathcal{E}_{(\Omega^*)^{\times 3}}\left(\sum_{ijab} F_\text{MP2,d}^{ijab}(\vk_i, \vk_j, \vk_a) + F_\text{MP2,x}^{ijab}(\vk_i, \vk_j, \vk_a), (\mathcal{K})^{\times 3}\right). 
\]
Using the same methods for exchange energy calculation above, we analyze the quadrature errors of the two integrands, separately, for each set of band indices $i,j,a,b$.
Recall that we assume using a fixed number of virtual orbitals for MP2 energy calculations in both the finite and the TDL cases. 

\subsection{Quadrature error of the MP2 direct term}
Consider the change of variable $\vk_a \rightarrow \vk_i + \vq$ and define 
\begin{align*}
\wt{F}_\text{MP2,d}^{ijab}(\vk_i, \vk_j, \vq) 
& = F_\text{MP2,d}^{ijab}(\vk_i, \vk_j, \vk_i + \vq)
\\
& = 
2R_{ijab}(\vk_i, \vk_j, \vq)R_{abij}(\vk_i+\vq, \vk_j-\vq, -\vq)E_{ijab}(\vk_i,\vk_j,\vq). 
\end{align*}
This function is periodic with respect to $\vk_i,\vk_j,\vq$ and is smooth everywhere except at $\vq \in \mathbb{L}^*$. 
Then similar to \cref{eqn:exchange_partial_q}, it can shown that 
\begin{align}
\mathcal{E}_{(\Omega^*)^{\times 3}}\left(F_\text{MP2,d}^{ijab}(\vk_i,\vk_j,\vk_a), \mathcal{K}^{\times 3}\right) 
&  = \mathcal{E}_{(\Omega^*)^{\times 3}}\left(\wt{F}_\text{MP2,d}^{ijab}(\vk_i,\vk_j,\vq), \mathcal{K}\times\mathcal{K}\times\mathcal{K}_\vq\right)
\nonumber \\
&\lesssim \max_{\vk_i,\vk_j \in \Omega^*} \mathcal{E}_{\Omega^*}\left(\wt{F}_\text{MP2,d}^{ijab}(\vk_i,\vk_j,\vq), \mathcal{K}_\vq\right), 
\label{eqn:mp2_direct_partial_q}
\end{align}
where the first equality uses the periodicity of $F_\text{MP2,d}^{ijab}$, $\mathcal{K}_\vq$ is an MP mesh containing all the minimum images of $\vk_a - \vk_i$ in $\Omega^*$ with $\vk_i, \vk_a \in \mathcal{K}$, and 
the second estimate uses the fact that $\wt{F}_\text{MP2,d}^{ijab}(\vk_i, \vk_j, \vq)$ is smooth and periodic with respect to $\vk_i, \vk_j$.

We first check the quadrature error $\mathcal{E}_{\Omega^*}\left(\wt{F}_\text{MP2,d}^{ijab}(\vk_i,\vk_j,\cdot), \mathcal{K}_\vq\right)$
with any fixed $\vk_i,\vk_j$. 
To simplify the notation, in this subsection, we omit the dependence on $\vk_i$, $\vk_j$, $i$, $j$, $a$, $b$, and rename the three components in  $\wt{F}_\text{MP2,d}^{ijab}$ as
\begin{align*}
2E_{ijab}(\vk_i,\vk_j,\vq) & =: E(\vq), 
\\
R_{ijab}(\vk_i, \vk_j, \vq) 
& =: R_1(\vq) =
\sum_{\vG\in \mathbb{L}^*}\dfrac{r_1(\vq + \vG)}{|\vq + \vG|^2},
\\
R_{abij}(\vk_i+\vq, \vk_j-\vq, -\vq) 
& =:R_2(\vq) =
\sum_{\vG\in \mathbb{L}^*}\dfrac{r_2(\vq + \vG)}{|\vq + \vG|^2}, 
\end{align*}
where $E(\vq)$ is smooth and periodic, and $r_1(\vq) = r_{ijab}(\vk_i, \vk_j, \vq)$ and $r_2(\vq) = r_{abij}(\vk_i+\vq, \vk_j-\vq, -\vq)$ are both smooth. 
Since $i\neq a$ and $j\neq b$, both $r_1(\vq)$ and $r_2(\vq)$ scale as  $\Or(|\vq|^2)$ near $\vq= \bm{0}$ according to \cref{eqn:expanion_rijab}.

To extract the non-smooth part of $\wt{F}_\text{MP2,d}^{ijab}$, we split the two ERIs above with $\vq$ restricted in $\Omega^*$ as 
\begin{align*}
R_1(\vq)
& 
= 
\left(
R_1(\vq) - \dfrac{r_1(\vq)}{|\vq|^2}H(\vq) 
\right)
+ 
\dfrac{r_1(\vq)}{|\vq|^2}H(\vq) 
= R_1^\text{smooth}(\vq) + R_1^\text{singular}(\vq), 
\\
R_2(\vq)
& = 
\left(
R_2(\vq) - \dfrac{r_2(\vq)}{|\vq|^2}H(\vq) 
\right)
+ 
\dfrac{r_2(\vq)}{|\vq|^2}H(\vq) 
= R_2^\text{smooth}(\vq) + R_2^\text{singular}(\vq), 
\end{align*}
where $R_*^\text{smooth}(\vq)$ is periodic (i.e., $R_*^\text{smooth}(\vq)$ and its derivatives satisfy the periodic boundary condition on $\partial\Omega^*$) and also  smooth in $\Omega^*$, and $R_*^\text{singular}(\vq)$ is in fractional form \cref{eqn:fractional_form}. 
Function $\wt{F}_\text{MP2,d}^{ijab}$ can then be decomposed as 
\begin{equation*}
\wt{F}_\text{MP2,d}^{ijab}
= 
R_1^\text{smooth}R_2^\text{smooth}E + R_1^\text{singular}R_2^\text{smooth}E + R_1^\text{smooth} R_2^\text{singular}E 
+ R_1^\text{singular}R_2^\text{singular}E. 
\end{equation*}
The first term above is periodic and smooth with respect to $\vq$ and thus has  super-algebraically decaying quadrature error.
The second term is in the fractional form as
\[
\dfrac{r_1(\vq)H(\vq)R_2^\text{smooth}(\vq)E(\vq)}{|\vq|^2},
\]
where the numerator is smooth,  compactly supported in $\Omega^*$ (due to the localizer $H(\vq)$), and scales as $\Or(|\vq|^2)$ near $\vq = \bm{0}$ (due to $r_1(\vq)$). 
This term fits \cref{cor:em_n_fraction} with $\gamma_\text{min} = 0$, and thus its quadrature error scales as $\Or(N_\vk^{-1})$. 
The third term is similar to the second one and also has $\Or(N_\vk^{-1})$ quadrature error. 

The last term has the form 
\[
\dfrac{r_1(\vq)r_2(\vq)E(\vq)H(\vq)^2}{|\vq|^4}, 
\]
where the numerator is smooth, compactly supported with respect to $\vq$ in $\Omega^*$ (due to $H(\vq)$), and scales as $\Or(|\vq|^4)$ near $\vq = \bm{0}$ (due to $r_1(\vq)$ and $r_2(\vq)$). 
Note that the exponent of the denominator is $4$, and this term also fits \cref{cor:em_n_fraction} with $\gamma_\text{min} = 0$ and has $\Or(N_\vk^{-1})$ quadrature error. 

Due to the smoothness and periodicity of $r_1(\vq), r_2(\vq), E(\vq)$ with respect to $\vk_i,\vk_j\in\Omega^*$, 
the overall asymptotic error scaling $\Or(N_\vk^{-1})$ above has its prefactor bounded by an $\Or(1)$ constant that is independent of $\vk_i$, $\vk_j$ by a similar discussion as for \cref{eqn:error_max_k}.
Thus, we have 
\[
\max_{\vk_i,\vk_j \in \Omega^*} \mathcal{E}_{\Omega^*}\left(\wt{F}_\text{MP2,d}^{ijab}(\vk_i,\vk_j,\vq), \mathcal{K}_\vq\right) 
=
\Or(N_\vk^{-1}). 
\]
Combining this estimation  with  \cref{eqn:mp2_direct_partial_q}, we have 
\[
\mathcal{E}_{\Omega^*\times \Omega^*\times \Omega^*}\left( \sum_{ijab} \wt{F}_\text{MP2,d}^{ijab}(\vk_i,\vk_j,\vq), \mathcal{K}\times\mathcal{K}\times\mathcal{K}_\vq\right)  = \Or(N_\vk^{-1}).
\]

\subsection{Quadrature error of the exchange term}
Recall that $F_\text{MP2,x}^{ijab}(\vk_i, \vk_j, \vk_a)$ is defined as 
\[
F_\text{MP2,x}^{ijab}(\vk_i, \vk_j, \vk_a)  
= 
- R_{ijba}(\vk_i, \vk_j, \vk_b - \vk_i)R_{abij}(\vk_a, \vk_b, \vk_i -\vk_a)E_{ijab}(\vk_i,\vk_j,\vk_a-\vk_i),
\]
with $\vk_b = \vk_i + \vk_j - \vk_a$. 
To isolate the integrand singularities to single variables, define $\vq_1 = \vk_b - \vk_i$ and $\vq_2 = \vk_i - \vk_a$ which lead to the  change of variables $\vk_a \rightarrow \vk_i - \vq_2$ and $\vk_j \rightarrow \vk_i + \vq_1 - \vq_2$.    
Define 
\begin{align*}
\wt{F}_\text{MP2,x}^{ijab}(\vk_i, \vq_1, \vq_2)  
& =
F_\text{MP2,x}^{ijab}(\vk_i, \vk_i+\vq_1-\vq_2, \vk_i-\vq_2)
\\
& =  
- R_{ijba}(\vk_i, \vk_i+\vq_1-\vq_2, \vq_1)R_{abij}(\vk_i-\vq_2, \vk_i + \vq_1, \vq_2)E_{ijab}(\vk_i,\vk_i+\vq_1-\vq_2,-\vq_2),
\end{align*}
which is periodic with respect to $\vk_i,\vq_1,\vq_2$ and smooth everywhere except at $\vq_1 \in \mathbb{L}^*$ or $\vq_2\in \mathbb{L}^*$. 

Similar to \cref{eqn:mp2_direct_partial_q}, we can show that 
\begin{align}
\mathcal{E}_{(\Omega^*)^{\times 3}}(F_\text{MP2,x}^{ijab}(\vk_i,\vk_j,\vk_a), \mathcal{K}^{\times 3}) 
& = \mathcal{E}_{(\Omega^*)^{\times 3}}(\wt{F}_\text{MP2,x}^{ijab}(\vk_i,\vq_1,\vq_2), \mathcal{K}\times\mathcal{K}_\vq\times\mathcal{K}_\vq) 
\nonumber\\
& \lesssim \max_{\vk_i \in \Omega^*} \mathcal{E}_{\Omega^*\times \Omega^*}(\wt{F}_\text{MP2,x}^{ijab}(\vk_i,\vq_1,\vq_2), \mathcal{K}_\vq\times\mathcal{K}_\vq), 
\label{eqn:mp2_exchange_partial_q}
\end{align}
where $\mathcal{K}_\vq$ is an MP mesh containing the minimum images of $\vk_b - \vk_i$ with $\vk_i, \vk_b \in \mathcal{K}$. 
Note that $\mathcal{K}_\vq$ is closed under inversion, i.e., $-\vq \in \mathcal{K}_\vq$ if $\vq \in \mathcal{K}_\vq$, and therefore also contains the minimum images of $\vk_i -\vk_a$ with $\vk_i, \vk_a \in\mathcal{K}$.

We next check  the quadrature error $\mathcal{E}_{\Omega^*\times \Omega^*}(\wt{F}_\text{MP2,x}^{ijab}(\vk_i,\vq_1,\vq_2), \mathcal{K}_\vq\times\mathcal{K}_\vq)$ with any fixed $\vk_i$. To simplify the notation, in this subsection, 
we omit the dependence on $\vk_i$, $i$, $j$, $a$, $b$, and rename the  three components in  $\wt{F}_\text{MP2,x}^{ijab}$ as
\begin{align*}
-E_{ijab}(\vk_i, \vk_i+\vq_1-\vq_2, -\vq_2) & =: E(\vq_1, \vq_2),
\\
R_{ijba}(\vk_i, \vk_i+\vq_1-\vq_2, \vq_1)
& =: R_1(\vq_1, \vq_2) =
\sum_{\vG\in \mathbb{L}^*}\dfrac{r_1(\vq_1 + \vG, \vq_2)}{|\vq_1 + \vG|^2},
\\
R_{abij}(\vk_i-\vq_2, \vk_i + \vq_1, \vq_2)
& =: R_2(\vq_1, \vq_2) =
\sum_{\vG\in \mathbb{L}^*}\dfrac{r_2(\vq_1, \vq_2 + \vG)}{|\vq_2 + \vG|^2}.
\end{align*}
Here $E(\vq_1,\vq_2)$ is smooth and periodic, and 
\begin{itemize}
        \item $r_1(\vq_1, \vq_2) = r_{ijba}(\vk_i, \vk_i+\vq_1-\vq_2, \vq_1)$ is smooth with respect to $\vq_1,\vq_2$, periodic with respect to $\vq_2$, and scales as $\Or(|\vq_1|^2)$ near $\vq_1 = \bm{0}$. 
        
        \item $r_2(\vq_1, \vq_2) = r_{abij}(\vk_i-\vq_2, \vk_i+\vq_1, \vq_2)$ is smooth with respect to $\vq_1,\vq_2$, periodic with respect to $\vq_1$, and scales as $\Or(|\vq_2|^2)$ near $\vq_2 = \bm{0}$. 
\end{itemize}
Further split the two ERIs above with $\vq_1, \vq_2$ restricted in $\Omega^*$ as 
\begin{align*}
R_1(\vq_1, \vq_2)
& 
= 
\left(
R_1(\vq_1, \vq_2) - \dfrac{r_1(\vq_1, \vq_2)}{|\vq_1|^2}H(\vq_1) 
\right)
+ 
\dfrac{r_1(\vq_1, \vq_2)}{|\vq_1|^2}H(\vq_1) 
= R_1^\text{smooth} + R_1^\text{singular}, 
\\
R_2(\vq_1, \vq_2)
& = 
\left(
R_2(\vq_1,\vq_2) - \dfrac{r_2(\vq_1,\vq_2)}{|\vq_2|^2}H(\vq_2) 
\right)
+ \dfrac{r_2(\vq_1,\vq_2)}{|\vq_2|^2}H(\vq_2) 
= R_2^\text{smooth}+ R_2^\text{singular}, 
\end{align*}
where $R_*^\text{smooth}$ is periodic and smooth with respect to $\vq_1,\vq_2$, and $R_*^\text{singular}(\vq_1,\vq_2)$ is in the fractional form with respect to $\vq_1$ or $\vq_2$. 
Function $\wt{F}_\text{MP2,x}^{ijab}$ can then be decomposed into  four terms, 
\[
\wt{F}
= 
R_1^\text{smooth}R_2^\text{smooth}E + R_1^\text{singular}R_2^\text{smooth}E + R_1^\text{smooth} R_2^\text{singular}E 
+ R_1^\text{singular}R_2^\text{singular}E. 
\]
The first term is periodic and smooth with respect to $\vq_1,\vq_2$ and has super-algebraically decaying quadrature error.
The second term is of the fractional form 
\[
\dfrac{r_1(\vq_1, \vq_2)H(\vq_1) (R_2^\text{smooth}E)(\vq_1,\vq_2)}{|\vq_1|^2}
\]
where the numerator is periodic and smooth with respect to $\vq_2$, smooth and compactly supported with respect to $\vq_1$ in $\Omega^*$, and scales as $\Or(|\vq_1|^2)$ near $\vq_1 = \bm{0}$. 
By the same analysis for exchange energy, the quadrature error for this term is dominated by the quadrature over $\vq_1$ with any fixed $\vq_2$ and overall scales as $\Or(N_\vk^{-1})$. 
The third term is similar to the second term and also has $\Or(N_\vk^{-1})$ quadrature error. 

The last term is still in the fractional form but now is a product of two fractions with two different denominators, i.e., 
\[
\dfrac{r_1(\vq_1,\vq_2)H(\vq_1)}{|\vq_1|^2}\dfrac{r_2(\vq_1,\vq_2)H(\vq_2)E(\vq_1,\vq_2)}{|\vq_2|^2}. 
\]
This term fits \cref{cor:em_n_fraction} with $n=2$ and $\gamma_\text{min} = 0$ and thus has $\Or(N_\vk^{-1})$ quadrature error.

The overall asymptotic error scaling $\Or(N_\vk^{-1})$ obtained above has its prefactor bounded by an $\Or(1)$ constant that is independent of $\vk_i$ due to the smoothness of all the components $r_1,r_2, E$ with respect to $\vk_i$. Thus, we have 
\[
\max_{\vk_i} \mathcal{E}_{\Omega^*\times \Omega^*}\left(\wt{F}_\text{MP2,x}^{ijab}(\vk_i,\vq_1,\vq_2), \mathcal{K}_\vq\times \mathcal{K}_\vq\right) 
=
\Or(N_\vk^{-1}). 
\]
Combining this estimation with \cref{eqn:mp2_exchange_partial_q}, we have
\[
\mathcal{E}_{\Omega^* \times \Omega^* \times \Omega^*}\left( \sum_{ijab} \wt{F}_\text{MP2,x}^{ijab}(\vk_i,\vq_1,\vq_2), \mathcal{K}\times\mathcal{K}_\vq\times\mathcal{K}_\vq\right)  = \Or(N_\vk^{-1}).
\]

Combining the two separate analysis for the direct and the exchange terms of the MP2 energy, \cref{thm:mp2_error} concludes that the quadrature error in MP2 energy calculation scales as $\Or(N_\vk^{-1})$. 
\begin{thm}[MP2 correlation energy  for 3D periodic systems]\label{thm:mp2_error}
       The finite-size error in the MP2 energy calculation satisfies 
        \[
        E_\textup{mp2}^\textup{TDL} - E_\textup{mp2}(N_\vk) 
        =
        \dfrac{1}{|\Omega^*|^3}
        \mathcal{E}_{(\Omega^*)^{\times 3}}\left(\sum_{ijab} F_\textup{mp2,d}^{ijab}+ F_\textup{mp2,x}^{ijab}, \mathcal{K}^{\times 3}\right) = \Or(N_\vk^{-1}). 
        \]
\end{thm}

\section{Madelung-constant correction, shifted Ewald kernel, and low dimensional systems}\label{sec:madelung}
\subsection{Madelung-constant  correction for 3D periodic systems}
From the analysis in \cref{thm:exchange_error}, the $\Or(N_\vk^{-\frac13})$ quadrature error in the exchange energy calculation is due to the non-smooth terms $\frac{r_{ijji}(\vk_i,\vk_i+\vq,\vq)}{|\vq|^2}$ with $i = j$, which are all of form $\frac{1}{|\vq|^2}$ asymptotically near $\vq = \bm{0}$. 
To reduce this error, it is a common practice to add a Madelung-constant shift \cite{FraserFoulkesRajagopalEtAl1996,ChiesaCeperleyMartinEtAl2006,DrummondNeedsSorouriEtAl2008} to the Ewald kernel in ERI computation as
\begin{equation}\label{eqn:shifted_ewald}
        \hat{v}_\text{shift}(\vG) = 
        \left\{
        \begin{array}{ll}
                \frac{4\pi}{|\vG|^2} & \vG \neq \bm{0} \\
                -|\Omega|N_\vk \xi & \vG = \bm{0}
        \end{array}
        \right., 
\end{equation}
with
\begin{equation}\label{eqn:madelung}
        \xi
        = 
        \dfrac{|\Omega^*|}{(2\pi)^3N_\vk}\sum_{\vq\in\mathcal{K}_\vq}
        \xsum_{\vG \in \mathbb{L}^*}\dfrac{4\pi e^{-\varepsilon|\vq+\vG|^2}}{|\vq + \vG|^2}
        - \dfrac{1}{(2\pi)^3}\int_{\mathbb{R}^3}\ud \vq\dfrac{4\pi e^{-\varepsilon|\vq|^2}}{|\vq|^2}
        -\frac{4\pi \varepsilon}{| \Omega |N_\vk}
        + \xsum_{\mathbf{R} \in \mathbb{L}_{\mathcal{K}_\vq}} \frac{\operatorname{erfc}\left(\varepsilon^{-1/2}|\mathbf{R}|/2\right)}{|\mathbf{R}|}.
\end{equation}
The constant $\varepsilon > 0$ can be arbitrary,
$\mathcal{K}_\vq$ is an $N_\vk$-sized $\Gamma$-centered MP mesh in $\Omega^*$, and $\mathbb{L}_{\mathcal{K}_\vq}$ is the real-space lattice associated with the reciprocal-space lattice $\vq+\vG$ with $\vq\in \mathcal{K}_\vq,\vG\in\mathbb{L}^*$. 
Specifically, when $\mathbb{L} = \{c_1\va_1 + c_2\va_2 + c_3\va_3: c_1,c_2,c_3\in \mathbb{Z}\}$ and $\mathcal{K}_\vq$ is of size $m\times m \times m$, this real-space lattice is defined as
\[
\mathbb{L}_{\mathcal{K}_\vq} = \{c_1m\va_1 + c_2m\va_2 + c_3m\va_3: c_1,c_2,c_3\in \mathbb{Z}\}. 
\]
Note that $\xi$ is independent of parameter $\varepsilon$ and scales as $\Or(N_\vk^{-\frac13})$ \cite{FraserFoulkesRajagopalEtAl1996}. 
With this shifted Ewald kernel, a correction is added to ERIs as 
\begin{equation}\label{eqn:eri_correction}
        \braket{n_1\vk_1,n_2\vk_2|n_3\vk_3,n_4\vk_4} 
        \rightarrow
        \braket{n_1\vk_1,n_2\vk_2|n_3\vk_3,n_4\vk_4} - \delta_{n_1n_3}\delta_{n_2n_4}\delta_{\vk_1\vk_3}\delta_{\vk_2\vk_4} \xi. 
\end{equation}
The Madelung-corrected exchange energy can then be written as
\begin{equation}\label{eqn:exchange_correction}
        E_\text{X}^\text{corrected}(N_\vk) = E_\text{X}(N_\vk) + N_\textup{occ}\xi,
\end{equation}
where $N_\text{occ}$ denotes the number of occupied bands. 

\cref{thm:exchange_madelung_error} rigorously proves that the Madelung constant correction reduces the quadrature error in the exchange energy calculation to $\Or(N_\vk^{-1})$.
Furthermore, this correction is closely connected to a singularity subtraction method, which is a classical numerical quadrature technique for singular integrals.
The basic idea of this technique is to construct an auxiliary function $h$ that has the same singularity as any concerned integrand $g$, subtract $h$ from $g$, and then compute the numerical quadrature of $g$ as 
\[
\mathcal{Q}_V(g - h, \mathcal{X}) + \mathcal{I}_V(h) \xrightarrow{\mathcal{X}\rightarrow V} \mathcal{I}_V(g),
\]
where $\mathcal{I}_V(h)$ may be computed either analytically,  or precomputed numerically with high precision, and the quadrature error  becomes $\mathcal{E}_V(g - h, \mathcal{X})$. Since $g - h$ has improved smoothness properties compared to $g$, the error $\mathcal{E}_V(g - h, \mathcal{X})$ could be asymptotically smaller than $\mathcal{E}_V(g , \mathcal{X})$. 
Note that 
\[
\mathcal{Q}_V(g - h, \mathcal{X}) + \mathcal{I}_V(h) = \mathcal{Q}_V(g, \mathcal{X}) + \mathcal{E}_V(h, \mathcal{X}).
\]
The method is thus also equivalent to adding a correction $\mathcal{E}_V(h, \mathcal{X})$ to the original quadrature $\mathcal{Q}_V(g, \mathcal{X})$.
                
The singularity subtraction method has also been used directly in the exchange energy calculation in the literature, referred to as the auxiliary function methods \cite{GygiBaldereschi1986,WenzienCappelliniBechstedt1995, CarrierRohraGorling2007,DucheminGygi2010}. 
A discussion similar to \cref{thm:exchange_madelung_error} can also be used to analyze the remaining quadrature error in existing auxiliary function methods.

\begin{thm}[Madelung corrected Fock exchange energy for 3D periodic systems]
\label{thm:exchange_madelung_error}
        The Madelung constant correction \cref{eqn:exchange_correction} reduces the finite-size error to $\Or(N_\vk^{-1})$ as
        \[
        E_\textup{x}^\textup{TDL} - E_\textup{x}^\textup{corrected}(N_\vk)
        = 
        -\dfrac{1}{|\Omega^*|^2}
        \left(
        \mathcal{E}_{\Omega^*\times \Omega^*}\left(\sum_{ij} \wt{F}_\textup{x}^{ij}(\vk_i, \vq), \mathcal{K}\times\mathcal{K}_\vq\right)
        + N_\textup{occ}|\Omega^*|^2\xi
        \right)
        = 
        \Or\left(N_\vk^{-1}\right). 
        \]
\end{thm}

\begin{proof}
        The summation of all the non-smooth terms in $\sum_{ij}\wt{F}_\text{X}^{ij}(\vk_i, \vq)$ with $\vq$ restricted to $\Omega^*$ can be expanded near $\vq = \bm{0}$ as
        \begin{equation}\label{eqn:leading_nonsmooth}
                \sum_{ij}\dfrac{r_{ijji}(\vk_i, \vk_i+\vq, \vq)}{|\vq|^2} = \dfrac{4\pi N_\text{occ}}{|\Omega|}
                        \dfrac{1}{|\vq|^2}
                        + 
                        \dfrac{\Or(|\vq|^2)}{|\vq|^2},
        \end{equation}
                where the first term turns out to be  the only source that leads to the dominant $\Or(N_\vk^{-\frac13})$ quadrature error in the exchange energy calculation (which can be proved using the localizer $H(\vq)$ and \cref{cor:em_n_fraction}). 
                                In this expansion, it is important that $\mathcal{K}_\vq$ is closed under inversion, i.e., $-\vq \in \mathcal{K}_\vq$ if $\vq \in \mathcal{K}_\vq$. This allows us to remove possible first order contribution $\frac{\vv^T\vq}{|\vq|^2}$. (More specifically, we can implicitly replace $\wt{F}_\text{X}^{ij}(\vk_i, \vq)$ by $\frac{1}{2}(\wt{F}_\text{X}^{ij}(\vk_i, \vq) + \wt{F}_\text{X}^{ij}(\vk_i, -\vq))$ in the following quadrature error analysis, see a detailed, similar discussion in \cref{rem:cubic}).

                In the corrected exchange energy calculation, the correction $N_\text{occ}|\Omega^*|^2\xi$  is exactly connected to a singularity subtraction method that removes the leading non-smooth term 
                $\frac{4\pi N_\text{occ}}{|\Omega|}
                \frac{1}{|\vq|^2}$ in \cref{eqn:leading_nonsmooth}. 
        Specifically, define a periodic function $h_\varepsilon(\vk_i, \vq)$ as
        \begin{equation} \label{eqn:h_eps}
        h_\varepsilon(\vk_i, \vq) = \frac{4\pi}{ |\Omega|}\sum_{\vG \in \mathbb{L}^*}\dfrac{ e^{-\varepsilon|\vq+\vG|^2}}{|\vq + \vG|^2},              
        \end{equation}
        where $\varepsilon > 0$ is an arbitrary constant (i.e., independent of $N_\vk$). Note that $h_\varepsilon(\vk_i, \vq)$ does no vary  with respect to $\vk_i$, and we introduce this dependence in the definition to facilitate later discussions when evaluating the numerical quadrature of $h_{\varepsilon}$ on $\mathcal{K}\times \mathcal{K}_\vq$. 
        The difference $\sum_{ij}\wt{F}_\text{X}^{ij}(\vk_i, \vq) - N_\text{occ} h_\varepsilon(\vk_i, \vq)$ is still periodic and smooth with respect to $\vk_i, \vq$ except at $\vq\in \mathbb{L}^*$ and its non-smooth part with $\vq \in \Omega^*$ can be extracted as 
        \begin{equation*}\label{eqn:function_subtraction}
        G(\vk_i, \vq) = \dfrac{\sum_{ij} r_{ijji}(\vk_i,\vk_i+\vq,\vq) - \frac{4\pi N_\text{occ}}{|\Omega|} e^{-\varepsilon|\vq|^2}}{|\vq|^2}H(\vq). 
        \end{equation*}
        This numerator is smooth and periodic with respect to $\vk_i$, and smooth and compactly supported with respect to $\vq \in \Omega^*$.
        More importantly, the numerator now scales as $\Or(|\vq|^2)$ near $\vq = \bm{0}$. 
        Thus, applying \cref{cor:em_n_fraction} with $\gamma_\text{min} = 0$ to $G(\vk_i, \vq)$ and using the same analysis approach for $\sum_{ij}G^{ij}$ in \cref{thm:exchange_error}, we have
        \begin{equation*}
                \mathcal{E}_{\Omega^*\times \Omega^*}\left(\sum_{ij}\wt{F}_\text{X}^{ij} - N_\text{occ} h_\varepsilon, \mathcal{K}\times \mathcal{K}_\vq \right)
                \lesssim 
                \mathcal{E}_{\Omega^*\times \Omega^*}\left(G(\vk_i, \vq), \mathcal{K}\times \mathcal{K}_\vq \right) = \Or(N_\vk^{-1}). 
        \end{equation*}
        
        Next, we rewrite the quadrature error for $\textstyle\sum_{ij}\wt{F}_\text{X}^{ij}$ as  
        \begin{equation}\label{eqn:singularity_subtraction}     
                \mathcal{E}_{\Omega^*\times \Omega^*}(\textstyle\sum_{ij}\wt{F}_\text{X}^{ij}, \mathcal{K}\times\mathcal{K}_\vq) 
                =
                \mathcal{E}_{\Omega^*\times \Omega^*}(N_\text{occ}h_\varepsilon, \mathcal{K}\times\mathcal{K}_\vq) + 
                \mathcal{E}_{\Omega^*\times \Omega^*}(\textstyle\sum_{ij}\wt{F}_\text{X}^{ij}-N_\text{occ}h_\varepsilon, \mathcal{K}\times\mathcal{K}_\vq).
        \end{equation}
        The singularity subtraction method defines $-\mathcal{E}_{\Omega^*\times \Omega^*}(N_\text{occ}h_\varepsilon, \mathcal{K}\times\mathcal{K}_\vq)$ as the finite-size correction, and the remaining quadrature error, i.e., the last term above,  scales as $\Or(N_\vk^{-1})$ as explained above.
        This correction can be further computed as 
        \begin{align*}
                -\mathcal{E}_{\Omega^*\times \Omega^*}(N_\text{occ}h_\varepsilon, \mathcal{K}\times\mathcal{K}_\vq)
                & = 
                \frac{4\pi N_\text{occ}}{|\Omega|}|\Omega^*|
                \left(
                \dfrac{|\Omega^*|}{N_\vk}\sum_{\vq\in\mathcal{K}_\vq}
                \xsum_{\vG \in \mathbb{L}^*}\dfrac{ e^{-\varepsilon|\vq+\vG|^2}}{|\vq + \vG|^2}
                -
                \int_{\Omega^*}\ud \vq\sum_{\vG \in \mathbb{L}^*}\dfrac{e^{-\varepsilon|\vq+\vG|^2}}{|\vq + \vG|^2}
                \right)
                \\
                & 
                =
                N_\text{occ}|\Omega^*|^2
                \left(           
                \dfrac{|\Omega^*|}{(2\pi)^3N_\vk}\sum_{\vq\in\mathcal{K}_\vq}
                \xsum_{\vG \in \mathbb{L}^*}\dfrac{4\pi e^{-\varepsilon|\vq+\vG|^2}}{|\vq + \vG|^2}
                - 
                \dfrac{1}{(2\pi)^3}
                \int_{\mathbb{R}^3}\ud \vq\dfrac{4\pi e^{-\varepsilon|\vq|^2}}{|\vq|^2}
                \right),
        \end{align*}
        which only has $\Or(N_\vk^{-1})$ difference from the Madelung constant correction $N_\text{occ}|\Omega^*|^2\xi$.                
        Thus, the Madelung constant correction $N_\text{occ}|\Omega^*|^2\xi$ also reduces the quadrature error  for $\textstyle\sum_{ij}\wt{F}_\text{X}^{ij}$ to $\Or(N_\vk^{-1})$ and is connected  to the above singular subtraction method using $-N_\text{occ}h_\varepsilon(\vk_i, \vq)$. 
\end{proof}

\begin{rem}[A new correction based on singularity subtraction]
The proof above actually proposes a slightly different finite-size correction as
\begin{align}
        E_\textup{x}^\textup{corrected,2}(N_\vk)
        & = E_\textup{x}(N_\vk)  
        -\dfrac{1}{|\Omega^*|^2}\mathcal{E}_{\Omega^*\times \Omega^*}(N_\textup{occ}h_\varepsilon, \mathcal{K}\times \mathcal{K}_\vq)
        \nonumber\\
        & = E_\textup{x}(N_\vk) + 
        N_\textup{occ}\left(
        \dfrac{|\Omega^*|}{(2\pi)^3N_\vk}\sum_{\vq\in\mathcal{K}_\vq}
        \xsum_{\vG \in \mathbb{L}^*}\dfrac{4\pi e^{-\varepsilon|\vq+\vG|^2}}{|\vq + \vG|^2}
        - 
        \dfrac{1}{(2\pi)^3}
        \int_{\mathbb{R}^3}\ud \vq\dfrac{4\pi e^{-\varepsilon|\vq|^2}}{|\vq|^2}
        \right). 
                \label{eqn:exchange_correction2}
\end{align}
Unlike the Madelung constant correction \cref{eqn:exchange_correction}, this correction depends on parameter $\varepsilon$ and also works for non-$\Gamma$-centered MP meshes $\mathcal{K}_\vq$ that is closed under inversion (recall that \cref{eqn:leading_nonsmooth} requires the inverse symmetry of $\mathcal{K}_\vq$ to remove the first-order term). 
For a $\Gamma$-centered mesh $\mathcal{K}_\vq$, this correction converges to the Madelung correction $\frac{N_\textup{occ}}{2}\xi$ when $\varepsilon \rightarrow 0$ by the facts that $\xi$ in \cref{eqn:madelung} is independent of $\varepsilon$
and its last two terms decay to zero when $\varepsilon\rightarrow 0$. 
Fixing $\varepsilon$, both corrections reduce the quadrature error to $\Or(N_\vk^{-1})$. 
\end{rem}

\begin{rem}[Madelung corrected  orbital energy]
In the above finite-size error analysis of the exchange and MP2 energies, the orbital energies at any $\vk$ point are assumed to be exact. 
However, there is also finite-size error in the orbital energy calculation even if assuming the orbital functions to be exact. 
Specifically, in the Hartree-Fock calculation with a finite MP mesh $\mathcal{K}$, the computation of an orbital energy $\varepsilon_{n\vk}$ contains a summation term
\[
 - \sum_{\vk_j\in\mathcal{K}}\sum_{j} \braket{j\vk_j,n\vk | n\vk, j\vk_j} \xrightarrow{N_\vk\rightarrow \infty} 
- \dfrac{1}{|\Omega^*|}\int_{\Omega^*}\ud \vk_j \sum_{j} R_{jnnj}(\vk_j, \vk, \vk - \vk_j). 
\]
Similar to the exchange energy, it could be shown that the quadrature error of this term scales as $\Or(N_\vk^{-1})$ if $n$ is a virtual band, and $\Or(N_\vk^{-\frac13})$ if $n$ is an occupied band. 
Following a similar discussion in \cref{thm:exchange_madelung_error}, it can be further proved that a Madelung constant correction, i.e.,
\[
\varepsilon_{i\vk_i}^\textup{corrected} = \varepsilon_{i\vk_i} + \xi,
\]
can reduce the quadrature error in each occupied orbital energy to $\Or(N_\vk^{-1})$. 
No correction is needed for the virtual orbitals. As a result, to achieve $\Or(N_\vk^{-1})$ finite-size error in practical MP2 energy or higher-order perturbation energy calculations, it is necessary to apply this Madelung constant correction to all occupied orbital energies. 
\label{rem:orbital_madelung}
\end{rem}

\subsection{Low-dimensional periodic systems}\label{subsec:lowdim}
The above error analysis for the exchange and MP2 energies is also applicable to quasi-1D and quasi-2D periodic systems, for which we consider a common model that uses the shifted Ewald kernel \cref{eqn:shifted_ewald} and samples $\vk$ points, i.e., $\mathcal{K}$, on the corresponding 1D axis and 2D plane in $\Omega^*$, respectively. 
Such an axis/plane in $\Omega^*$, denoted as $\Omega^*_\text{low}$ and illustrated in \cref{fig:omega_low}, always contains the $\Gamma$ point.
When using a $\Gamma$-centered MP mesh $\mathcal{K}$ in $\Omega^*_\text{low}$ for $\vk$ points, this model is equivalent to a supercell model where the supercell is extended in one or two periodic directions only, and the molecular orbitals in the numerical calculation satisfy the periodic boundary condition over the supercell.

\begin{figure}[htbp]
        \centering
        \includegraphics[width=0.4\textwidth]{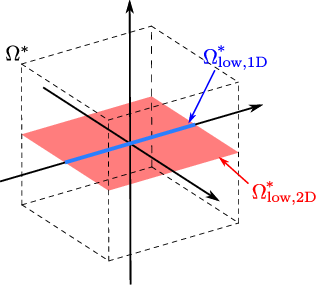}
        \caption{Illustration of $\Omega_\textup{low}^*$ for quasi-1D and quasi-2D systems.}
        \label{fig:omega_low}
\end{figure}

The energies of this low-dimensional model in the TDL can still be represented in integral forms  similar to \cref{eqn:exchange_TDL} and \cref{eqn:mp2_TDL}, sharing the same integrands but changing the integration domain for $\vk_i,\vk_j,\vk_a$ from $\Omega^*$ to the corresponding axis/plane $\Omega^*_\text{low}$, i.e., the integral $\frac{1}{|\Omega^*|}\int_{\Omega^*}\ud\vk$ is replaced by $\frac{1}{|\Omega^*_\text{low}|}\int_{\Omega^*_\text{low}}\ud\vk$ in the TDL.
Intermediate variables $\vq, \vq_1, \vq_2$ introduced in the analysis also lie in $\Omega_\text{low}^*$. 
The non-smooth terms in all the integrands are also in the same form.
An additional term is added to the exchange energy in the TDL due to the Madelung constant correction to the Ewald kernel (see \cref{appendix_low_dim}). 
Minor modifications are also needed in the numerical calculation of the correction based on singularity subtraction in \cref{eqn:exchange_correction2}. 

As detailed in \cref{appendix_low_dim}, the quadrature errors in the exchange and MP2 energy calculations by this  low-dimensional periodic model can be formulated as 
\begin{align*}
        E_\text{x,low}^\text{TDL} - E_\text{x,low}(N_\vk)
        & = 
        -\dfrac{1}{|\Omega^*|^2}
        \mathcal{E}_{\Omega_\text{low}^*\times \Omega_\text{low}^*}\left(\textstyle\sum_{ij}\wt{F}_\text{X}^{ij} - N_\text{occ}h_\varepsilon, \mathcal{K}\times\mathcal{K}_\vq\right) + \Or(N_\vk^{-1}),  
        \\
        E_\text{MP2, low}^\text{TDL} - E_\text{MP2, low}(N_\vk)
        & =
        \dfrac{1}{|\Omega^*_\text{low}|^3}
        \mathcal{E}_{(\Omega_\text{low}^*)^{\times 3}}\left(\sum_{ijab} F_\text{MP2,d}^{ijab}(\vk_i, \vk_j, \vk_a) + F_\text{MP2,x}^{ijab}(\vk_i, \vk_j, \vk_a), (\mathcal{K})^{\times 3}\right),      
\end{align*}
which both still scale as $\Or(m^{-d}) = \Or(N_\vk^{-1})$ for quasi-1D and quasi-2D systems by a similar discussion as for 3D systems. 
Here, $h_\varepsilon(\vk_i,\vq)$ is the auxiliary function in \cref{eqn:h_eps} that connects the Madelung constant correction with the singularity subtraction method.

\section{Removable discontinuity and staggered mesh method}\label{sec:removable}

Our analysis of the finite-size errors of exchange and MP2 energies is sharp for general systems.
However, the convergence rate can be improved for certain special systems with removable discontinuities. 
We first explain this concept (\Cref{sec:error_removable}), and then apply the analysis to corrected exchange energy calculations (\Cref{sec:removable_exchange}), and MP2 energy calculations (\Cref{sec:removable_mp2}). 
In particular, when the discontinuities are removable, and if $\mc{K}_\vq$ is closed under inversion and does not contain the point $\vq=\bm{0}$, the convergence rate can be improved to $o(N_{\vk}^{-1})$. 
Unfortunately, in standard exchange and MP2 calculations, $\mc{K}_{\vq}$ always includes $\vq=\bm{0}$. 
We demonstrate that a staggered mesh method is able to construct a $\mc{K}_{\vq}$ mesh that does not involve the point $\vq=\bm{0}$ for MP2 energy calculations (\Cref{sec:staggermp2}). 
We then propose a different staggered mesh method for exchange energy calculations (\Cref{subsec:staggered_ex}).
The staggered mesh method only requires some additional computation of orbitals and orbital energies.
In electronic structure calculations, these quantities can be evaluated non-self-consistently, and the additional cost can be negligible.

\subsection{Quadrature error for functions with removable discontinuity}\label{sec:error_removable}
The non-smooth terms in the corrected exchange and MP2 energy calculations that lead to dominant quadrature errors are of the fractional forms
\begin{equation}\label{eqn:removable_nonsmooth}
\begin{split}
&
\dfrac{f(\vq)H(\vq)}{|\vq|^2}\ \text{with}\  f(\vq) = \Or(|\vq|^2), 
\qquad 
\dfrac{f(\vq)H(\vq)}{|\vq|^4}\ \text{with}\  f(\vq) = \Or(|\vq|^4), 
\\
& 
\dfrac{f_1(\vq_1,\vq_2)H(\vq_1)}{|\vq_1|^2}\dfrac{f_2(\vq_1,\vq_2)H(\vq_2)}{|\vq_2|^2}\ \text{with}\  f_1 = \Or(|\vq_1|^2), \ f_2 = \Or(|\vq_2|^2),
\end{split}
\end{equation}
where the localizer $H(\vq)$ extracts the non-smooth parts out of the original integrands and all the numerators are smooth. 
Using \cref{cor:em_n_fraction},  these non-smooth terms  are shown to have $\Or(N_\vk^{-1})$ quadrature error.
This error estimate is generally sharp as supported by the numerical examples in \cref{fig:cor_em}. 


However, for certain type of integrands, the quadrature error of a trapezoidal rule can be improved. 
First consider a simple example: $\frac{f(\vq)}{|\vq|^{2}}$ with $f(\vq) = |\vq|^2$. This function  equals $1$ everywhere except at $\vq = \bm{0}$, where we set the indeterminate function to some arbitrary value (e.g., zero). 
Then the quadrature error of a trapezoidal rule equals zero if the uniform mesh does not contain $\vq = \bm{0}$ and $\Or(N_\vk^{-1})$ otherwise due to the artificially assigned value $0$ at $\vq = \bm{0}$.

Now consider the more general non-smooth term $g(\vq) = \frac{f(\vq)H(\vq)}{|\vq|^2}$ with $f(\vq) = \Or(|\vq|^2)$  in $V=[-\frac12,\frac12]^d$. 
If $f(\vq)$ can be expanded at $\vq = \bm{0}$ as
\begin{equation}\label{eqn:fq_quad_remove}
f(\vq) = C |\vq|^2 + r(\vq), \quad r(\vq) = \Or(|\vq|^3), 
\end{equation}
where $C$ denotes a generic constant, then the discontinuity of $g(\vq)$ at $\vq = \bm{0}$ becomes removable. 
Specifically, in this case, $\lim_{\vq \rightarrow \bm{0}} g(\vq) = C$ exists but $g(\bm{0})$ is indeterminate. 
We can redefine $g(\vq)$ as 
\[
\wt{g}(\vq)
        =\begin{cases}
        g(\vq), & \vq \neq \bm{0} \\
        C,            & \vq = \bm{0}
    \end{cases},
\]
which becomes continuous at $\vq = \bm{0}$.
Note that when \cref{eqn:fq_quad_remove} holds, we have $f(\vq)=\Or(|\vq|^2)$, but the converse may not be true. 

Since $\mathcal{I}_V(g) = \mathcal{I}_V(\wt{g})$, the quadrature error for $g(\vq)$ with a $\Gamma$-centered mesh $\mathcal{K}_\vq$ 
can be split as 
\begin{align}
        \mathcal{E}_{V}\left(g, \mathcal{K}_\vq\right)
        &= 
        \mathcal{I}_{V}\left(g \right) - \dfrac{|V|}{N_\vk}\sum_{\bm{0}\neq\vq \in \mathcal{K}_\vq} g(\vq) 
        = 
        \mathcal{I}_{V}\left(g\right) - \dfrac{|V|}{N_\vk}\sum_{\vq \in \mathcal{K}_\vq} \wt{g}(\vq) +
         \dfrac{|V|}{N_\vk} \wt{g}(\bm{0})
       \nonumber\\
         & = 
         \mathcal{E}_{V}\left(\wt{g}, \mathcal{K}_\vq\right)
         +
         \dfrac{|V|}{N_\vk} \wt{g}(\bm{0}),
         \label{eqn:error_removable}
\end{align}
where the first equality skips $\vq = \bm{0}$ as $g(\bm{0})$ is set to $0$ in the numerical quadrature. 
Further, we have
\[
\mathcal{E}_{V}\left(\wt{g}, \mathcal{K}_\vq\right)
=
\mathcal{E}_{V}\left(CH(\vq) , \mathcal{K}_\vq\right)
+ 
\mathcal{E}_{V}\left(\frac{r(\vq)H(\vq)}{|\vq|^2}, \mathcal{K}_\vq\right)
= 
\Or(N_\vk^{-1-\frac1d}),
\]
where the quadrature error of $CH(\vq)$ decays super-algebraically by \cref{cor:euler_maclaurin} and that of $\frac{r(\vq)H(\vq)}{|\vq|^2}$ scales as $\Or(m^{-d-1}) = \Or(N_\vk^{-1-\frac1d})$ by \cref{cor:em_n_fraction}. 
Combining the two equations above, the dominant quadrature error for $g(\vq)$ scales as $\Or(N_\vk^{-1})$ and solely comes from the term $\frac{|V|}{N_\vk} \wt{g}(\bm{0})$ in \cref{eqn:error_removable}. 
This dominant error could be avoided if the MP mesh $\mathcal{K}_\vq$ does not contain $\vq = \bm{0}$, in which case the quadrature error satisfies 
$
        \mathcal{E}_{V}\left(g, \mathcal{K}_\vq\right)
        = 
        \mathcal{E}_{V}\left(\wt{g}, \mathcal{K}_\vq\right)
        = 
        \Or(N_\vk^{-1-\frac1d}). 
$

We could further generalize the above discussion and show that if $f(\vq)$ can be expanded at $\vq = \bm{0}$ as 
\begin{equation}
f(\vq) = |\vq|^2r_1(\vq) + r_2(\vq), \quad r_2(\vq) = \Or(|\vq|^{2+s}), 
\label{eqn:fq_quadratic}
\end{equation}
with smooth functions $r_1(\vq), r_2(\vq)$, the quadrature error for $g(\vq)$ scales as 
\begin{equation}\label{eqn:error_removable_Kq}
\mathcal{E}_{V}\left(g(\vq), \mathcal{K}_\vq\right)
        = \left\{
\begin{array}{ll}
        \Or(N_\vk^{-1})                                                          & \bm{0}\in \mathcal{K}_\vq \\
        \Or(N_\vk^{-1-\frac{s}{d}})          & \bm{0}\not\in \mathcal{K}_\vq
\end{array}
\right.. 
\end{equation}


To demonstrate the validity of the analysis above, \cref{fig:removable} illustrates the performance of the trapezoidal rules over two simple examples.  
The shifted $\Gamma$-centered MP mesh is obtained from a half-mesh-size shift
of a $\Gamma$-centered mesh in all directions (see \cref{fig:staggered_Kq} for
an example of such a mesh).
The singularity of the integrand in \cref{fig:removable_a} is removable, and the shifted $\Gamma$-centered mesh method significantly outperforms the standard method, both in terms of the asymptotic scaling and the preconstant of the error. In \cref{fig:removable_b}, the singularity of the integrand is not removable. The asymptotic scaling of the two methods is the same, but the preconstant of the shifted $\Gamma$-centered mesh method is still smaller. The discussions for the other two non-smooth forms in \cref{eqn:removable_nonsmooth} are similar.

\begin{figure}
        \centering
        \captionsetup{justification=centering}
        \subfloat[$\frac{H(\vx) \sin^2(|\vx|)}{|\vx|^2}$\label{fig:removable_a}]
        {
                \includegraphics[width=0.38\textwidth]{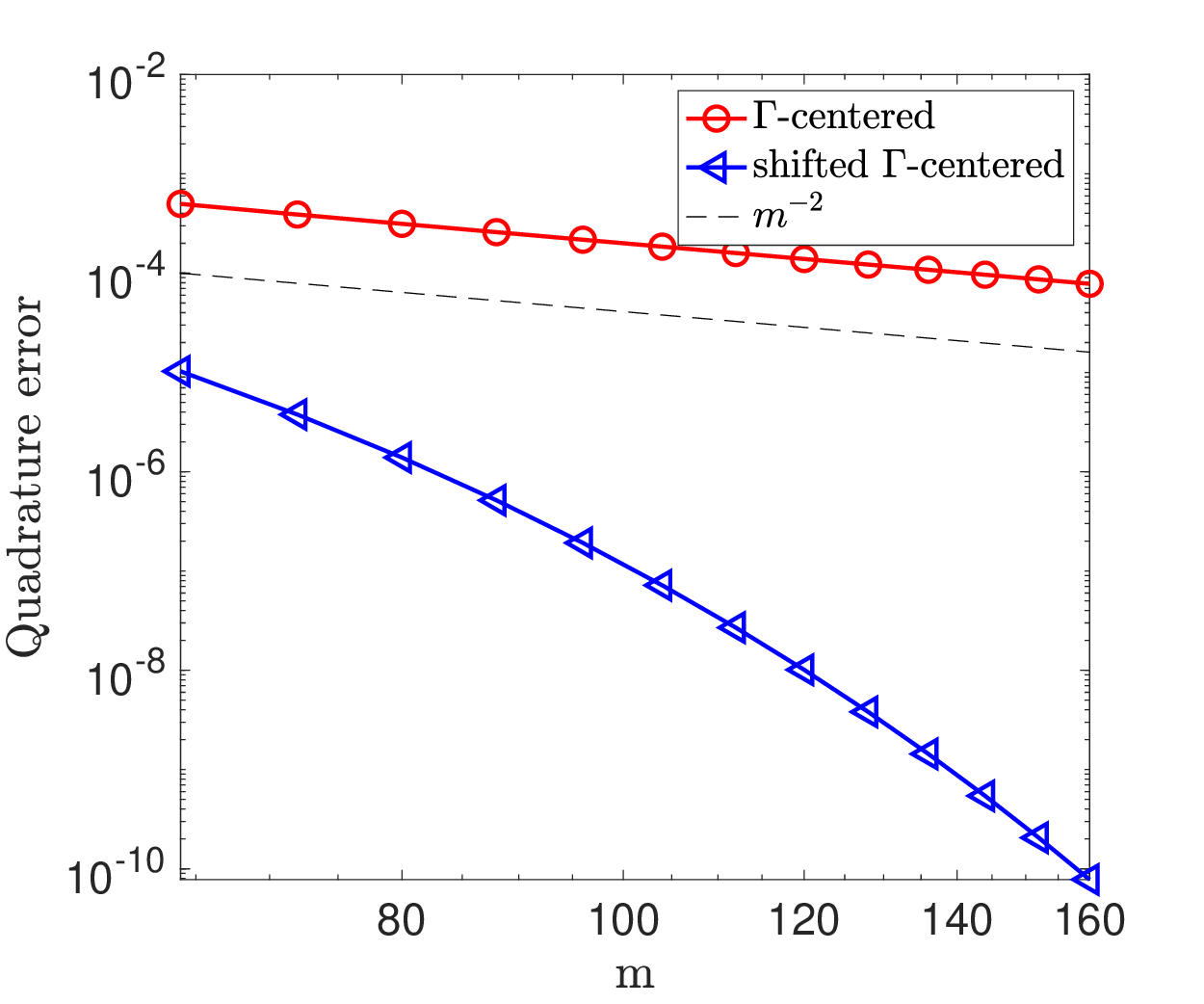}
        }
        \hspace{1em}
        \subfloat[$\frac{H(\vx) \sin^2(\sqrt{\vx^T M \vx})}{|\vx|^2}$, $M = \text{diag}(1, 100)$\label{fig:removable_b}]
        {
                \includegraphics[width=0.38\textwidth]{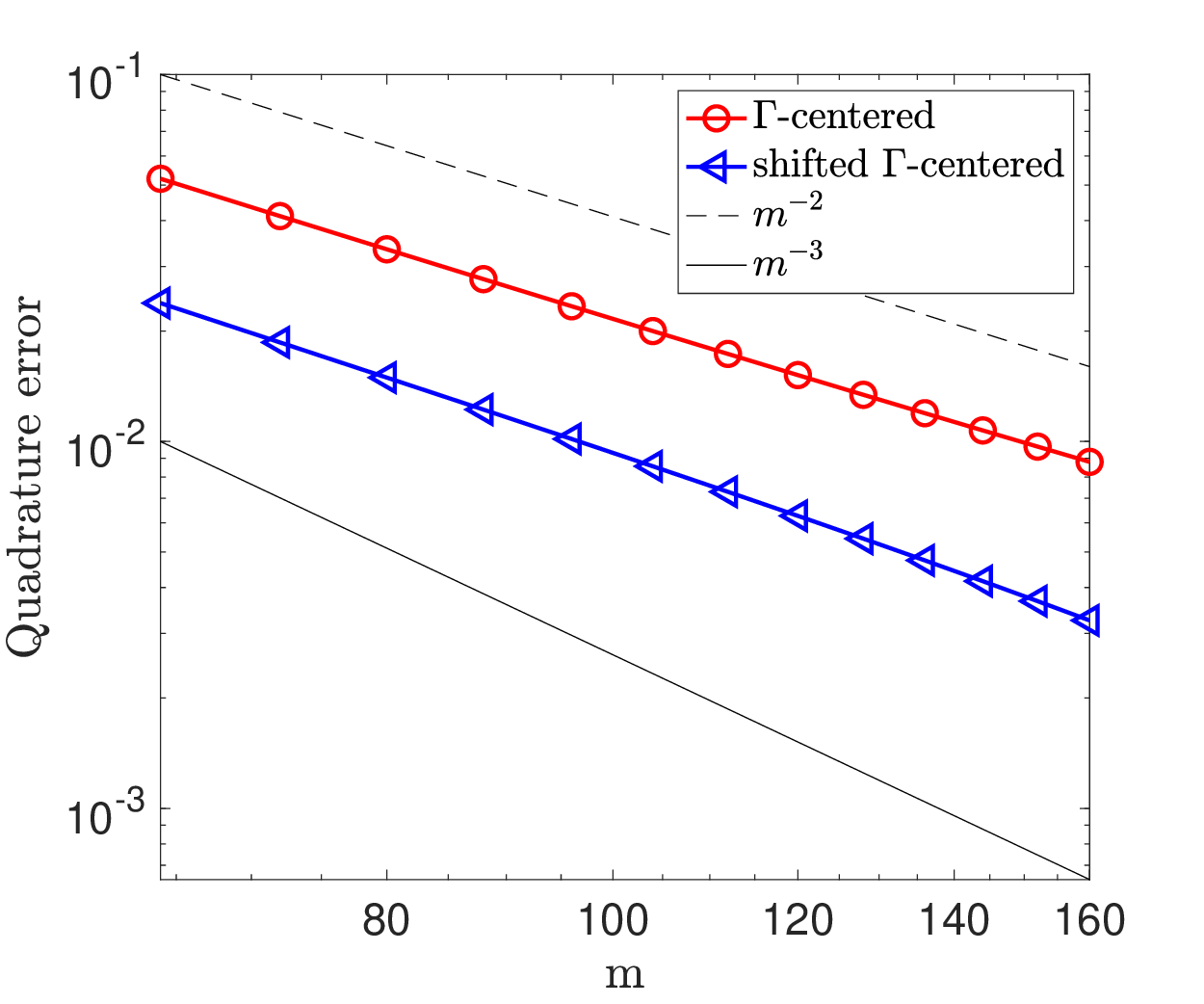}
        }
        \caption{
                Quadrature error of trapezoidal rules over functions with or without removable discontinuities.
                Consider a two-dimensional integration domain $[-\frac12, \frac12]\times[-1,1]$, and $H(\vx)$ is the localizer in \cref{eqn:localizer}.
                Two $m\times m$ MP meshes are used: the $\Gamma$-centered mesh and the half-mesh-size shift of the $\Gamma$-centered mesh. 
                The functions in (a) and (b) have removable and non-removable discontinuities at $\vx =\bm{0}$, respectively. 
                The  shifted MP mesh has super-algebraically decaying error in (a) according to \cref{eqn:error_removable_Kq}. 
                \label{fig:removable}
        }
\end{figure}

\subsection{Finite-size error of Fock exchange energy with removable discontinuity}\label{sec:removable_exchange}

Consider the corrected exchange energy $E_\text{X}^\text{corrected, 2}(N_\vk)$ in \cref{eqn:exchange_correction2} for 3D periodic systems. 
According to \cref{eqn:singularity_subtraction}, the quadrature error of the calculation writes as
\[
E_\text{X}^\text{TDL} - E_\text{X}^\text{corrected, 2}(N_\vk) = 
 -\dfrac{1}{|\Omega^*|^2}
 \mathcal{E}_{\Omega^*\times \Omega^*}(\textstyle\sum_{ij}\wt{F}_\text{X}^{ij}- N _\text{occ}h_\varepsilon, \mathcal{K}\times\mathcal{K}_\vq).
\]
Recall that $\sum_{ij}\wt{F}_\text{X}^{ij}(\vk_i, \vq)- N _\text{occ}h_\varepsilon(\vk_i, \vq) =:\wt{G}(\vk_i, \vq)$ is smooth and periodic with respect to $\vk_i$, and its quadrature error is dominated by the quadrature over $\vq$. 
We could require $\wt{G}(\vk_i, \vq)$ to have removable discontinuity at $\vq = \bm{0}$ for each $\vk_i$. 
Such a condition is sufficient to guarantee $o(N_\vk^{-1})$ quadrature error when $\mathcal{K}_\vq$ does not contain $\vq = \bm{0}$, but this is too strong. 
Specifically, similar to the discussion in \cref{eqn:exchange_partial_q}, we could integrate $\wt{G}$ over $\vk_i$ first and show that 
\begin{align*}
\mathcal{E}_{\Omega^*\times \Omega^*}(\wt{G} , \mathcal{K}\times\mathcal{K}_\vq)
& = 
\mathcal{E}_{\Omega^*}\left(\int_{\Omega^*}\ud\vk_i\wt{G}(\vk_i, \cdot) , \mathcal{K}_\vq\right)
+ 
\dfrac{|\Omega^*|}{N_\vk}\sum_{\vq \in \mathcal{K}_\vq}\mathcal{E}_{\Omega^*}\left(\wt{G}(\cdot, \vq) , \mathcal{K}_\vq\right)
\\
& \lesssim 
\mathcal{E}_{\Omega^*}\left(\int_{\Omega^*}\ud\vk_i\wt{G}(\vk_i, \cdot) , \mathcal{K}_\vq\right)
\end{align*}
where the omitted term decays super-algebraically.
Thus, it is sufficient to require 
\[
\int_{\Omega^*}\ud\vk_i \wt{G}(\vk_i, \vq)
= 
\sum_{\vG}
 \dfrac{\int_{\Omega^*}\ud\vk_i \left(\sum_{ij}r_{ijji}(\vk_i, \vk_i+\vq, \vq+\vG) - \frac{4\pi N_\text{occ}}{|\Omega|}e^{-\varepsilon|\vq+\vG|^2}\right)}{|\vq + \vG|^2}
\]
as a function of $\vq\in \Omega^*$ to have removable discontinuity at $\vq = \bm{0}$. 
The non-smooth term of this function is associated with $\vG = \bm{0}$ and also of the fractional form. 
From \cref{eqn:fq_quadratic} and \cref{eqn:error_removable_Kq}, the condition of removable discontinuity can be simplified as
\begin{equation}\label{eqn:removecond_fock}
\int_{\Omega^*}\ud \vk_i \left(\sum_{ij} r_{ijji}(\vk_i,\vk_i+\vq,\vq)  - \frac{4\pi N_\text{occ}}{|\Omega|}\right) = C|\vq|^2 + \Or(|\vq|^4),
\end{equation}
where the first and third order terms are removed implicitly by the assumption that $\mathcal{K}_\vq$ is closed under inversion. 
Under this condition, \cref{thm:exchange_removable} gives the convergence rate of the corrected exchange energy
which depends on whether $\mathcal{K}_\vq$ contains $\vq = \bm{0}$ or not. 

\begin{thm}[Corrected exchange energy for 3D periodic systems with removable discontinuity]\label{thm:exchange_removable}
        If the condition \cref{eqn:removecond_fock} holds and for an MP mesh $\mathcal{K}_\vq$ that is closed under inversion, the finite size error of the corrected exchange energy scales as
        \[
        E_\textup{x}^\textup{TDL} - E_\textup{x}^\textup{corrected,2}(N_\vk) = 
        \left\{
        \begin{array}{ll}
                \Or(N_\vk^{-1}) & \bm{0} \in \mathcal{K}_\vq \\
                \Or(N_\vk^{-\frac{5}{3}}) & \bm{0} \not\in \mathcal{K}_\vq \\
        \end{array}
        \right. .
        \]
\end{thm}

\begin{rem}[Systems with cubic unit cells]\label{rem:cubic}
Denote the left hand side of the condition \cref{eqn:removecond_fock} as a function $J(\vq)$. When the unit cell $\Omega$ is a cube (thus $\Omega^*$ is a cube centered at the origin) and $\mathcal{K}_\vq$ is cubically symmetric around $\vq = \bm{0}$, the quadrature error of $J(\vq)$ can be equivalently represented as 
\begin{align*}
        \mathcal{E}_{\Omega^*}(J(\vq), \mathcal{K}_\vq) 
        & = 
        \mathcal{E}_{\Omega^*}\left(\dfrac{1}{48}\sum_{\vq' \sim \vq} J(\vq'), \mathcal{K}_\vq\right),
\end{align*}
where `$\sim$' denotes the equivalence among vectors $(\pm q_1, \pm q_2, \pm q_3)$ and all their permutations. 
It can be verified that $\wt{J}(\vq):= \frac{1}{48}\sum_{\vq' \sim \vq} J(\vq') = C|\vq|^2 + \Or(|\vq|^4)$, satisfying the condition \cref{eqn:removecond_fock}. 
Thus, for a 3D periodic system with a cubic unit cell, the corrected exchange energy calculation with a cubically symmetric mesh $\mathcal{K}_\vq$ always has its integrand effectively satisfying the removable discontinuity condition.
This proves the observation  in \cite{DrummondNeedsSorouriEtAl2008} for  the special role of cubic symmetry for exchange energy calculations.
\end{rem}

\begin{rem}[Low-dimensional systems]
        For low-dimensional systems, the corrected exchange energy $E_\textup{x,low}^\textup{corrected,2}$ is defined in \cref{eqn:exchange_correction2_low}, and has $\Or(N_\vk^{-1})$ finite-size error when $\bm{0}\in\mathcal{K}_\vq$.
        The removable discontinuity conditions are similar to \cref{eqn:removecond_fock} simply with $\Omega^*$ replaced by $\Omega_\textup{low}^*$.
        For quasi-2D systems under the condition, the finite-size error scales as $\Or(N_\vk^{-2})$ when $\bm{0}\not\in\mathcal{K}_\vq$.       
        For quasi-1D systems, the integrand discontinuity is always removable, and the finite-size error decays super-algebraically when $\bm{0}\not\in \mathcal{K}_\vq$. 
\end{rem}

\subsection{Finite-size error of MP2 energy with removable discontinuity}\label{sec:removable_mp2}

For the MP2 energy, we can first integrate over $\vk_i, \vk_j$ for the direct term and $\vk_i$ for the exchange term, and then similarly show that the overall quadrature errors in the two terms are dominated by those of 
\[
\int_{\Omega^*}\ud\vk_i \int_{\Omega^*}\ud\vk_j\sum_{ijab}\wt{F}_\text{MP2,d}^{ijab}(\vk_i, \vk_j, \vq)
\quad \text{and} \quad 
\int_{\Omega^*}\ud\vk_i \sum_{ijab}\wt{F}_\text{MP2,x}^{ijab}(\vk_i, \vq_1,\vq_2)
\]
over $\vq$ and  $\vq_1, \vq_2$, respectively.
The partial integration for the direct term can be detailed as
\[
         \sum_{\vG,\vG'} 
        \dfrac{ \iint\ud\vk_i\ud\vk_j\sum_{ijab}E_{ijab}(\vk_i,\vk_j,\vq) r_{ijab}(\vk_i,\vk_j, \vq+\vG)
                r_{abij}(\vk_i+\vq,\vk_j-\vq, -\vq-\vG')}{|\vq+\vG|^2|\vq + \vG'|^2},
\]
where the non-smooth terms are associated with 1) $\vG=\bm{0},\vG'\neq \bm{0}$, or  $\vG\neq\bm{0},\vG'=\bm{0}$, and 2) $\vG=\vG'= \bm{0}$. 
The first case corresponds to a denominator $|\vq|^2$ and the second one corresponds to $|\vq|^4$.
Gathering these two types of terms separately, the condition of removable discontinuity can be written as 
\begin{align}   
        \sum_{\vG'\neq \bm{0}}\left.\dfrac{\iint\ud\vk_i\ud\vk_j\sum_{ijab}E_{ijab} r_{ijab}r_{ijab}}{|\vq + \vG'|^2}\right|_{\vG = \bm{0}} + \sum_{\vG\neq \bm{0}}\left.\dfrac{\iint\ud\vk_i\ud\vk_j\sum_{ijab}E_{ijab} r_{ijab}r_{ijab}}{|\vq + \vG|^2}\right|_{\vG' = \bm{0}}  &= C_1 |\vq|^2 + \Or(|\vq|^4), 
        \nonumber\\      
        \left.\iint\ud\vk_i\ud\vk_j\sum_{ijab}E_{ijab} r_{ijab}r_{ijab}\right|_{\vG=\vG' = \bm{0}} &= C_2 |\vq|^4 + \Or(|\vq|^6).
        \label{eqn:removecond_mp2_d}
\end{align}
Note that all the odd order terms are removed by the assumption that $\mathcal{K}_\vq$ for $\vq$ is closed under inversion. 

Similarly, the partial integration for the exchange term can be detailed as 
\[
\sum_{\vG,\vG'} 
\dfrac{ \int\ud\vk_i\sum_{ijab}E_{ijab}(\vk_i,\vk_i+\vq_1-\vq_2,-\vq_2)r_{ijba}(\vk_i,\vk_i+\vq_1-\vq_2, \vq_1+\vG)
        r_{abij}(\vk_i-\vq_2,\vk_i+\vq_1, \vq_2+\vG')}{|\vq_1+\vG|^2|\vq_2 + \vG'|^2}. 
\]
where the non-smooth terms are associated with 1) $\vG=\bm{0},\vG'\neq \bm{0}$, 2) $\vG\neq\bm{0},\vG'=\bm{0}$, and 3) $\vG=\vG'= \bm{0}$, corresponding to 
denominators $|\vq_1|^2, |\vq_2|^2$, and $|\vq_1|^2|\vq_2|^2$, respectively.
The condition of removable discontinuity then can be written as 
\begin{equation}\label{eqn:removecond_mp2_x}
        \begin{split}
                \sum_{\vG'\neq \bm{0}}\left.\dfrac{\int\ud\vk_i\sum_{ijab}E_{ijab} r_{ijba}r_{ijab}}{|\vq_2 + \vG'|^2}\right|_{\vG = \bm{0}} 
                & = C_3(\vq_2)|\vq_1|^2 + \Or(|\vq_1|^4),
                \\
                \sum_{\vG\neq \bm{0}}\left.\dfrac{\int\ud\vk_i\sum_{ijab}E_{ijab} r_{ijba}r_{ijab}}{|\vq_1 + \vG|^2}\right|_{\vG' = \bm{0}} 
                & = C_4(\vq_1)|\vq_2|^2 + \Or(|\vq_2|^4),
                \\
                \left.\int\ud\vk_i\sum_{ijab}E_{ijab} r_{ijba}r_{ijab}\right|_{\vG=\vG' = \bm{0}} 
                & = C_5(|\vq_1|^2 + \Or(|\vq_1|^4))(|\vq_2|^2 + \Or(|\vq_2|^4)),
        \end{split}
\end{equation}
where $C_3(\vq_2)$ and $C_4(\vq_1)$ denote two generic smooth functions.
Under these conditions, \cref{thm:mp2_removable} gives the convergence rate of the MP2 energy which depends on whether $\mathcal{K}_\vq$ contains $\vq  = \bm{0}$. 

\begin{thm}[MP2 energy for 3D periodic systems with removable discontinuity]\label{thm:mp2_removable}
        If the conditions \cref{eqn:removecond_mp2_d} and \cref{eqn:removecond_mp2_x} hold and for an MP mesh $\mathcal{K}_\vq$ that is closed under inversion, the finite size error of the MP2 energy scales as
        \[
        E_\textup{mp2}^\textup{TDL} - E_\textup{mp2}(N_\vk) = 
        \left\{
        \begin{array}{ll}
                \Or(N_\vk^{-1}) & \bm{0} \in \mathcal{K}_\vq \\
                \Or(N_\vk^{-\frac{5}{3}}) & \bm{0} \not\in \mathcal{K}_\vq \\
        \end{array}
        \right. .
        \]
\end{thm}

\begin{rem}[Low-dimensional systems]\label{rem:mp2_lowdim}
        The removable discontinuity conditions can be similarly derived  with $\Omega^*$ replaced by $\Omega_\textup{low}^*$ for low-dimensional systems.
        For quasi-2D systems under these conditions, the finite-size error scales as $\Or(N_\vk^{-2})$ when $\bm{0}\not\in \mathcal{K}_\vq$ .      
        For quasi-1D systems, the integrand discontinuity is always removable, and the finite-size error decays super-algebraically when $\bm{0}\not\in \mathcal{K}_\vq$. 
\end{rem}

\subsection{Staggered mesh method for MP2}\label{sec:staggermp2}
As shown in the earlier analysis,  when the integrand in the corresponding calculations have removable discontinuities, the dominant $\Or(N_\vk^{-1})$ quadrature error solely comes from including the point of discontinuity $\bm{0}$ in $\mathcal{K}_\vq$. 
This $\Or(N_\vk^{-1})$ error can be avoided  by constructing an MP mesh $\mathcal{K}_\vq$ that does not contain $\bm{0}$ for variables $\vq, \vq_1, \vq_2$, and thus the overall quadrature error could be smaller than $\Or(N_\vk^{-1})$. 
This observation is the main idea in the recently proposed staggered mesh method \cite{XingLiLin2021} for MP2 energy calculations. 
This method computes the MP2 energy using a non-$\Gamma$-centered MP mesh $\mathcal{K}_\vq$. 

Specifically, an MP mesh $\mathcal{K}_\text{occ}$ is used for occupied momentum vectors $\vk_i,\vk_j$, and a different, same-sized MP mesh $\mathcal{K}_\text{vir}$ is used for virtual momentum vectors $\vk_a,\vk_b$, where $\mathcal{K}_\text{vir}$ is obtained by shifting $\mathcal{K}_\text{occ}$ with half mesh size in all extended directions. 
By this choice, $\vq$ in $\wt{F}^{ijab}_\text{MP2,d}(\vk_i,\vk_j,\vq)$ and $\vq_1,\vq_2$ in $\wt{F}^{ijab}_\text{MP2,x}(\vk_i,\vq_1, \vq_2)$ still share the same mesh $\mathcal{K}_\vq$ in the numerical quadrature, but now $\mathcal{K}_\vq$ is obtained from the half-mesh-size shift of a $\Gamma$-centered $N_\vk$-sized MP mesh in $\Omega^*$.
See \cref{fig:staggered_Kq} for a 2D illustration. 
Note that $\mathcal{K}_\vq$ is closed under inversion. 

\begin{figure}[htbp]
        \centering
        \includegraphics[width=0.8\textwidth]{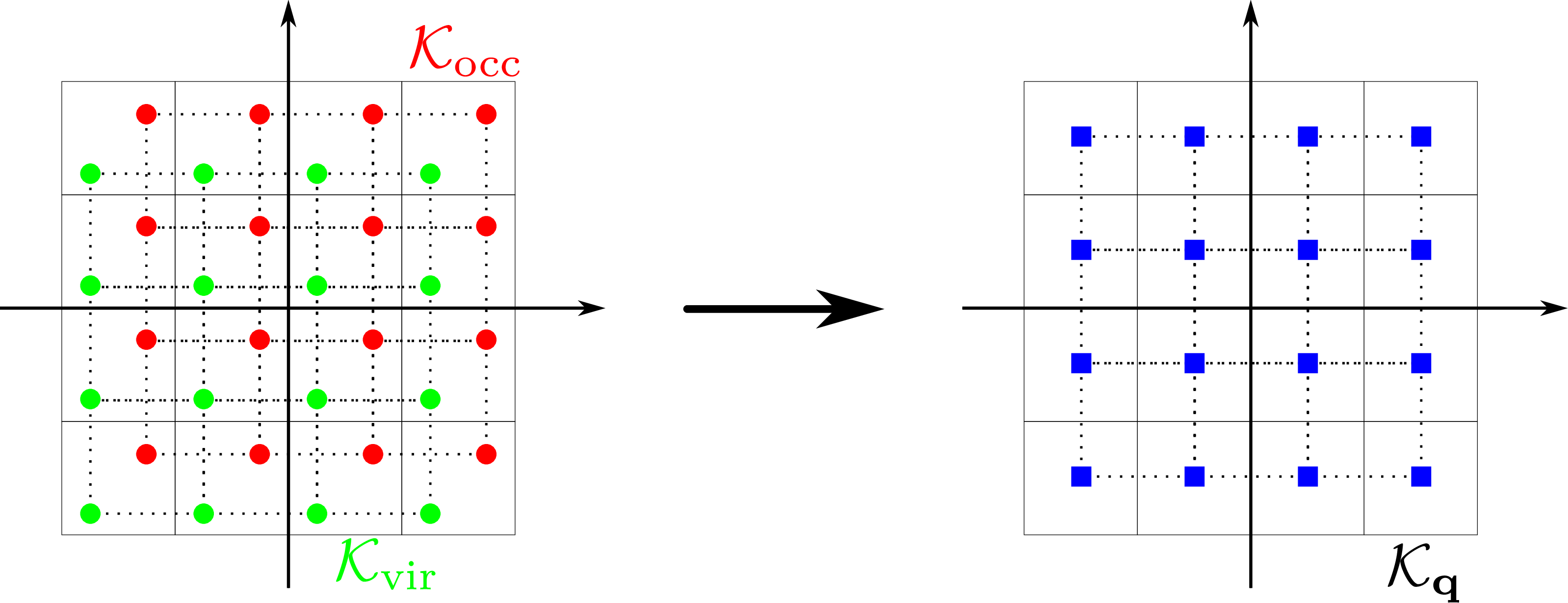}
        \caption{2D illustration of the staggered mesh method. The MP mesh $\mathcal{K}_\vq$ induced by two staggered MP meshes $\mathcal{K}_\textup{occ}$ and $\mathcal{K}_\textup{vir}$ is the half-mesh-size shift of a $\Gamma$-centered MP mesh in $\Omega^*$.}
        \label{fig:staggered_Kq}
\end{figure}

Ref. \cite{XingLiLin2021} does not contain a rigorous proof of the effectiveness of the staggered mesh method. 
Based on \cref{thm:mp2_removable} and \cref{rem:mp2_lowdim}, together with the fact that the new $\mathcal{K}_\vq$ is closed under inversion, 
we can prove in \cref{cor:mp2_staggered} that the quadrature error of the staggered mesh method is $o(N_\vk^{-1})$ quadrature error in the MP2 energy calculation when the integrand have removable discontinuities. 
Especially in the quasi-1D case, the integrand discontinuity is always removable and, more importantly, the integrand becomes smooth after the removal, leading to the super-algebraically decaying quadrature error. 

\begin{cor}[Staggered mesh method for MP2 correlation energy]
\label{cor:mp2_staggered}
        For quasi-1D systems, the staggered mesh method for MP2 energy calculation has super-algebraically decaying quadrature error. 
        For general quasi-2D and 3D systems, the quadrature errors both scale as $\Or(N_\vk^{-1})$.
        For quasi-2D and 3D systems under the removable discontinuity condition  \cref{eqn:removecond_mp2_d} and \cref{eqn:removecond_mp2_x}, 
        the quadrature errors scale as $\Or(N_\vk^{-2})$ and $\Or(N_\vk^{-\frac53})$, respectively. 
\end{cor}

In practice, it can be difficult to numerically check the conditions \cref{eqn:removecond_mp2_d} and \cref{eqn:removecond_mp2_x} for quasi-2D and 3D systems. 
Numerical tests suggest that systems with higher symmetries are more likely to satisfy the conditions and have faster decaying finite-size errors using the staggered mesh method. 
As a supporting numerical evidence for \cref{cor:mp2_staggered}, \cref{fig:stagger_mp2_num} illustrates the comparison between the standard and the staggered mesh methods for computing the MP2 energy for a quasi-2D and a quasi-1D model systems with a fixed effective potential field. 
Specifically, let the unit cell be $[0,1]^3$ and use $20\times 20\times 20$ planewave basis functions to discretize functions in the unit cell. 
The effective potential takes the local, isotropic form, 
\begin{equation}\label{eqn:potential}
V(\vx) = 
        \left\{
        \begin{array}{ll}
                - V_0 & r \leqslant  0.1 
                \\
                - V_0 \dfrac{e^{-\frac{1}{0.4-r}}}{e^{-\frac{1}{r-0.1}} + e^{-\frac{1}{0.4-r}}}   & 0.1 < r < 0.4
                \\
                0 & r \geqslant 0.4
        \end{array}
        \right.
        ,
        \quad 
        \text{with}\ \ r = |\vx - \vr_0|, 
\end{equation}
centered at $\vr_0 = (0.5, 0.5, 0.5)$ with height $V_0 = 60$. 
For each momentum vector $\vk$, we solve the corresponding effective Kohn-Sham equation to obtain $n_\text{occ} = 1$ occupied orbitals and $n_\text{vir} = 3$ virtual orbitals.
There is a direct gap between the occupied and virtual bands for this model system. 
\cref{fig:stagger_mp2_num} shows that the convergence rate of the staggered mesh method is much faster than that of the standard method, and the rate matches the analysis in \cref{cor:mp2_staggered}.
We refer readers to \cite{XingLiLin2021} for additional numerical examples illustrating the effectiveness of the staggered mesh method for MP2 calculations in model systems and real materials. 

\begin{figure}[htbp]
        \centering
        \subfloat[Quasi-2D system]
        {
                \includegraphics[width=0.45\textwidth]{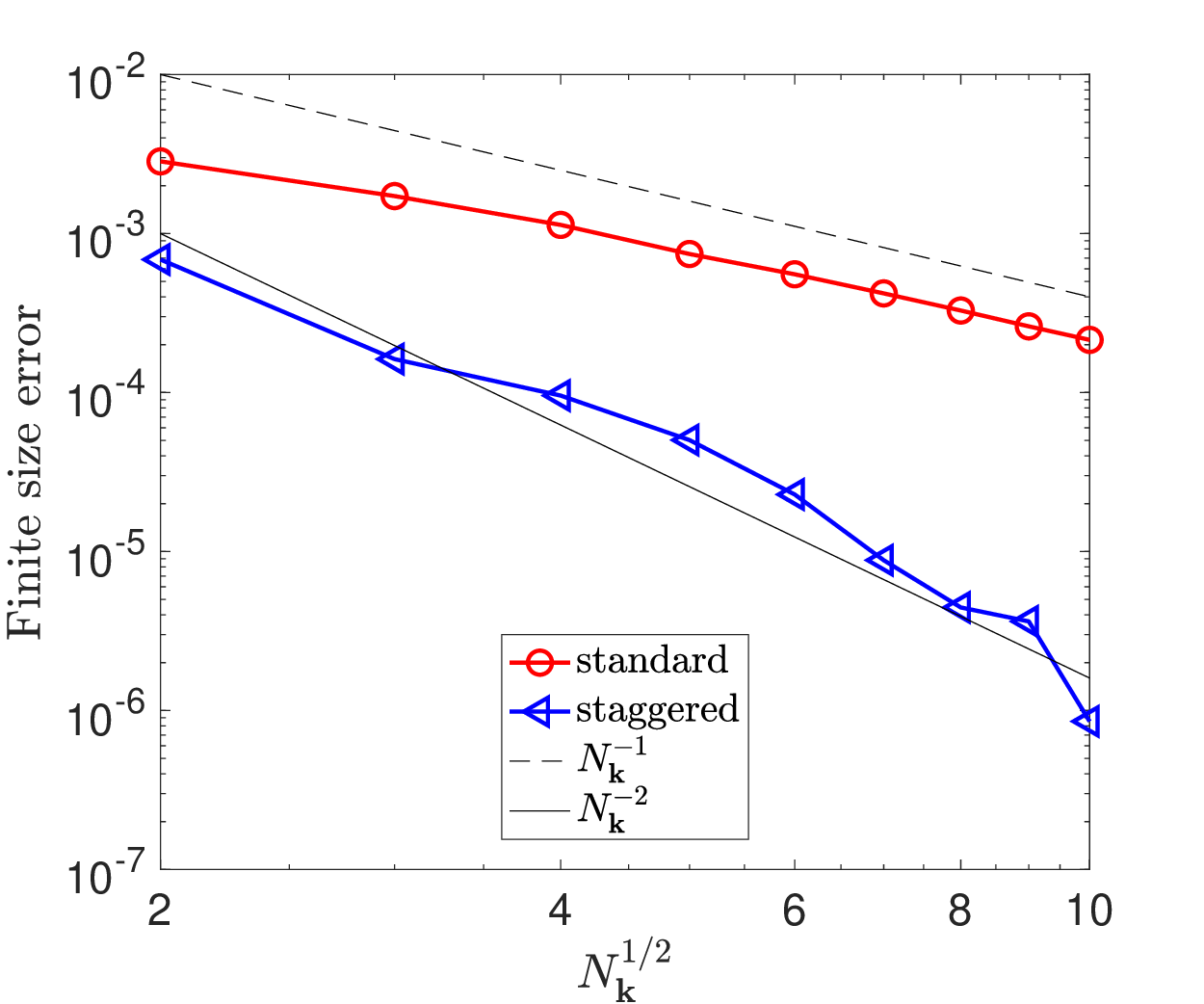}
        }
        \subfloat[Quasi-1D system]
            {
                \includegraphics[width=0.45\textwidth]{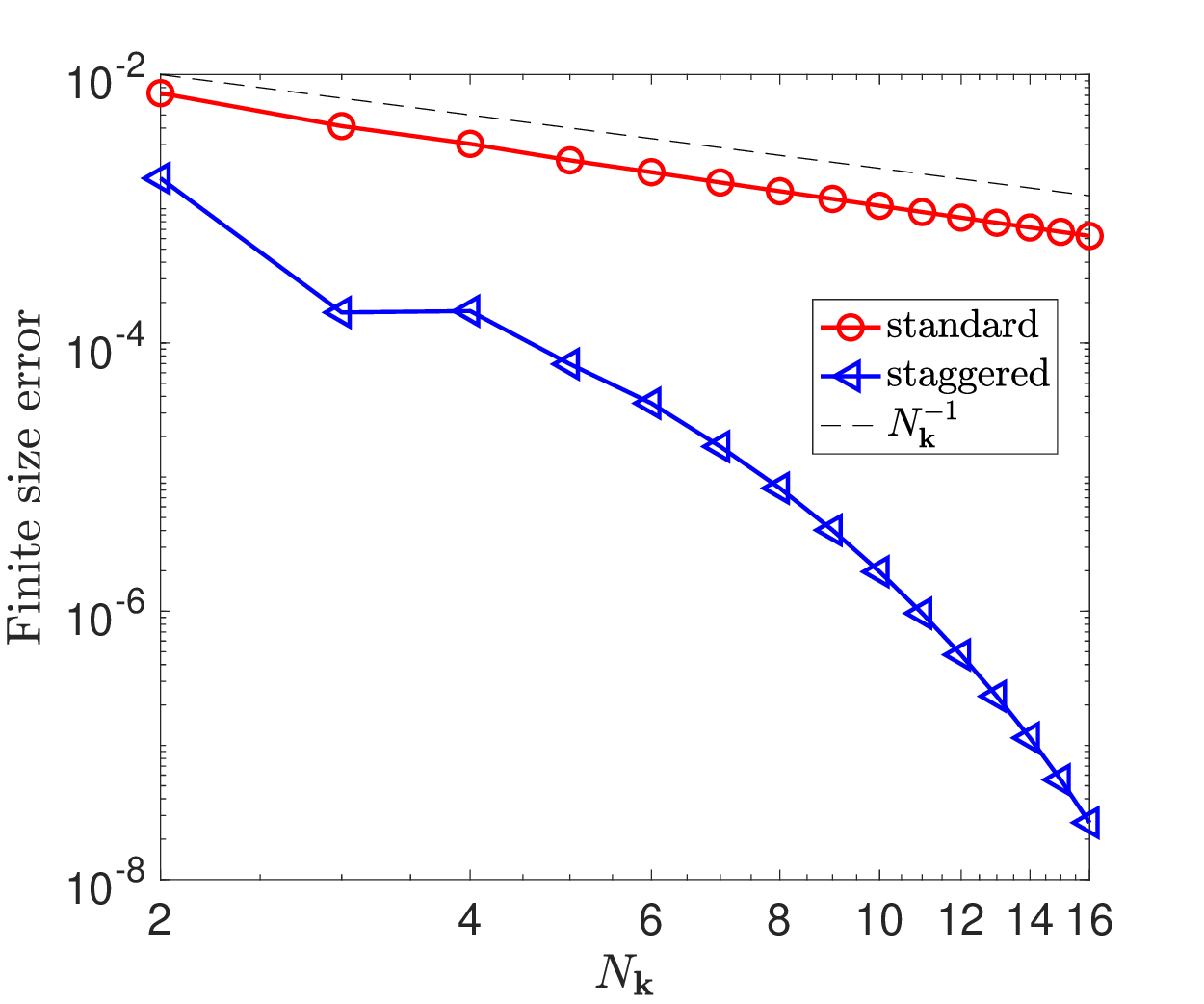}
            }
        \caption{Finite-size error of the MP2 energy calculation by the standard and the staggered mesh method for a quasi-2D and a quasi-1D model systems with MP meshes of sizes $1\times N_\vk^{ \frac12} \times N_\vk^{\frac12}$ and $1 \times 1\times N_\vk$, respectively. 
        The reference MP2 energies for the two systems is computed by the staggered mesh method with mesh size $1\times 14 \times 14$ and $1\times 1 \times 20$.}
        \label{fig:stagger_mp2_num}
\end{figure}

\subsection{Staggered mesh method for Fock exchange energy}\label{subsec:staggered_ex}
Our analysis indicates that the staggered mesh method can be an effective strategy for reducing the finite-size error when the discontinuities are removable.
We present a new staggered mesh method for exchange energy calculations using the newly introduced corrected exchange energy $E_\text{X}^\text{corrected,2}$ in \cref{eqn:exchange_correction2}. (Recall that the Madelung constant correction is only defined for $\Gamma$-centered $\mathcal{K}_\vq$ and thus cannot be combined with the staggered mesh method.)
Here $\vk_i$ and $\vk_j$ belong to two staggered MP meshes denoted by $\mathcal{K}_{\vk_i}$ and $\mathcal{K}_{\vk_j}$,  which differ by a half-mesh-size (see \cref{fig:staggered_Kqij} for a 2D illustration). 
In this way, the MP mesh $\mathcal{K}_\vq$ defined as the possible minimum images of $\vq = \vk_j - \vk_i$ with $\vk_i\in \mathcal{K}_{\vk_i}, \vk_j \in \mathcal{K}_{\vk_j}$ is also the half-mesh-size shift of a $\Gamma$-centered MP mesh in $\Omega^*$. 
We note that this $\mathcal{K}_\vq$ is used to compute the finite-size correction \cref{eqn:exchange_correction2} for the exchange energy calculation. 
In a similar manner, we could prove in \cref{cor:exchange_staggered} that the staggered mesh method has smaller than $\Or(N_\vk^{-1})$ quadrature error when the corresponding non-smooth terms have removable discontinuities. 
Note that for low-dimensional systems, minor modifications are added to the correction to the exchange energy calculation as detailed in  in \cref{appendix_low_dim} (\cref{rem:alternative_exchange_correction}). 
\begin{cor}[Staggered mesh method for Fock exchange energy]
\label{cor:exchange_staggered}
        For quasi-1D systems, the staggered mesh method for the exchange energy calculation has super-algebraically decaying quadrature error. 
For general quasi-2D and 3D systems, the quadrature errors both scale as $\Or(N_\vk^{-1})$.
For quasi-2D and 3D systems under the condition of removable discontinuity \cref{eqn:removecond_fock},
the quadrature errors scale as $\Or(N_\vk^{-2})$ and $\Or(N_\vk^{-\frac53})$, respectively. 
\end{cor}

\begin{figure}[htbp]
        \centering
        \includegraphics[width=0.8\textwidth]{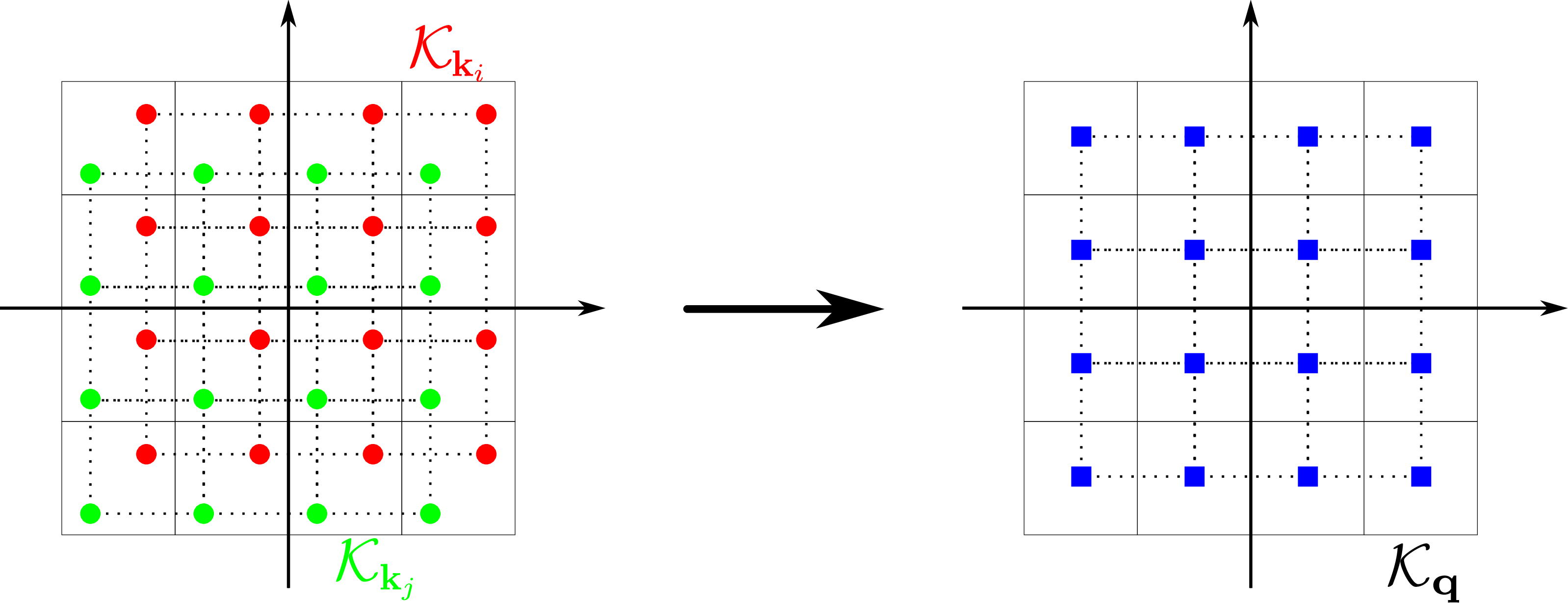}
        \caption{2D illustration of the staggered mesh method for the corrected exchange energy calculation.}
        \label{fig:staggered_Kqij}
\end{figure}

\cref{fig:stagger_ex_num} illustrates the standard and the staggered mesh methods for the corrected exchange energy calculation for a quasi-1D and a 3D model systems with the same effective potential \cref{eqn:potential}. 
The potential height $V_0$ is now set differently to $30$ for better illustration of the error scaling. 
The parameter $\varepsilon$ of the exchange energy correction in \cref{eqn:exchange_correction2} is set to $0.1$. 
These results confirm the superior performance of the staggered mesh method, and that the convergence rate in \cref{cor:mp2_staggered} is sharp.

\begin{figure}[htbp]
        \centering
        \subfloat[3D system]
        {
            \includegraphics[width=0.45\textwidth]{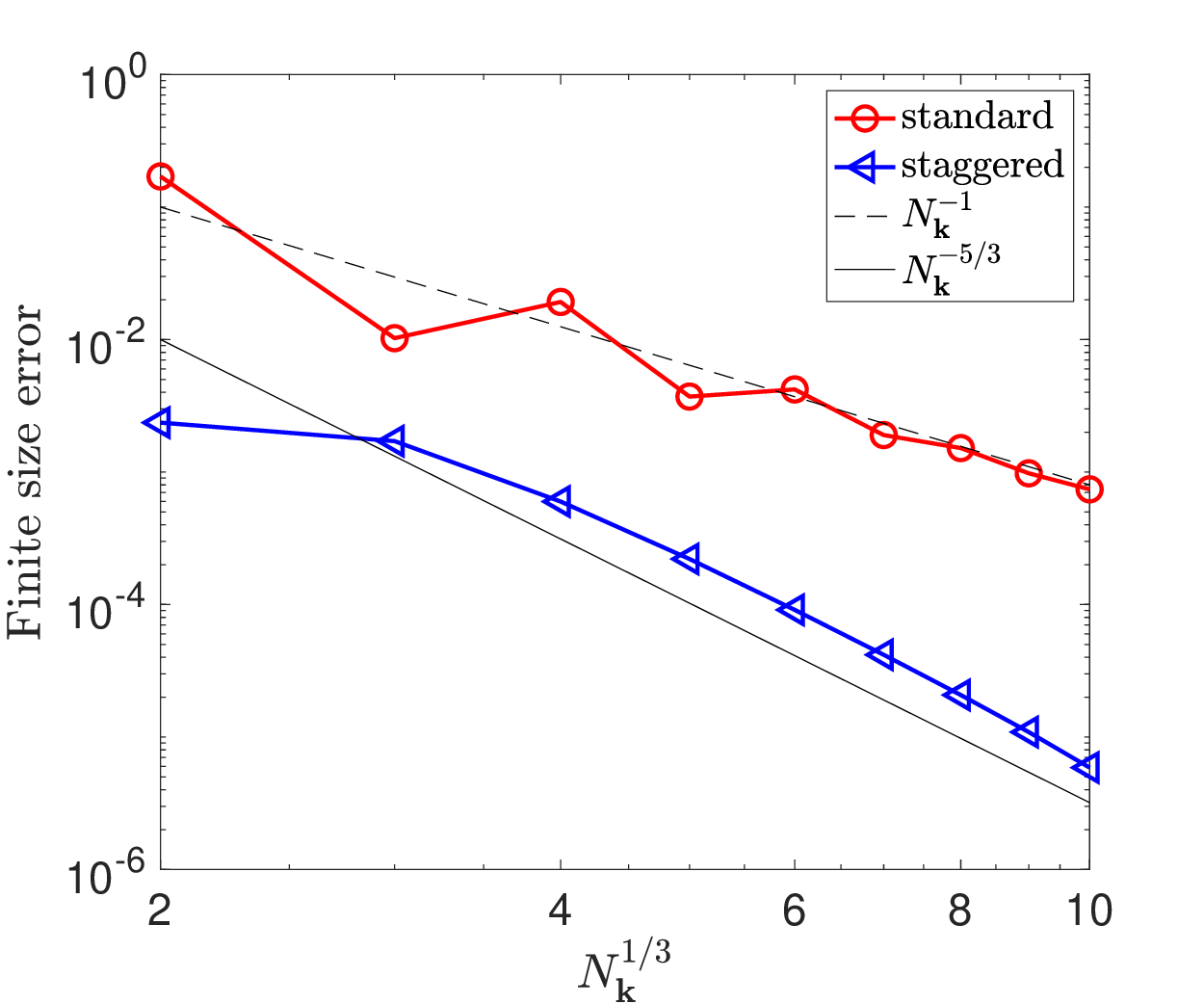}     
        }
        \subfloat[Quasi-1D system]
        {
                \includegraphics[width=0.45\textwidth]{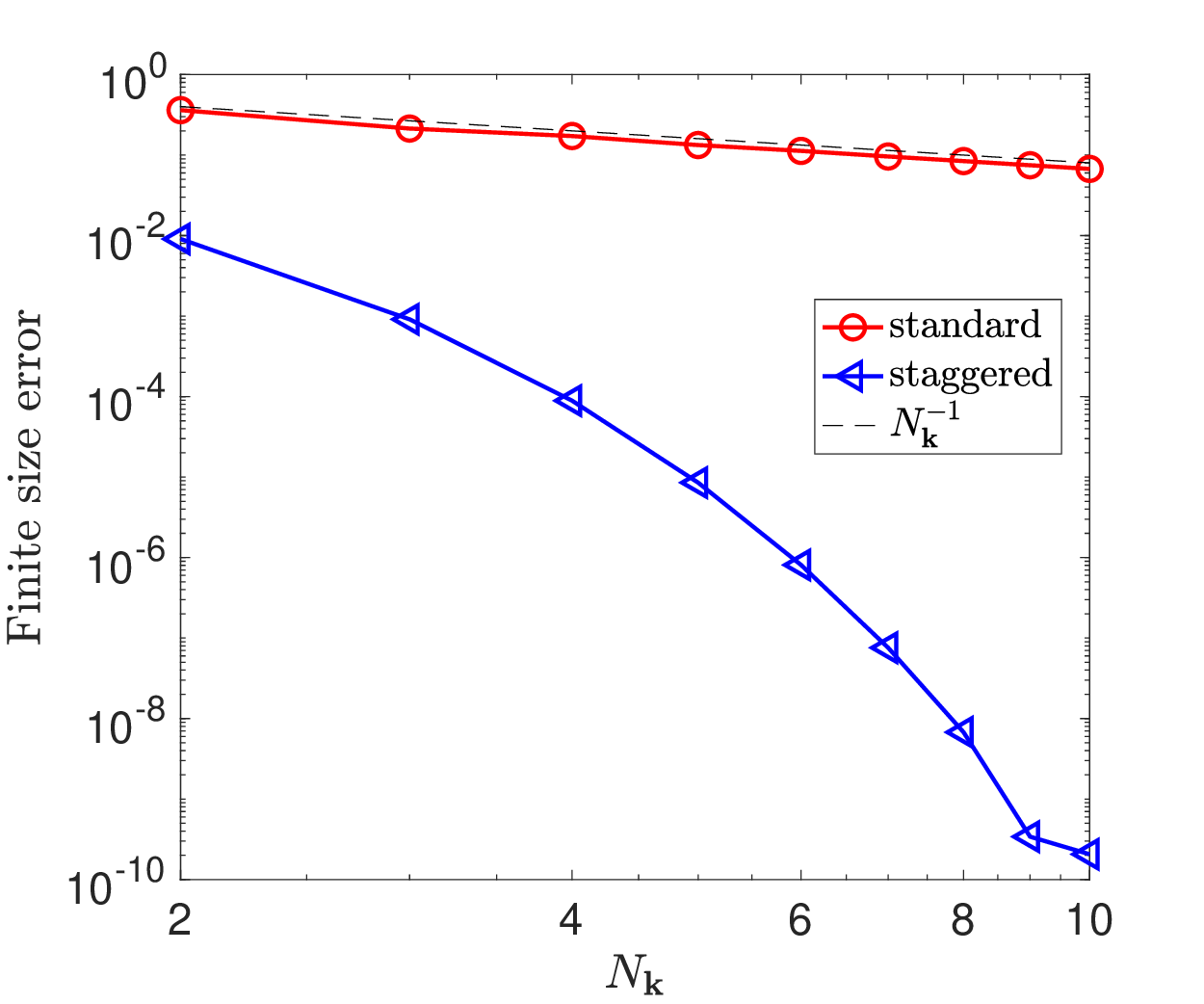}      
                }       
        \caption{Finite-size error of the corrected exchange energy calculation by the standard and the staggered mesh methods for a 3D and a quasi-1D model systems with MP meshes of sizes $N_\vk^\frac13 \times N_\vk^\frac13 \times N_\vk^\frac13$ and $1\times 1\times N_\vk$, respectively. 
        The reference energies are computed by the staggered mesh method with mesh size $14\times 14 \times 14$ and $1\times 1\times 20$.
        To avoid the effect of basis function incompleteness while illustrating the super-algebraic convergence, we increase the number of planewave functions to $40\times 40\times 40$ for the quasi-1D system in (b). Note that the error of the staggered mesh method oscillates around $10^{-10}$ after $N_\vk=9$.
        (When the number of planewave functions is $20\times 20 \times 20$, the oscillation starts at $N_\vk = 7$ and is around $10^{-8}$. )
        }
        \label{fig:stagger_ex_num}
\end{figure}

\section{Conclusion}

From the unified analysis of finite-size errors of the periodic HF theory and the MP2 theory, an immediate question is whether the finite-size errors of higher order M{\o}ller-Plesset perturbation theories (MPn) for periodic systems can be analyzed in a similar fashion. This is also a timely question, given the recent resurgence of interests on the third and fourth order perturbation theories in quantum chemistry~\cite{BanerjeeSokolov2019,BertelsLeeHead-Gordon2019,LeeLinHead-Gordon2020,RettigHaitBertelsEtAl2020,DoranHirata2021}.
We expect that the quadrature based analysis of finite-size errors can be carried out to all finite-order perturbation theories, where each energy term in the TDL is a multi-layer integral over $\Omega^*$, and its numerical calculation corresponds to a trapezoidal quadrature rule. 
The main challenge is to examine all possible non-smooth terms in the integrands and analyze their quadrature errors via a special Euler-Maclaurin formula similar to that in \cref{thm:em_n_fraction}. 
However, preliminary analysis indicates that even for the third order M{\o}ller-Plesset perturbation theories (MP3), there exists certain non-smooth components that are not in the fractional form \cref{eqn:fractional_form} and can not be readily analyzed using \cref{thm:em_n_fraction}. 
We expect that the result in \cref{thm:em_n_fraction} can be generalized to a broader class of non-smooth functions, which could then enable the analysis of finite-size errors of MPn energy calculations for fixed $n>2$.

\REV{Our error analysis focuses on insulating systems with a direct gap where the main source of the finite-size error is 
the Coulomb singularity in energy calculations. 
For gapless systems (e.g., metals), additional singularities are introduced by the orbital energy fractions $1/\varepsilon_{i\vk_i,j\vk_j}^{a\vk_a,b\vk_b}$ 
and the occupation number near the Fermi surface. 
Therefore quadrature error analysis in this case needs to take into account of these additional singularity structures, which  are not available in general except
in some special cases (such as homogeneous electron gas).
Moreover, we note that the MP2 energy calculation may diverge in the TDL for gapless systems \cite{GruneisMarsmanKresse2010,GellMannBrueckner1957}. Therefore for gapless systems, the finite-size error analysis of correlated electronic structure theories remains an open question in general.
}

Besides finite-order perturbation theories, another possible generalization is to consider certain infinite order perturbation energies with a selected set of Feynman diagrams. Examples include the random phase approximation (RPA) \REV{(see \cite{XingLin2022})} and the coupled-cluster theory (CC). 
While it may be possible to analyze the quadrature error of the contribution from each order of the diagram, the infinite summation can still pose a significant challenge for the rigorous analysis of the finite-size error. 
The generalization of the singularity subtraction method (related to the Madelung constant correction) and the staggered mesh method is also of practical interest for such higher-order and infinite-order perturbation theory calculations.

\vspace{1em}
\noindent\textbf{Acknowledgement:}

\REV{This material is based upon work supported by the U.S. Department of Energy, Office of Science, Office of Advanced Scientific Computing Research and Office of Basic Energy Sciences, Scientific Discovery through Advanced Computing (SciDAC) program (X.X.).
This work was also partially supported by the China Scholarship Council under File No. 201906040071. (X.L), the Air Force Office of Scientific Research under award number FA9550-18-1-0095, by the Department of Energy under Grant No.\ DE-SC0017867,  the Center for Advanced Mathematics for Energy Research Applications (CAMERA) program (L.L.). L.L. is a Simons Investigator.}
We thank Timothy Berkelbach and Garnet Chan for insightful discussions on the finite size  effects.

\bibliographystyle{abbrv}
\bibliography{mp2}

\begin{thebibliography}{10}

\bibitem{BakerHodgson1971}
C.~T. Baker and G.~S. Hodgson.
\newblock Asymptotic expansions for integration formulas in one or more
  dimensions.
\newblock {\em SIAM J. Numer. Anal.}, 8(2):473--480, 1971.

\bibitem{BanerjeeSokolov2019}
S.~Banerjee and A.~Y. Sokolov.
\newblock {Third-order algebraic diagrammatic construction theory for electron
  attachment and ionization energies: Conventional and Green's function
  implementation}.
\newblock {\em J. Chem. Phys.}, 151(22):224112, 2019.

\bibitem{BertelsLeeHead-Gordon2019}
L.~W. Bertels, J.~Lee, and M.~Head-Gordon.
\newblock Third-order {M{\o}ller--Plesset} perturbation theory made useful?
  {Choice} of orbitals and scaling greatly improves accuracy for
  thermochemistry, kinetics, and intermolecular interactions.
\newblock {\em J. Phys. Chem. Lett.}, 10(15):4170--4176, 2019.

\bibitem{blochl1994improved}
P.~E. Bl{\"o}chl, O.~Jepsen, and O.~K. Andersen.
\newblock Improved tetrahedron method for brillouin-zone integrations.
\newblock {\em Physical Review B}, 49(23):16223, 1994.

\bibitem{BoothGruneisKresseAlavi2013}
G.~H. Booth, A.~Gr{\"u}neis, G.~Kresse, and A.~Alavi.
\newblock Towards an exact description of electronic wavefunctions in real
  solids.
\newblock {\em Nature}, 493(7432):365--370, 2013.

\bibitem{BrouderPanatiCalandraEtAl2007}
C.~Brouder, G.~Panati, M.~Calandra, C.~Mourougane, and N.~Marzari.
\newblock Exponential localization of {Wannier} functions in insulators.
\newblock {\em Phys. Rev. Lett.}, 98:046402, 2007.

\bibitem{CancesEhrlacherGontierEtAl2020}
E.~Cances, V.~Ehrlacher, D.~Gontier, A.~Levitt, and D.~Lombardi.
\newblock Numerical quadrature in the {B}rillouin zone for periodic
  schr{\"o}dinger operators.
\newblock {\em Numer. Math.}, 144(3):479--526, 2020.

\bibitem{CarrierRohraGorling2007}
P.~Carrier, S.~Rohra, and A.~G{\"o}rling.
\newblock {General treatment of the singularities in Hartree-Fock and
  exact-exchange Kohn-Sham methods for solids}.
\newblock {\em Phys. Rev. B}, 75(20):205126, 2007.

\bibitem{ChiesaCeperleyMartinEtAl2006}
S.~Chiesa, D.~M. Ceperley, R.~M. Martin, and M.~Holzmann.
\newblock {Finite-size error in many-body simulations with long-range
  interactions}.
\newblock {\em Phys. Rev. Lett.}, 97(7):6--9, 2006.

\bibitem{DaboKozinskySingh-MillerEtAl2008}
I.~Dabo, B.~Kozinsky, N.~E. Singh-Miller, and N.~Marzari.
\newblock {Electrostatics in periodic boundary conditions and real-space
  corrections}.
\newblock {\em Phys. Rev. B}, 77(11):1--13, 2008.

\bibitem{LeeuwPerramSmith1980}
S.~W. de~Leeuw, J.~W. Perram, and E.~R. Smith.
\newblock Simulation of electrostatic systems in periodic boundary conditions.
  {I. Lattice} sums and dielectric constants.
\newblock {\em Proc. Roy. Soc. A}, 373(1752):27--56, 1980.

\bibitem{DoranHirata2021}
A.~E. Doran and S.~Hirata.
\newblock Stochastic evaluation of fourth-order many-body perturbation
  energies.
\newblock {\em J. Chem. Phys.}, 154(13):134114, 2021.

\bibitem{DrummondNeedsSorouriEtAl2008}
N.~D. Drummond, R.~J. Needs, A.~Sorouri, and W.~M.~C. Foulkes.
\newblock {Finite-size errors in continuum quantum Monte Carlo calculations}.
\newblock {\em Phys. Rev. B}, 78(12):1--19, 2008.

\bibitem{DucheminGygi2010}
I.~Duchemin and F.~Gygi.
\newblock A scalable and accurate algorithm for the computation of
  {Hartree--Fock} exchange.
\newblock {\em Comput. Phys. Commun.}, 181(5):855--860, 2010.

\bibitem{FoulkesMitasNeedsEtAl2001}
W.~Foulkes, L.~Mitas, R.~Needs, and G.~Rajagopal.
\newblock {Quantum Monte Carlo simulations of solids}.
\newblock {\em Rev. Mod. Phys.}, 73:33, 2001.

\bibitem{FraserFoulkesRajagopalEtAl1996}
L.~M. Fraser, W.~M.~C. Foulkes, G.~Rajagopal, R.~Needs, S.~Kenny, and
  A.~Williamson.
\newblock {Finite-size effects and Coulomb interactions in quantum Monte Carlo
  calculations for homogeneous systems with periodic boundary conditions}.
\newblock {\em Phys. Rev. B}, 53(4):1814--1832, 1996.

\bibitem{FreysoldtNeugebauerVan2009}
C.~Freysoldt, J.~Neugebauer, and C.~G. Van~de Walle.
\newblock Fully ab initio finite-size corrections for charged-defect supercell
  calculations.
\newblock {\em Phys. Rev. Lett.}, 102:016402, 2009.

\bibitem{GellMannBrueckner1957}
M.~Gell-Mann and K.~A. Brueckner.
\newblock Correlation energy of an electron gas at high density.
\newblock {\em Phys. Rev.}, 106(2):364, 1957.

\bibitem{GruberLiaoTsatsoulisEtAl2018}
T.~Gruber, K.~Liao, T.~Tsatsoulis, F.~Hummel, and A.~Gr{\"{u}}neis.
\newblock {Applying the coupled-cluster ansatz to solids and surfaces in the
  thermodynamic limit}.
\newblock {\em Phys. Rev. X}, 8(2):021043, 2018.

\bibitem{GruneisMarsmanKresse2010}
A.~Gr{\"u}neis, M.~Marsman, and G.~Kresse.
\newblock {Second-order M{\o}ller--Plesset perturbation theory applied to
  extended systems. II. Structural and energetic properties}.
\newblock {\em J. Chem. Phys.}, 133(7):074107, 2010.

\bibitem{GygiBaldereschi1986}
F.~Gygi and A.~Baldereschi.
\newblock {Self-consistent Hartree-Fock and screened-exchange calculations in
  solids: Application to silicon}.
\newblock {\em Phys. Rev. B}, 34:4405--4408, 1986.

\bibitem{HolzmannClayMoralesEtAl2016}
M.~Holzmann, R.~C. Clay, M.~A. Morales, N.~M. Tubman, D.~M. Ceperley, and
  C.~Pierleoni.
\newblock {Theory of finite size effects for electronic quantum Monte Carlo
  calculations of liquids and solids}.
\newblock {\em Phys. Rev. B}, 94(3):1--16, 2016.

\bibitem{JavedTrefethen2014}
M.~Javed and L.~N. Trefethen.
\newblock A trapezoidal rule error bound unifying the {Euler--Maclaurin}
  formula and geometric convergence for periodic functions.
\newblock {\em Proc. Royal Soc. A}, 470(2161):20130571, 2014.

\bibitem{KohnSham1965}
W.~Kohn and L.~Sham.
\newblock {Self-consistent equations including exchange and correlation
  effects}.
\newblock {\em Phys. Rev.}, 140:A1133--A1138, 1965.

\bibitem{LeeLinHead-Gordon2020}
J.~Lee, L.~Lin, and M.~Head-Gordon.
\newblock Systematically improvable tensor hypercontraction: {Interpolative}
  separable density-fitting for molecules applied to exact exchange, second-and
  third-order {M{\o}ller--Plesset} perturbation theory.
\newblock {\em J. Chem. Theory Comput.}, 16(1):243--263, 2020.

\bibitem{LiaoGrueneis2016}
K.~Liao and A.~Gr{\"{u}}neis.
\newblock {Communication: Finite size correction in periodic coupled cluster
  theory calculations of solids}.
\newblock {\em J. Chem. Phys.}, 145(14):141102, 2016.

\bibitem{LinZongCeperley2001}
C.~Lin, F.~Zong, and D.~M. Ceperley.
\newblock {Twist-averaged boundary conditions in continuum quantum Monte Carlo
  algorithms}.
\newblock {\em Phys. Rev. E}, 64(1):016702, 2001.

\bibitem{Lyness1976}
J.~Lyness.
\newblock An error functional expansion for {$N$}-dimensional quadrature with
  an integrand function singular at a point.
\newblock {\em Math. Comput.}, 30(133):1--23, 1976.

\bibitem{LynessMcHugh1970}
J.~N. Lyness and J.~B.~B. McHugh.
\newblock {On the remainder term in the $N$-dimensional Euler Maclaurin
  expansion}.
\newblock {\em Numer. Math.}, 15(4):333--344, 1970.

\bibitem{MakovPayne1995}
G.~Makov and M.~C. Payne.
\newblock {Periodic boundary conditions in ab initio calculations}.
\newblock {\em Phys. Rev. B}, 51:4014, 1995.

\bibitem{MarsmanGruneisPaierKresse2009}
M.~Marsman, A.~Gr{\"u}neis, J.~Paier, and G.~Kresse.
\newblock {Second-order M{\o}ller--Plesset perturbation theory applied to
  extended systems. I. Within the projector-augmented-wave formalism using a
  plane wave basis set}.
\newblock {\em J. Chem. Phys.}, 130(18):184103, 2009.

\bibitem{Martin2008}
R.~Martin.
\newblock {\em Electronic Structure: Basic Theory and Practical Methods}.
\newblock Cambridge Univ. Pr., 2008.

\bibitem{McClainSunChanEtAl2017}
J.~McClain, Q.~Sun, G.~K.~L. Chan, and T.~C. Berkelbach.
\newblock {Gaussian-based coupled-cluster theory for the ground-state and band
  structure of solids}.
\newblock {\em J. Chem. Theory Comput.}, 13(3):1209--1218, 2017.

\bibitem{MihmMcIsaacShepherd2019}
T.~N. Mihm, A.~R. McIsaac, and J.~J. Shepherd.
\newblock An optimized twist angle to find the twist-averaged correlation
  energy applied to the uniform electron gas.
\newblock {\em J. Chem. Phys.}, 150(19):191101, 2019.

\bibitem{MihmYangShepherd2020}
T.~N. Mihm, B.~Yang, and J.~J. Shepherd.
\newblock Power laws used to extrapolate the coupled cluster correlation energy
  to the thermodynamic limit, 2020.

\bibitem{MonacoPanatiPisanteEtAl2018}
D.~Monaco, G.~Panati, A.~Pisante, and S.~Teufel.
\newblock Optimal decay of wannier functions in chern and quantum hall
  insulators.
\newblock {\em Commun. Math. Phys.}, 359(1):61--100, 2018.

\bibitem{MonkhorstPack1976}
H.~J. Monkhorst and J.~D. Pack.
\newblock Special points for {B}rillouin-zone integrations.
\newblock {\em Phys. Rev. B}, 13(12):5188, 1976.

\bibitem{MuellerPaulus2012}
C.~M{\"u}ller and B.~Paulus.
\newblock Wavefunction-based electron correlation methods for solids.
\newblock {\em Phys. Chem. Chem. Phys.}, 14(21):7605--7614, 2012.

\bibitem{RettigHaitBertelsEtAl2020}
A.~Rettig, D.~Hait, L.~W. Bertels, and M.~Head-Gordon.
\newblock {Third-order M{\o}ller--Plesset} theory made more useful? {The} role
  of density functional theory orbitals.
\newblock {\em J. Chem. Theory Comput.}, 16(12):7473--7489, 2020.

\bibitem{SchaeferRambergerKresse2017}
T.~Sch{\"a}fer, B.~Ramberger, and G.~Kresse.
\newblock Quartic scaling {MP2} for solids: {A} highly parallelized algorithm
  in the plane wave basis.
\newblock {\em J. Chem. Phys.}, 146(10):104101, 2017.

\bibitem{ShavittBartlett2009}
I.~Shavitt and R.~J. Bartlett.
\newblock {\em {Many-body methods in chemistry and physics: MBPT and
  coupled-cluster theory}}.
\newblock Cambridge Univ. Pr., 2009.

\bibitem{spencer08}
J.~Spencer and A.~Alavi.
\newblock Efficient calculation of the exact exchange energy in periodic
  systems using a truncated {Coulomb} potential.
\newblock {\em Phys. Rev. B}, 77(19):193110, 2008.

\bibitem{SundararamanArias2013}
R.~Sundararaman and T.~A. Arias.
\newblock {Regularization of the Coulomb singularity in exact exchange by
  Wigner-Seitz truncated interactions: Towards chemical accuracy in nontrivial
  systems}.
\newblock {\em Phys. Rev. B}, 87(16):165122, 2013.

\bibitem{SzaboOstlund1989}
A.~Szabo and N.~Ostlund.
\newblock {\em {Modern Quantum Chemistry: Introduction to Advanced Electronic
  Structure Theory}}.
\newblock McGraw-Hill, New York, 1989.

\bibitem{TrefethenWeideman2014}
L.~N. Trefethen and J.~Weideman.
\newblock The exponentially convergent trapezoidal rule.
\newblock {\em SIAM Review}, 56(3):385--458, 2014.

\bibitem{WenzienCappelliniBechstedt1995}
B.~Wenzien, G.~Cappellini, and F.~Bechstedt.
\newblock Efficient quasiparticle band-structure calculations for cubic and
  noncubic crystals.
\newblock {\em Phys. Rev. B}, 51(20):14701, 1995.

\bibitem{XingLiLin2021}
X.~Xing, X.~Li, and L.~Lin.
\newblock {Staggered mesh method for correlation energy calculations of solids:
  Second order M{\o}ller-Plesset perturbation theory}.
\newblock {\em J. Chem. Theory Comput.}, 17(8):4733--4745, 2021.

\bibitem{XingLin2022}
X.~Xing and L.~Lin.
\newblock Staggered mesh method for correlation energy calculations of solids:
  Random phase approximation in direct ring coupled cluster doubles and
  adiabatic connection formalisms.
\newblock {\em J. Chem. Theory Comput.}, 18(2):763--775, 2022.

\end{thebibliography}

\newpage
\begin{appendices}
\crefalias{section}{appendix}

\section{Euler-Maclaurin formula for a special class of non-smooth functions}\label{sec:em}

In this section, we introduce the main technical results that have been used throughout the paper, i.e., the improved quadrature error analysis for a special class of non-smooth integrands that are in the fractional form \cref{eqn:fractional_form}. 
We first recall the standard Euler-Maclaurin analysis of the quadrature error for a trapezoidal rule in \Cref{sec:standard_em}, and then introduce the generalized Euler-Maclaurin analysis in \Cref{sec:improve_em}.
The sharpness of the convergence rates is demonstrated using numerical results for various types of singular integrands in \Cref{sec:numer_em}.
Our analysis generalizes the analysis of Lyness \cite{Lyness1976} to a broader class of singular integrands.
Compared to \cite{Lyness1976}, our analysis is also simpler and in particular does not rely on certain special properties of homogeneous polynomials.

\subsection{Standard Euler-Maclaurin formula}\label{sec:standard_em}
Consider a hypercube $V = [0, L]^d + \vb$ of edge length $L$ and cornered at point $\vb$ in $\mathbb{R}^d$. Let $\mathcal{X}$ be an $m\times \cdots \times m$ uniform mesh inside $V$ defined as
\[
\mathcal{X}  = \left\{\vb + \frac{L}{m}\left(\left(j_1, j_2,\cdots, j_d\right) + \vx_*\right),\ j_1,j_2,\ldots,j_d = 0,1,\ldots, m-1 \right\},
\]
where $\vx_* \in [0,1]^d$ is referred to as the \textit{relative offset} of $\mathcal{X}$ with respect to $V$. 
Alternatively, we may also first partition $V$ uniformly into $m\times\cdots \times m$ hypercubes as 
\[ 
\{V_i\} := \left\{ \vb + \dfrac{L}{m}(j_1,j_2,\ldots, j_d) + \left [0, \frac{L}{m}\right]^d,\ j_1,j_2,\ldots,j_d = 0,1,\ldots, m-1\right\},
\]
and then sampling one mesh point $\vx_i$ from each hypercube $V_i$ with offset $\frac{L}{m}\vx_*$ to generate the set $\mathcal{X}$. 
 
The trapezoidal rule over $V$ using the uniform mesh $\mathcal{X}$ is defined as
\begin{equation}\label{eqn:trapezoidal_rule}
        \mathcal{Q}_V(g, \mathcal{X}) = \dfrac{|V|}{|\mathcal{X}|}\sum_{\vx_i \in \mathcal{X}} g(\vx_i)  = \sum_{\vx_i \in \mathcal{X}}  |V_i|g(\vx_i) = \sum_{V_i} \mathcal{Q}_{V_i}(g, \{\vx_i\}), 
\end{equation}
which is equivalent to applying a single-point quadrature rule to each subdomain $V_i$ using node $\vx_i$. 
More precisely, $\mathcal{Q}_V(g, \mathcal{X})$ is an $m^d$-sized composition of the single-point quadrature rule 
\[
\mathcal{Q}_{[0,1]^d}(g, \{\vx_*\}) =g(\vx_*). 
\]

For sufficiently smooth functions, \cref{thm:em_general} gives the standard  Euler-Maclaurin formula that explicitly characterizes the quadrature error of a trapezoidal rule over a hypercube \cite{LynessMcHugh1970,BakerHodgson1971}. 
\begin{thm}[Standard Euler-Maclaurin formula]
\label{thm:em_general}
        Given a hypercube $V = [0, L]^d + \vb$ and an $m^d$-sized uniform mesh $\mathcal{X}$ in $V$ with relative offset $\vx_*$, for $g\in C^l(V)$, the quadrature error of the trapezoidal rule can be expressed as
        \begin{equation}\label{eqn:em_general}
                \mathcal{E}_V(g, \mathcal{X}) = \sum_{s=1}^{l-1} \dfrac{L^s}{m^s} \sum_{|\vbeta| = s} c_\vbeta(\vx_*) \int_V g^{(\vbeta)}(\vx) \ud \vx + \dfrac{L^l}{m^l} \sum_{|\vbeta| = l} \int_V h_{\vbeta,\vx_*} \left(m\frac{\vx - \vb}{L}\right)g^{(\vbeta)}(\vx)\ud \vx.
        \end{equation}
        The kernel function $h_{\vbeta, \vx_*}(\vx)$ is bounded and is periodic along each dimension with period $1$. The explicit form of $h_{\vbeta, \vx_*}(\vx)$ in one- and two-dimensional spaces can be found in \cite{BakerHodgson1971}, and its general form in $\mathbb{R}^d$ is studied in \cite{LynessMcHugh1970}. 
        The coefficient $c_\vbeta(\vx)$ is defined as 
        \begin{align*}
                c_\vbeta(\vx) 
                &=
                - \dfrac{B_{\beta_1}(x_1)}{\beta_1!}\dfrac{B_{\beta_2}(x_2)}{\beta_2!}\cdots \dfrac{B_{\beta_d}(x_d)}{\beta_d!},
        \end{align*}
        where $B_k(x)$ is the periodic Bernoulli polynomial of order $k$.
\end{thm}

Note  that when  $g(\vx)$ and its derivatives up to $(l-2)$th order satisfy the periodic boundary condition on $\partial V$, all the integrals of $g^{(\vbeta)}(\vx)$ in the formula \cref{eqn:em_general} vanish. 
Thus, we can prove in \cref{cor:euler_maclaurin} that, in this case, the quadrature error scale as $\Or(m^{-l})$. 
\begin{cor}[Standard Euler-Maclaurin formula for functions with periodic boundary condition]\label{cor:euler_maclaurin}If $g(\vx)$ and its derivatives up to $(l-2)$th order satisfy the periodic boundary condition on $\partial V$, all the integrals over $g^{(\vbeta)}(\vx)$ in \cref{eqn:em_general} vanish and the quadrature error satisfies, 
        \[
        \mathcal{E}_V(g, \mathcal{X}) = \Or(m^{-l}).
        \]
        Furthermore, if $g$ and its derivatives are smooth (i.e., $l = \infty$) and all satisfy the periodic boundary condition on $\partial V$, the quadrature error decays super-algebraically, i.e., faster than $\Or(m^{-s})$ with any $s > 0$.
\end{cor}

In the exchange and MP2 energy calculations, their integrands are all periodic and smooth everywhere except at a measure-zero set of points, and the standard Euler-Maclaurin formula above cannot be applied directly.
On the other hand, an accurate estimate of their quadrature errors need to account for both the function periodicity and non-smoothness.

\subsection{Generalized Euler-Maclaurin formula for a special class of non-smooth functions}\label{sec:improve_em}
The key idea used in the analysis below is that when partitioning $V$ into subdomains $\{V_i\}$, an integrand could be smooth in many of these subdomains, and the standard Euler-Maclaurin formula can be applied to the single-point quadrature rule \cref{eqn:trapezoidal_rule}  in each subdomain. 
Specifically, consider $g(\vx)$ in $V = [0,L]^d + \vb$ that is smooth everywhere except at subdomains $\{V_i\}_{i\in \mathcal{T}}$ indexed by $\mathcal{T}$. Define $V_\mathcal{T} = \bigcup_{i\in\mathcal{T}} V_i$. 
The quadrature error can then be split into two parts: 
\begin{equation}\label{eqn:error_split}
        \mathcal{E}_V(g,\mathcal{X})
        = 
        \sum_{i\in\mathcal{T}}\mathcal{E}_{V_i}(g,\{\vx_i\})
        +
        \sum_{i \not\in \mathcal{T}} \mathcal{E}_{V_i}(g,\{\vx_i\}). 
\end{equation}

In each $V_i = [0, \frac{L}{m}]^d + \vb_i$ with $i\not\in \mathcal{T}$, $g(\vx)$ is smooth and its quadrature error in $V_i$ can be formulated by applying \cref{thm:em_general} to the single-point quadrature rule $\mathcal{Q}_{V_i}(g, \{\vx_i\})$ as 
\[
\mathcal{E}_{V_i}(g,\{\vx_i\})
=
\sum_{s=1}^{l-1}\dfrac{L^s}{m^s}
\sum_{|\vbeta| = s} c_\vbeta(\vx_*) \int_{V_i} g^{(\vbeta)}(\vx) \ud \vx+ \dfrac{L^l}{m^l} \sum_{|\vbeta| = l} \int_{V_i} h_{\vbeta,\vx_*} \left(\frac{\vx - \vb_i}{L/m}\right)g^{(\vbeta)}(\vx)\ud \vx,
\]
where the factor $\frac{L}{m}$ comes from the edge length $\frac{L}{m}$ of $V_i$ and the mesh size $1$ of the quadrature, and the order $l$ could be arbitrarily large. 
Summing over all $V_i$ with $i\not\in \mathcal{T}$, we can write the smooth part of the overall error as 
\begin{equation}\label{eqn:error_split_smooth}
\sum_{i \not\in \mathcal{T}}\mathcal{E}_{V_i}(g,\{\vx_i\})
= 
\sum_{s=1}^{l-1}\dfrac{L^s}{m^s}
\sum_{|\vbeta| = s} c_\vbeta(\vx_*) \int_{V\setminus V_\mathcal{T}} g^{(\vbeta)}(\vx) \ud \vx
+
\dfrac{L^l}{m^l} \sum_{|\vbeta| = l} \int_{V\setminus V_\mathcal{T}} h_{\vbeta,\vx_*} \left(\frac{\vx - \vb}{L/m}\right)g^{(\vbeta)}(\vx)\ud \vx, 
\end{equation}
which uses the fact that $h_{\vbeta,\vx_*} \left(\frac{\vx - \vb}{L/m}\right) = h_{\vbeta,\vx_*} \left(\frac{\vx - \vb_i}{L/m}\right)$ by the periodicity of $h_{\vbeta, \vx_*}(\vx)$. 
This gives a \textit{partial Euler-Maclaurin formula} for a trapezoidal rule only over the subdomains where $g(\vx)$ is sufficiently smooth. 

Using this formula, it is possible to exploit the boundary conditions of $g(\vx)$ to estimate the smooth part of the quadrature error. 
For example, if $g^{(\vbeta)}(\vx)$ with $|\vbeta| < l$ is integrable in $V_\mathcal{T}$, the first integral above can be further split as 
\[
\int_{V\setminus V_\mathcal{T}} g^{(\vbeta)}(\vx) \ud \vx
= 
\int_{V} g^{(\vbeta)}(\vx) \ud\vx - \int_{V_\mathcal{T}} g^{(\vbeta)}(\vx) \ud\vx,
\]
where the first term vanishes if $g(\vx)$ and its derivatives are  periodic.  
Further using the boundedness of $h_{\vbeta, \vx_*}$, we could get a preliminary estimate of the overall quadrature error as 
\begin{align*}
        \mathcal{E}_V(g,\mathcal{X})
        \lesssim 
        \sum_{i\in\mathcal{T}}\mathcal{E}_{V_i}(g,\{\vx_i\}) + \sum_{s=1}^{l-1}\dfrac{L^s}{m^s}
        \sum_{|\vbeta| = s} |c_\vbeta(\vx_*)| \mathcal{I}_{V_\mathcal{T}}(|g^{(\vbeta)}|)
        + 
        \dfrac{L^l}{m^l} \sum_{|\vbeta| = l}  \mathcal{I}_{V \setminus V_\mathcal{T}}(|g^{(\vbeta)}|). 
\end{align*}

Based on the above idea of partial Euler-Maclaurin formula, we now prove \cref{thm:em_n_fraction} which gives a generalized Euler-Maclaurin formula for non-smooth functions in the fractional form.

\begin{thm}[Generalized Euler-Maclaurin formula for functions in the fractional form]
\label{thm:em_n_fraction}
        Consider $n$ smooth functions $\{f_i(\vx_1, \vx_2, \ldots, \vx_n)\}_{i=1}^n$ in $\mathbb{R}^d \times \cdots \times \mathbb{R}^d$. 
        For each $i = 1,\ldots, n$,  $f_i(\vx_1, \vx_2, \ldots, \vx_n)$ is analytic at $\vx_i = \bm{0}$ and scales as $\Or(|\vx_i|^{a_i})$ near $\vx_i = \bm{0}$.
        Define the integrand
        \[
        g(\vx_1, \vx_2, \ldots, \vx_n) = 
        \dfrac{f_1(\vx_1, \vx_2, \ldots, \vx_n)}{ (\vx_1^T M \vx_1)^{p_1}} 
        \dfrac{f_2(\vx_1, \vx_2, \ldots, \vx_n)}{ (\vx_2^T M \vx_2)^{p_2}} 
        \cdots 
        \dfrac{f_n(\vx_1, \vx_2, \ldots, \vx_n)}{  (\vx_n^T M \vx_n)^{p_n}},
                \]
        where $M\in\mathbb{R}^{d\times d}$ is a symmetric positive definite matrix and the exponent $p_i\in\mathbb{Z}$ satisfies $\gamma_i  = a_i - 2p_i  > -d$ 
        \REV{and $\gamma_i\in\mathbb{Z}$} for $i = 1,2,\ldots, n$. 
        Define $\gamma_\textup{min} = \min_i \gamma_i$. 
        
        The trapezoidal rule for $g(\vx_1, \vx_2, \ldots, \vx_n)$ over $V^{\times n}$ with $V= [-\frac12, \frac12]^d$ using an $m^{nd}$-sized uniform mesh $\mathcal{X}^{(n)} = \mathcal{X}_1\times \ldots \times \mathcal{X}_n$ with relative offset $\vx^{(n)}\in [0,1]^{nd}$ has quadrature error 
        \[
        \mathcal{E}_{V^{\times n}}(g, \mathcal{X}^{(n)}) 
        = \sum_{s=1}^{d+\gamma_\textup{min}-1} \dfrac{1}{m^s} 
        \left(
        \sum_{|\vbeta| = s} c_\vbeta(\vx^{(n)}) \int_{V^{\times n}} g^{(\vbeta)}\ud \vx_1\ldots \ud\vx_n
        \right) 
        + \Or\left(\dfrac{\ln m}{m^{d+\gamma_\textup{min}}}\right).
        \]
        Here $\vbeta$ is an $nd$-dimensional multi-index. 
        When $(\vx_1, \ldots, \vx_n)$ with some $\vx_i = \bm{0}$ is a quadrature node, $g(\vx_1,\ldots, \vx_n)$ is indeterminate and set to $0$ in the numerical quadrature. 

        The prefactor of the $\Or(m^{-(d+\gamma_\textup{min})}\ln m)$ remainder is bounded by an $\Or(1)$ constant that depends on the upperbounds of functions 
         $|\partial^{\valpha}f_i|/|\vx_i|^{a_i-|\valpha|}, 0 \leqslant |\valpha| \leqslant  a_i$ and $|\partial^{\valpha}f_i|, a_i < |\valpha| \leqslant d+\gamma_\textup{min}$ with $(\vx_1,\ldots, \vx_n)\in V^{\times n}$. 
\end{thm}
\begin{proof}   
        Let $\{V_i\}$ be the $m^d$-sized uniform partitioning of $V$, and  $\mathcal{T}$ be the indices of all subdomains $V_i$ intersecting with $\vx = \bm{0}$.
        Let $V_\mathcal{T} = \bigcup_{i\in\mathcal{T}}V_i$. 
        When $m$ is even, there are $2^d$ subdomains with $\vx = \bm{0}$ on their vertices and \REV{$V_\mathcal{T} = [-\frac{1}{m}, \frac{1}{m}]^d$}. 
        When $m$ is odd, there is one subdomain with $\vx=\bm{0}$ at its center and \REV{$V_\mathcal{T} = [-\frac{1}{2m}, \frac{1}{2m}]^d$}. 
        Correspondingly, $V^{\times n}$ is partitioned into $m^{nd}$ subdomains $\{V_{j_1} \times \cdots \times V_{j_n}\}$ in $\mathbb{R}^{nd}$. 
        Multi-indices $(j_1,\ldots, j_n)$ of the subdomains where $g(\vx_1,\ldots, \vx_n)$ is non-smooth are collected as 
        \[
        \mathcal{T}_n = \{(j_1,\ldots, j_n): \exists 1 \leqslant k\leqslant n\ \text{ s.t. }\ j_k \in \mathcal{T}\}. 
        \]
        Following the idea in \cref{eqn:error_split} and \cref{eqn:error_split_smooth}, the quadrature error can be split as
        \begin{align}
                \mathcal{E}_{V^{\times n}}(g,\mathcal{X}^{(n)})
                & = 
                \sum_{(j_1,\ldots, j_n)\in\mathcal{T}_n}
                \mathcal{E}_{V_{j_1}\times \cdots \times V_{j_n}}(g,\{(\vx_{j_1},\cdots, \vx_{j_n})\})
                + 
                \sum_{(j_1,\ldots, j_n)\not\in \mathcal{T}_n} 
                \mathcal{E}_{V_{j_1}\times \cdots \times V_{j_n}}(g,\{(\vx_{j_1},\cdots, \vx_{j_n})\})
                \nonumber\\
                & = 
                \mathcal{E}_{V^{\times n}_{\mathcal{T}_n}}(g, \mathcal{X}^{(n)}_{\mathcal{T}_n})
                + 
                \sum_{s=1}^{d+\gamma_\text{min}-1} \dfrac{A_s}{m^s} + \dfrac{R_{d+\gamma_\text{min}}}{m^{d+\gamma_\text{min}}},
                \label{eqn:error_split_n}
        \end{align}
        where $V^{\times n}_{\mathcal{T}_n} = \bigcup_{(j_1, \ldots, j_n)\in \mathcal{T}_n}\left(V_{j_1}\times \cdots \times V_{j_n}\right)$, $\mathcal{X}^{(n)}_{\mathcal{T}_n} 
        = \bigcup_{(j_1, \ldots, j_n)\in \mathcal{T}_n}\{(\vx_{j_1}, \ldots, \vx_{j_n})\}$, and 
        \begin{align*}
                A_s 
                & = \sum_{|\vbeta| = s} c_\vbeta(\vx_*^{(n)}) \int_{V^{\times n} \setminus V^{\times n}_{ \mathcal{T}_n}} g^{(\vbeta)}(\vx_1, \ldots, \vx_n) \ud \vx_1\ldots \ud \vx_n, 
                \\
                R_{d+\gamma_\text{min}}
                & 
                = 
                \sum_{|\vbeta| = d+\gamma_\text{min}}  \int_{V^{\times n} \setminus V^{\times n}_{\mathcal{T}_n}}  h_{\vbeta, \vx_*^{(n)}} \left( m((\vx_1,\ldots, \vx_n)- (\vb, \ldots, \vb))\right) g^{(\vbeta)}(\vx_1, \ldots, \vx_n) \ud \vx_1\ldots \ud \vx_n. 
        \end{align*}
        Here, $\vb = (-\frac12, \ldots, -\frac12)$ denotes the corner of hypercube $V$. 
        
        The main proof below involves four parts: 
        \begin{enumerate}
                \item The derivatives $g^{(\vbeta)}$ in $A_s$ with $|\vbeta| = s \leqslant d+\gamma_\text{min}-1$ are integrable in $V^{\times n}$ and thus $A_s$ can be split as 
                \begin{equation}\label{eqn:As_split}
                        A_s  = \sum_{|\vbeta| = s} c_\vbeta(\vx_*^{(n)}) \int_{V^{\times n}} g^{(\vbeta)}\ud \vx_1\ldots \ud \vx_n 
                        - 
                        \sum_{|\vbeta| = s} c_\vbeta(\vx_*^{(n)}) \int_{V^{\times n}_{\mathcal{T}_n}} g^{(\vbeta)}\ud \vx_1\ldots \ud \vx_n.  
                \end{equation}
                \item Estimate of the integrals in the second part of the above splitting of $A_s$ as 
                \[
                \int_{V^{\times n}_{\mathcal{T}_n}} g^{(\vbeta)}\ud \vx_1\ldots \ud \vx_n   = \Or\left( \frac{1}{m^{d+\gamma_\text{min} - |\vbeta|} }\right). 
                \]
                \item Estimate of the remainder term $R_{d+\gamma_\text{min}}$ as 
                \[
                R_{d+\gamma_\text{min}} = \Or(\ln m). 
                \]
                \item Estimate of the quadrature error in the volume elements with non-smooth integrands as 
                \[
                \mathcal{E}_{V^{\times n}_{\mathcal{T}_n}}(g, \mathcal{X}^{(n)}_{\mathcal{T}_n}) = 
                \Or\left( \frac{1}{m^{d+\gamma_\text{min}}}\right). 
                \] 
        \end{enumerate}
        Combining these four results with \cref{eqn:error_split_n} gives the final formula of the theorem.
        
        We first introduce some basic tools.
        The derivative $g^{(\vbeta)}$ can be expanded as a linear combination of terms 
        \begin{equation}\label{eqn:expansion_derivative}
                g^{(\vbeta_1, \ldots, \vbeta_{n})}(\vx_1,\ldots, \vx_n)
                =
                \prod_{i=1}^n
                \dfrac{\partial^{|\vbeta_i|}}{\partial (\vx_1,\ldots, \vx_n)^{\vbeta_i}}
                \dfrac{f_i(\vx_1,\ldots, \vx_n)}{ (\vx_i^T M \vx_i)^{p_i}},
        \end{equation}
        with $\vbeta_1 + \cdots + \vbeta_{n} = \vbeta$.  
        Note that $\vbeta_i$'s are all $nd$-dimensional multi-indices. 
        Using the condition that $f_i = \Or(|\vx_i|^{a_i})$ near $\vx_i = \bm{0}$, it can be shown that near $\vx_i = \bm{0}$, 
        \begin{align}
                \left|
                \partial^{\vbeta_i}\dfrac{f_i}{ (\vx_i^T M \vx_i)^{p_i}}
                \right|
                & = 
                \left|
                \sum_{\valpha_1 + \valpha_2 = \vbeta_i} {\vbeta_i \choose \valpha_1} \partial^{\valpha_1}  f_i(\vx_1,\ldots, \vx_n) \partial^{\valpha_2} \dfrac{1}{(\vx_i^TM\vx_i)^{p_i}}
                \right|
                \nonumber \\
                & 
                = 
                \sum_{\valpha_1 + \valpha_2 = \vbeta_i}  \Or( |\vx_i|^{\max\{a_i - |\valpha_1|, 0\}}|\vx_i|^{-2p_i - |\valpha_2|})
                = \Or(|\vx_i|^{\gamma_i - |\vbeta_i|}). 
                \label{eqn:estimate_derivative}
        \end{align}
        This estimate can be extended to all $(\vx_1,\ldots, \vx_n)\in V^{\times n}$ by the function smoothness outside $\vx_i =\bm{0}$.  
        
        In the following discussions, we detail the four parts of the proof. 
        \paragraph{Part 1: Integrability of $g^{(\vbeta)}$ in $V^{\times n}$ with $|\vbeta|\leqslant d+\gamma_\text{min}-1$}
        Since $g^{(\vbeta)}$ is a linear combination of $g^{(\vbeta_1, \ldots, \vbeta_{n})}$ in \cref{eqn:expansion_derivative} with $\vbeta_1+ \ldots + \vbeta_{n} = \vbeta$, we only need to prove the integrability of each $g^{(\vbeta_1, \ldots, \vbeta_{n})}$ in $V^{\times n}$. 
        According to \cref{eqn:estimate_derivative}, we have 
        \[
                \int_{V^{\times n}} 
                \ud\vx_1\ldots \ud\vx_n
                |g^{(\vbeta_1, \ldots, \vbeta_{n})}|
                \lesssim
                \int_{V^{\times n}}
                \ud\vx_1\ldots \ud\vx_n
                |\vx_1|^{\gamma_1 - |\vbeta_1|}
                \ldots  
                |\vx_n|^{\gamma_n - |\vbeta_n|}.
        \]
        Since $\gamma_i - |\vbeta_i| \geqslant \gamma_\text{min}- |\vbeta| \geqslant -d +1$, the right hand side of the above inequality is finite and thus $g^{(\vbeta_1, \ldots, \vbeta_{n})}$ is integrable in $V^{\times n}$. 
        
        \paragraph{Part 2: Estimate of the integrals in $A_s$ over the volume elements with non-smooth integrands}
        Note that 
        \[
                V^{\times n}_{\mathcal{T}_n} = \bigcup_{k=1}^n (V \times \ldots \times \underbrace{V_\mathcal{T}}_{\text{for}\ \vx_k} \times \ldots \times V). 
                \]
        Since $g^{(\vbeta)}$ is a linear combination of $g^{(\vbeta_1, \ldots, \vbeta_{n})}$ with $\Or(1)$ coefficients, we instead estimate the integral of each $g^{(\vbeta_1, \ldots, \vbeta_{n})}$ with $\vbeta_1+ \ldots + \vbeta_{n} = \vbeta$ where $|\vbeta| =s \leqslant d+\gamma_\text{min}-1$. 
        The integral of $|g^{(\vbeta_1, \ldots, \vbeta_{n})}|$ over $V^{\times n}_{\mathcal{T}_n}$ in the splitting \cref{eqn:As_split} of $A_s$ can be first bounded as
        \[
        \int_{ V^{\times n}_{\mathcal{T}_n}} 
        \left|g^{(\vbeta_1, \ldots, \vbeta_n)}\right|
        \ud \vx_1\ldots \ud \vx_n       
        \leqslant 
        \sum_{k=1}^n
        \int_{V\times \ldots \times V_\mathcal{T} \times \ldots \times V} 
        \left|g^{(\vbeta_1, \ldots, \vbeta_n)}  \right|\ud \vx_1\ldots \ud \vx_n.
        \]
        Without loss of generality let us consider $k=1$. The corresponding term of the right hand side above could be further estimated as,
        \begin{align*}
                & 
                \int_{V_\mathcal{T}\times V^{\times (n-1)}} 
                \left|g^{(\vbeta_1, \ldots, \vbeta_n)}  \right|\ud \vx_1\ldots \ud \vx_n
                \\
                \lesssim & 
                \int_{V_\mathcal{T}\times V^{\times (n-1)}} \ud\vx_1\ldots \ud\vx_n 
                |\vx_1|^{\gamma_1 - |\vbeta_1|}\cdots |\vx_n|^{\gamma_n - |\vbeta_n|}
                \\
                \lesssim & 
                \int_{B(\bm{0}, Cm^{-1})\times V^{\times (n-1)}} \ud\vx_1\ldots \ud\vx_n
                |\vx_1|^{\gamma_1 - |\vbeta_1|}\cdots |\vx_n|^{\gamma_n - |\vbeta_n|}
                \\
                = & 
                \int_{B(\bm{0}, Cm^{-1})} |\vx_1|^{\gamma_1 - |\vbeta_1|}\ud\vx_1
                \int_{V^{\times (n-1)}} \ud\vx_2\ldots \ud\vx_n
                |\vx_2|^{\gamma_2 - |\vbeta_2|}\cdots |\vx_n|^{\gamma_n - |\vbeta_n|}
                \\
                = & 
                \Or\left(
                \dfrac{1}{m^{d+\gamma_1 - |\vbeta_1|}}
                \right),
        \end{align*}
        where $B(\vz, r)$ denotes a ball in $\mathbb{R}^d$ centered at $\vz$ with radius $r$, and the second inequality uses the fact that 
        \[
        V_{\mathcal{T}} \subset B\left(\bm{0}, Cm^{-1}\right), 
        \]
        with some $\Or(1)$ constant $C$. 
        Thus, we have 
        \begin{align}
                \left|
                \int_{ V^{\times n}_{\mathcal{T}_n}} 
                g^{(\vbeta)}\ud \vx_1\ldots \ud \vx_n
                \right|
                & \lesssim
                \sum_{\vbeta_1 + \ldots + \vbeta_n = \vbeta}
                \left|
                \int_{ V^{\times n}_{\mathcal{T}_n}} 
                g^{(\vbeta_1, \ldots, \vbeta_n)}\ud \vx_1\ldots \ud \vx_n
                \right|
                \nonumber \\
                & =
                \sum_{k=1}^n
                \Or\left(
                \dfrac{1}{m^{d+\gamma_k - |\vbeta_k|}}
                \right)
                = \Or\left(
                \dfrac{1}{m^{d+\gamma_\text{min} - |\vbeta|}}
                \right). 
                \label{eqn:bound_gbeta2}
        \end{align}
        
        \paragraph{Part 3: Estimate of the remainder term $R_{d+\gamma_\text{min}}$}
        To estimate each integral over $g^{(\vbeta)}$ in the remainder term, we first note that $h_{\vbeta, \vx_*^{(n)}}$ is bounded (see \cref{thm:em_general}) and we have 
        \begin{align}
                & 
                \left|
                \int_{V^{\times n} \setminus V^{\times n}_{\mathcal{T}_n}}  h_{\vbeta, \vx_*^{(n)}}(\cdots)g^{(\vbeta)}(\vx_1, \ldots, \vx_n) \ud \vx_1\ldots \ud \vx_n
                \right|
                \nonumber \\
                \lesssim 
                & 
                \int_{V^{\times n} \setminus V^{\times n}_{\mathcal{T}_n}}         \left| g^{(\vbeta)}(\vx_1, \ldots, \vx_n)       \right|\ud \vx_1\ldots \ud \vx_n
                \nonumber\\
                \lesssim 
                & 
                \sum_{\vbeta_1 + \ldots + \vbeta_n = \vbeta}
                \int_{V^{\times n} \setminus V^{\times n}_{\mathcal{T}_n}}
                |\vx_1|^{\gamma_1 - |\vbeta_1|}\ldots 
                |\vx_n|^{\gamma_n - |\vbeta_n|}
                \ud\vx_1\ldots \ud\vx_n.
                \label{eqn:bound_remainder}
        \end{align}
        It is thus sufficient to estimate the last integral above with $|\vbeta_1|+ \ldots + |\vbeta_{n+m}| = |\vbeta| = d+\gamma_\text{min}$.  
        There are two scenarios to consider. 
        \begin{enumerate}
                \item $\gamma_i - |\vbeta_i| = -d$ for some $i$.  
                Without loss of generality consider $i = 1$. Since $|\vbeta_1|\leqslant |\vbeta| = d+\gamma_\text{min}$ and $\gamma_1 \geqslant \gamma_\text{min}$, we have $\gamma_1 = \gamma_\text{min}$, $|\vbeta_1| = |\vbeta|$ and $\vbeta_i = \bm{0}$ for all $i \neq 1$.
                Then the integral on the right hand side of \cref{eqn:bound_remainder} can be bounded by
                \begin{align*}
                        &\quad \int_{V^{\times n} \setminus V^{\times n}_{\mathcal{T}_n}} 
                        |\vx_1|^{-d}|\vx_2|^{\gamma_2}
                        \cdots 
                        |\vx_n|^{\gamma_n}
                        \ud \vx_1\ldots \ud \vx_n\\
                        & \leqslant
                        \int_{(V\setminus V_\mathcal{T})\times V^{\times (n-1)}}  
                        |\vx_1|^{-d}|\vx_2|^{\gamma_2}
                        \cdots 
                        |\vx_n|^{\gamma_n}
                        \ud \vx_1\ldots \ud \vx_n
                        \\
                        &
                        \lesssim
                        \int_{V\setminus V_\mathcal{T}} |\vx_1|^{-d}  \ud \vx_1
                        \\
                        & 
                        \lesssim\int_{1/m}^{1}  r^{-d + d-1} \ud r 
                        = \Or( \ln m),
                \end{align*}
                where the first inequality uses the fact that $V^{\times n} \setminus V^{\times n}_{\mathcal{T}_n}\subset (V\setminus V_\mathcal{T})\times V^{\times (n-1)}$.
                
                \item $\gamma_i - |\vbeta_i| \geqslant -d+1$ for $i = 1,\ldots, n$. In this case, $\int_{V^{\times n}}          |\vx_1|^{\gamma_1 - |\vbeta_1|}
                \cdots 
                |\vx_n|^{\gamma_n - |\vbeta_n|} \ud \vx_1\ldots \ud \vx_n$ is finite and thus 
                \begin{align*}
                        \int_{V^{\times n} \setminus V^{\times n}_{\mathcal{T}_n}}                 
                        |\vx_1|^{\gamma_1 - |\vbeta_1|}
                        \cdots 
                        |\vx_n|^{\gamma_n - |\vbeta_n|}
                        \ud \vx_1\ldots \ud \vx_n
                        & 
                        = \Or(1).
                \end{align*}
        \end{enumerate}
        Combining the above estimation with \cref{eqn:bound_remainder}, we obtain
        \[
        R_{d+\gamma_\text{min}} = \Or(\ln m).
        \]
        
        \paragraph{Part 4: Estimate of the quadrature error in the volume elements with non-smooth integrand}
        Using a similar discussion as in \cref{eqn:bound_gbeta2}, the integral involved in $\mathcal{E}_{V^{\times n}_{\mathcal{T}_n}}(g, \mathcal{X}^{(n)}_{\mathcal{T}_n})$ can be first estimated as
        \begin{align*}
                \left|\mathcal{I}_{V^{\times n}_{\mathcal{T}_n}}(g)\right|
                =
                \Or\left(\dfrac{1}{m^{\gamma_\text{min} + d} }\right). 
        \end{align*}
        The quadrature part can be estimated as 
        \begin{align*}
                \left|
                \mathcal{Q}_{V^{\times n}_{\mathcal{T}_n}}(g, \mathcal{X}^{\times n}_{\mathcal{T}_n})
                \right|
                & \leqslant
                \dfrac{1}{m^{nd}}
                \sum_{(j_1,\ldots, j_n)\in\mathcal{T}_n}
                \left|
                g(\vx_{j_1},\cdots, \vx_{j_n})
                \right|
                \\
                & \leqslant 
                \dfrac{1}{m^{nd}}
                \sum_{i=1}^n
                \sum_{j_i \in \mathcal{T},\  j_1,\ldots,j_{i-1}, j_{i+1},\ldots  j_n}
                \left|
                g(\vx_{j_1},\cdots, \vx_{j_n})
                \right|
                \\
                & \lesssim
                \dfrac{1}{m^{nd}}
                \sum_{i=1}^n
                \sum_{j_i \in \mathcal{T},\  j_1,\ldots,j_{i-1}, j_{i+1},\ldots  j_n}
                \dfrac{1}{m^{\gamma_i}}
                \\
                & = 
                \dfrac{1}{m^{d}}
                \sum_{i=1}^n
                \sum_{j_i \in \mathcal{T}}
                \dfrac{1}{m^{\gamma_i}} = \Or\left(\dfrac{1}{m^{\gamma_\text{min}+ d} }\right).
        \end{align*}
        Thus, the quadrature error in all the volume elements with non-smooth integrands can be estimated as
        \begin{equation*}
                \mathcal{E}_{V^{\times n}_{\mathcal{T}_n}}(g, \mathcal{X}^{\times n}_{\mathcal{T}_n}) \le \left|\mathcal{I}_{V^{\times n}_{\mathcal{T}_n}}(g)\right|+\left|
                \mathcal{Q}_{V^{\times n}_{\mathcal{T}_n}}(g, \mathcal{X}^{\times n}_{\mathcal{T}_n})
                \right|=  \Or\left(\dfrac{1}{m^{\gamma_\text{min} + d} }\right). 
        \end{equation*}
        
        In the analysis of the four parts above, it can be noted that the prefactor of all the $\Or(m^{-(d+\gamma_\text{min})})$ and $\Or(\ln m)$ estimates depends on the prefactor of the estimate in \cref{eqn:estimate_derivative}, i.e., 
        \[
        \left|
        \partial^{\vbeta_i}\dfrac{f_i}{ (\vx_i^T M \vx_i)^{p_i}}
        \right|
        = \Or(|\vx_i|^{\gamma_i - |\vbeta_i|}), \quad (\vx_1,\ldots, \vx_n)\in V^{\times n},
        \]
        which is further proportional to the
        upper bound of all the functions $|\partial^{\valpha}f_i|/|\vx_i|^{a_i-|\valpha|}, 0 \leqslant |\valpha| \leqslant  a_i$ and $|\partial^{\valpha}f_i|, a_i < |\valpha| \leqslant d+\gamma_\text{min}$. 
        This concludes the characterization of the prefactor in the $\Or(m^{-(d+\gamma_\text{min})}\ln m)$ remainder term in the obtained Euler-Maclaurin formula. 
\end{proof}

\REV{
Although not directly related to the application in this paper, this quadrature error analysis result can be generalized to the case where $\gamma_i$ is non-integer.}
Based on \cref{thm:em_n_fraction}, \cref{cor:em_n_fraction} characterizes the quadrature error for fractional-form functions that also satisfy periodic boundary condition on $\partial V^{\times n}$. 

\begin{cor}[Generalized Euler-Maclaurin formula with periodic boundary conditions]
\label{cor:em_n_fraction}
        Under the setting of \cref{thm:em_n_fraction}, if $g(\vx_1,\ldots, \vx_n)$ and its derivatives also satisfy the periodic boundary condition on $\partial V^{\times n}$, all the integrals of $g^{(\vbeta)}$ in the derived Euler-Maclaurin formula vanish and 
        \[
        \mathcal{E}_{V^{\times n}}(g, \mathcal{X}^{(n)})  = \Or\left(\dfrac{\ln m}{m^{d+\gamma_\text{min}}}\right). 
        \]
\end{cor}

\REV{
\begin{rem}[Quadrature errors inside and outside the subdomains with singularity]
        \label{rem:error_splitting}
        According to \cref{eqn:error_split_n}, the overall quadrature error can be split into the error in the volume elements $V^{\times n}_{\mathcal{T}^n}$ 
        that contain the singular point $\vx = \bm{0}$ and the error in remaining volume elements. 
        From the estimate in \cref{thm:em_n_fraction}, these two parts contribute equally to the overall $\Or(m^{-(d+\gamma_\text{min})}\ln m)$ quadrature error for periodic functions in \cref{cor:em_n_fraction}. 
        For finite-size error corrections that compensate the omitted integral at the Coulomb singularity (such as the structure factor interpolation method mentioned in Introduction), 
        only the first part of the quadrature error above is corrected while the second part remains. 
        As a result, the overall quadrature error in general will not be reduced asymptotically. 
\end{rem}
}

\subsection{Numerical results}\label{sec:numer_em}
To demonstrate the sharpness of our error estimate in \cref{cor:em_n_fraction}, we consider a set of compactly supported functions listed in \cref{tab:testfunction} that are of the fractional form discussed in \cref{thm:em_n_fraction}. 
\cref{fig:cor_em} plots the numerical quadrature errors for these example functions by trapezoidal rules. 
The asymptotic scaling of these numerical results  is consistent with the analytic estimate in \cref{tab:testfunction} according to \cref{cor:em_n_fraction}.

\begin{table}[htbp]
        \centering
        \caption{Example functions for the numerical quadrature calculations in \cref{fig:cor_em}.
        The domain of integration for each variable  is $[-\frac12, \frac12]^d$. 
        Function $H(\vx)$ is the localizer defined in \cref{eqn:localizer} and $M$ equals $\text{diag}(10,1,0.1)$ for $d = 3$ and $\text{diag}(10,0.1)$ for $d = 2$.   
        The estimate of the quadrature error scaling is obtained according to \cref{cor:em_n_fraction}. 
        \label{tab:testfunction}}
        \begin{tabular}{l|c|c|l|c}
                \toprule
                 \multicolumn{1}{c}{} & \multicolumn{1}{c}{function form}                                                                                                                              & \multicolumn{1}{c}{dimension }                         &                \multicolumn{1}{c}{parameters}                                              & error scaling\\
                \midrule[0.2pt]
                $f_1(\vx)$         & \multirow{4}{*}{$\frac{H(\vx)\vx^\valpha}{(\vx^T M \vx)^p}$}         & $d=2$               & $\valpha=(0,2), p=1$                                              & $m^{-2}\ln m$          \\
                $f_2(\vx)$         &                                                                                                                                                        & $d=2$               & $\valpha=(0,4), p=2$                                              & $m^{-2}\ln m$          \\
                $f_3(\vx)$         &                                                                                                                                                        & $d=3$               & $\valpha=(0,0,0), p=1$                                            & $m^{-1}\ln m$          \\
                $f_4(\vx)$         &                                                                                                                                                        & $d=3$               & $\valpha=(0,2,2), p=1$                                            & $m^{-5}\ln m$          \\
                \midrule[0.2pt]
                $f_5(\vx_1,\vx_2)$ & 
                \multirow{2}{*}[-1ex]{$\frac{H(\vx_1)\vx_1^{\valpha_1}}{(\vx_1^T M \vx_1)^{p_1}}\frac{H(\vx_2)\vx_2^{\valpha_2}}{(\vx_2^T M \vx_2)^{p_2}}e^{-|\vx_1 + \vx_2|^2}$} 
                & $d=2$ & 
                \makecell[l]{$\valpha_1 = (0,2)$, $p_1=1$, \\\quad $\valpha_2 = (2,0)$, $p_2=1$} & $m^{-2}\ln m$          \\
                $f_6(\vx_1,\vx_2)$ &                                          & $d=2$ & 
                \makecell[l]{$\valpha_1 = (0,2)$, $p_1=1/2$,\\\quad $\valpha_2 = (2,0)$, $p_2=1/2$}    & $m^{-3}\ln m$         
                \\
                \bottomrule
        \end{tabular}
\end{table}

\begin{figure}[htbp]
        \centering
        \subfloat
        {
                \includegraphics[width=0.31\textwidth]{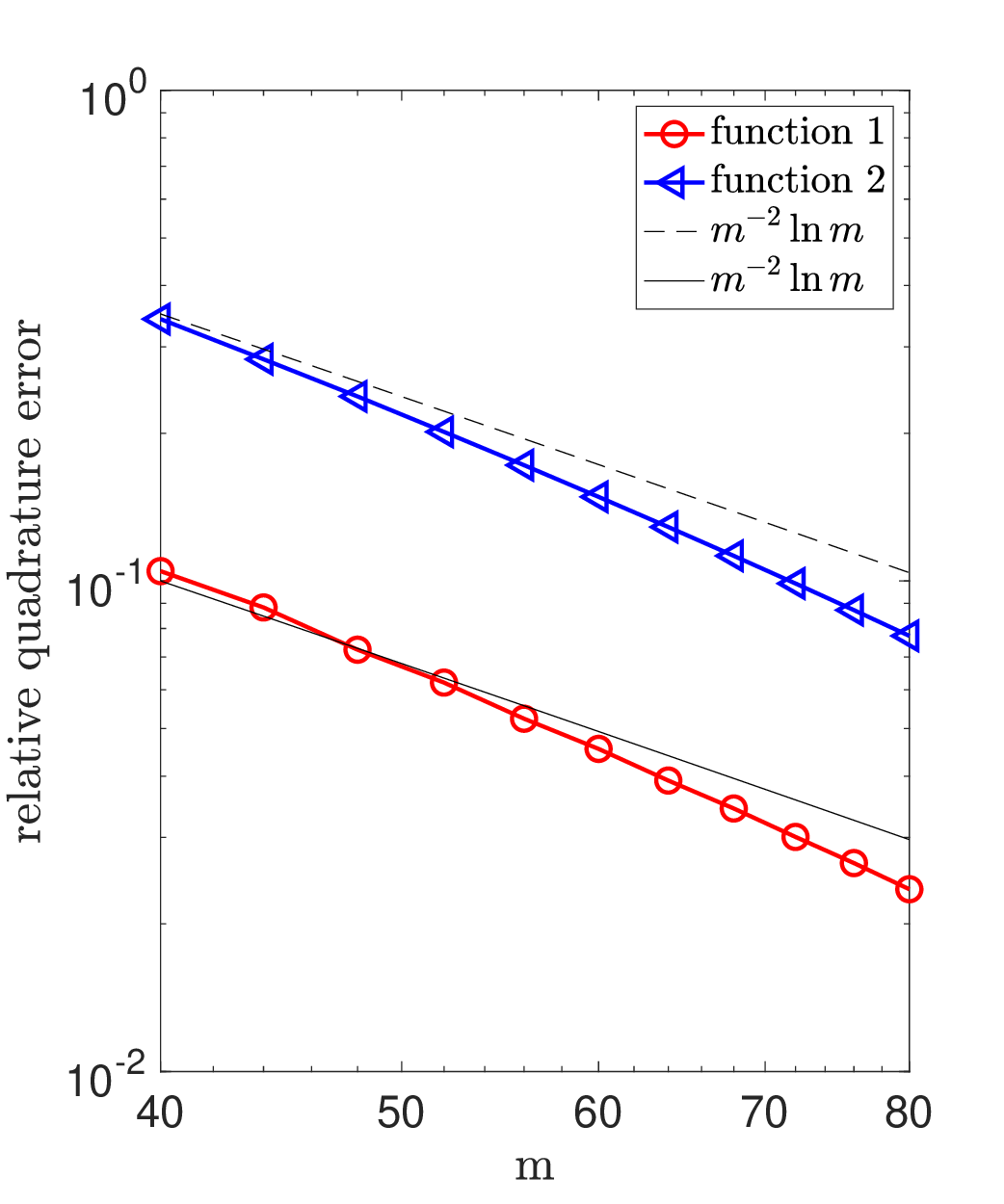}
        }
        \subfloat
        {
                \includegraphics[width=0.31\textwidth]{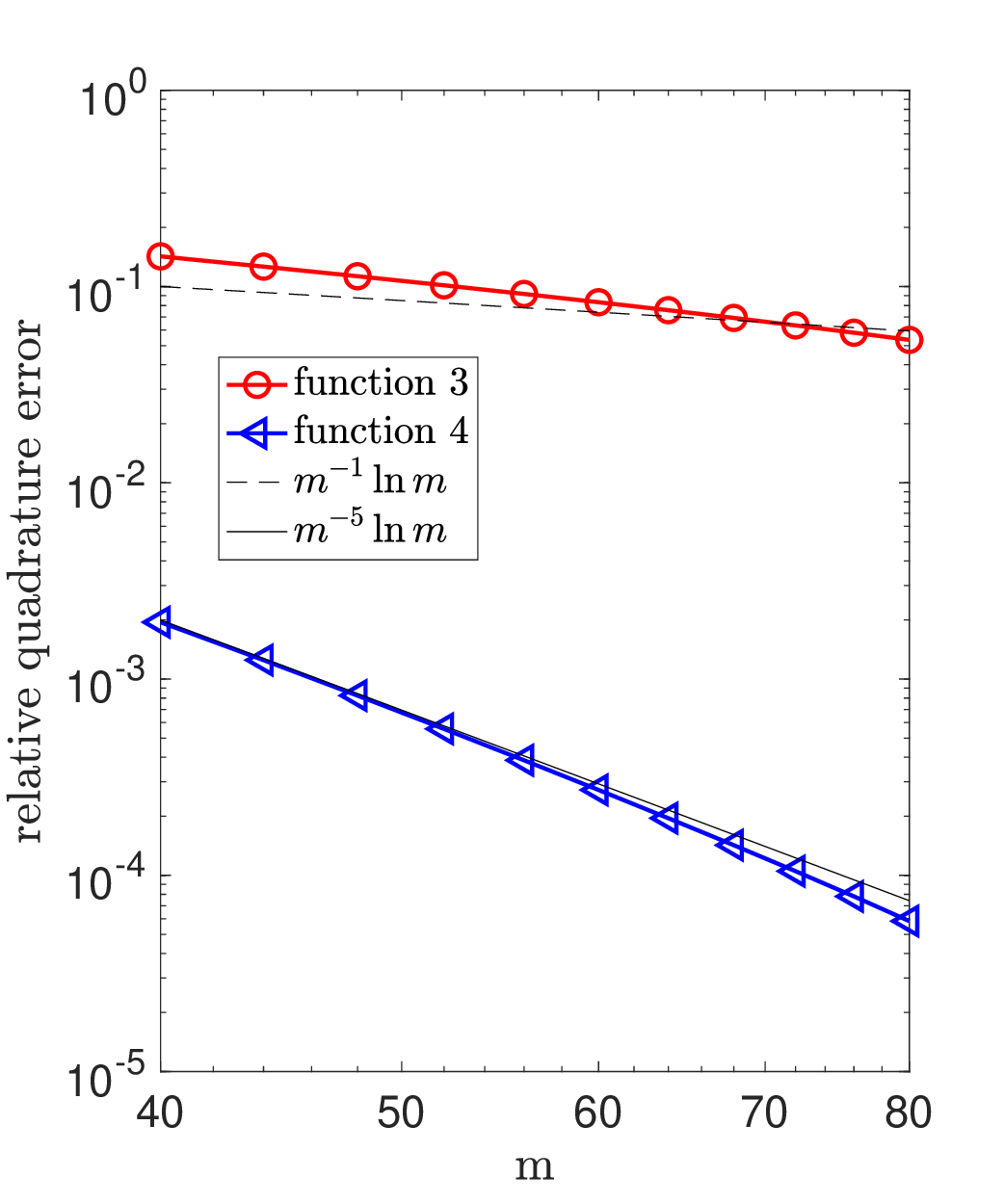}
        }
        \subfloat
        {
                \includegraphics[width=0.31\textwidth]{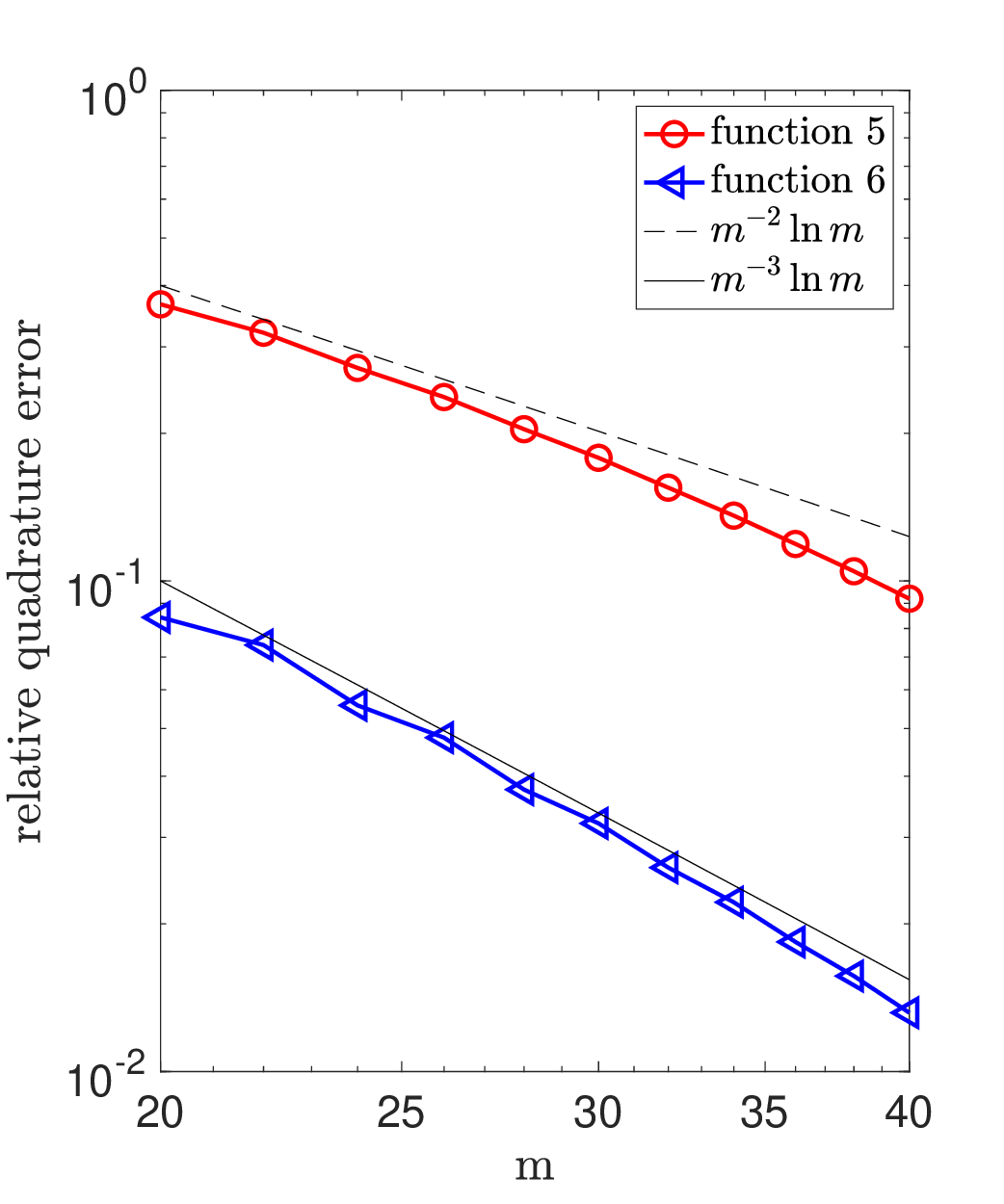}
        }
        \captionsetup{singlelinecheck=off}
        \caption[]{
                Relative quadrature errors of the example functions in \cref{tab:testfunction}. 
                For each test, a $\Gamma$-centered uniform mesh is used with $m$ points along each dimension in the integration domain.
                Reference integral values are computed with a sufficiently large $m$. 
                \label{fig:cor_em}
        }
\end{figure}

\section{Finite-size errors in the Hartree, potential, and kinetic energies}\label{appendix}
For completeness, we analyze the finite-size errors in the Hartree-Fock theory other than the Fock exchange term, namely the Hartree, potential, and kinetic energies. Unlike the Fock exchange energy and the MP2 correlation energy, the analysis of these terms does not involve singular integrands, and it is sufficient to analyze the finite-size errors using the standard Euler-Maclaurin formula in \cref{cor:euler_maclaurin}.

The Hartree energy with a finite MP mesh $\mathcal{K}$ can be computed as
\begin{align*}
        E_\text{h}(N_\vk)
        &  = \dfrac{1}{N_\vk} \sum_{ij}\sum_{\vk_i\vk_j\in\mathcal{K}} 
        \braket{i\vk_i,j \vk_j |i\vk_i, j\vk_j}
        \\&= 
        \dfrac{1}{|\Omega|}
        \xsum_{\vG\in\mathbb{L}^*}
        \frac{4\pi}{\abs{\vG}^2}  
        \left|
        \dfrac{1}{N_\vk}\sum_{i \vk_i}
        \hat{\varrho}_{i\vk_i, i\vk_i}(\mathbf{G}) 
        \right|^2
        \\&= 
        \dfrac{1}{|\Omega|}
        \xsum_{ \vG\in\mathbb{L}^*}
        \frac{4\pi}{\abs{\vG}^2}  
        \left|
        \int_{\Omega}\ud\vr e^{-\I\vG\cdot \vr}
        \dfrac{1}{N_\vk}\sum_{i \vk_i}
        |u_{i\vk_i}(\vr)|^2             \right|^2
        \\&= 
        \dfrac{1}{|\Omega|}
        \xsum_{\vG\in\mathbb{L}^*}
        \frac{4\pi}{\abs{\vG}^2}  
        \left|
        \int_{\Omega}\ud\vr e^{-\I\vG\cdot \vr}
        \rho_{N_\vk}(\vr)               \right|^2
        \\&= 
        \dfrac{1}{|\Omega|}
        \xsum_{\vG\in\mathbb{L}^*}
        \frac{4\pi}{\abs{\vG}^2}  
        |\hat{\rho}_{N_\vk}(\vG)|^2.
\end{align*}
Here $\rho_{N_\vk}(\vr) = \frac{1}{N_\vk}\sum_{i\vk_i}|u_{i\vk_i}(\vr)|^2$ (and its Fourier transform $\hat{\rho}_{N_\vk}(\vG)$) denotes the electron density obtained from the Hartree-Fock calculation. 
The Hartree energy in the TDL thus can be written as
\[
E_\text{h}^{\text{TDL}}=               \dfrac{1}{|\Omega|}
\xsum_{\vG\in\mathbb{L}^*}
\frac{4\pi}{\abs{\vG}^2}  
|\hat{\rho}^{\text{TDL}}(\vG)|^2.
\]
Therefore the finite-size error $E_\text{h}^\text{TDL} - E_\text{h}(N_\vk)$ only comes from the finite-size error of the electron density.
Following the same assumption used throughout the paper that all HF orbitals can be evaluated exactly at any $\vk\in \Omega^*$, the approximation $\hat{\rho}^\text{TDL}(\vG) \approx \hat{\rho}_{N_\vk}(\vG)$ for each $\vG \in \mathbb{L}^*\setminus\{\bm{0}\}$ can be treated as a numerical quadrature and the quadrature error is
\begin{align*}
        \hat{\rho}^\text{TDL}(\vG) - \hat{\rho}_{N_\vk}(\vG)
        =       
        \int_{\Omega^*}\ud\vk_i 
        \left(
        \sum_{i}
        \hat{\varrho}_{i\vk_i, i\vk_i}(\mathbf{G}) 
        \right)
        - 
        \dfrac{1}{N_\vk}\sum_{\vk_i}
        \left(
        \sum_i
        \hat{\varrho}_{i\vk_i, i\vk_i}(\mathbf{G}) 
        \right).
\end{align*}
Note that $\sum_i\hat{\varrho}_{i\vk_i, i\vk_i}(\mathbf{G})$ is a smooth and periodic function of $\vk_i$ over $\Omega^*$. 
By the standard Euler-Maclaurin formula in \cref{cor:euler_maclaurin}, the quadrature error above for each fixed $\vG$ decays super-algebraically. 
Assuming $|\hat{\rho}_{N_\vk}(\vG)|^2$ to be negligible with sufficiently large $\vG$, the summation in the Hartree energy calculation can thus be well approximated over a finite set of $\vG\in \mathbb{L}^*$. 
Then the finite-size error $E_\text{h}^\text{TDL} - E_\text{h}(N_\vk)$ decays super-algebraically with respect to $N_\vk$.

Since the potential energy due to an external potential field solely depends on the electron density, we could similarly show that the quadrature error in the potential energy also decays super-algebraically with respect to $N_k$. 

The kinetic energy in the Hartree-Fock calculation with MP mesh $\mathcal{K}$ is  computed as
\[
        E_\text{k}(N_\vk)
        = 
        \dfrac{1}{N_\vk} \sum_{\vk_i\in\mathcal{K}}
        \left( 
        \sum_{i}\int_{\Omega} |(\nabla+\I\vk_i) u_{i\vk_i}(\vr)|^2 \ud\vr 
        \right),
\]
and its TDL can be written as 
\[
E_\text{k}^\text{TDL}
= 
\dfrac{1}{|\Omega^*|}\int_{\Omega^*}\ud\vk_i
\left( 
\sum_{i}\int_{\Omega} |(\nabla+\I\vk_i) u_{i\vk_i}(\vr)|^2 \ud\vr 
\right).
\]
Thus the finite-size error in the kinetic energy can also be interpreted as the quadrature error 
\begin{equation}\label{eqn:kinetic}
E_\text{k}^\text{TDL} -   E_\text{k}(N_\vk)
=
\left(
\dfrac{1}{|\Omega^*|}\int_{\Omega^*}\ud\vk_i - \dfrac{1}{N_\vk}\sum_{\vk_i\in\mathcal{K}}
\right)\left( 
\sum_{i}\int_{\Omega} |(\nabla+\I\vk_i) u_{i\vk_i}(\vr)|^2 \ud\vr 
\right).
\end{equation}

Noting that $u_{i(\vk_i + \vG)}(\vr) = e^{-\I\vG\cdot\vr}u_{i\vk_i}(\vr)$ with any $\vG \in\mathbb{L}^*$, we could show that $|(\nabla+\I\vk_i) u_{i\vk_i}(\vr)|^2 $ is periodic with respect to $\vk_i$ over $\Omega^*$, i.e., for any $\vG \in \mathbb{L}^*$
\begin{align*}
|(\nabla+\I(\vk_i+\vG)) u_{i(\vk_i + \vG)}(\vr)|^2  
& = |\nabla ( e^{-\I\vG\cdot\vr}u_{i\vk_i}(\vr)) + \I(\vk_i+\vG)e^{-\I\vG\cdot\vr}u_{i\vk_i}(\vr)|^2
\\
& = |e^{-\I\vG\cdot\vr}  (\nabla -  \I \vG) u_{i\vk_i}(\vr) + \I(\vk_i+\vG)e^{-\I\vG\cdot\vr}u_{i\vk_i}(\vr)|^2
\\
& = |(\nabla + \I\vk_i) u_{i\vk_i}(\vr)|^2. 
\end{align*}
Thus, the integrand $\sum_{i}\int_{\Omega} |(\nabla+\I\vk_i) u_{i\vk_i}(\vr)|^2 \ud\vr$ in \cref{eqn:kinetic} is a smooth and periodic function over $\vk_i \in \Omega^*$. 
By the standard Euler-Maclaurin formula in \cref{cor:euler_maclaurin}, the quadrature error of the kinetic energy  decays super-algebraically.

\section{Low-dimensional periodic model}\label{appendix_low_dim}

The low-dimensional periodic model we consider in this paper samples $\vk$ points on a 1D-axis/2D-plane $\Omega^*_\text{low}$ in $\Omega^*$, and uses  the shifted Ewald kernel \cref{eqn:shifted_ewald} for particle interactions. 
The Madelung constant correction to the Ewald kernel is introduced based on a physical argument that the artificial interactions between a particles and its periodic images need to be removed.
From the numerical quadrature perspective, the Madelung correction is necessary since, otherwise, the leading non-smooth term $\frac{4\pi N_\text{occ}}{|\Omega||\vq|^2}$ in $\sum_{ij}\wt{F}_\text{X}^{ij}(\vk_i, \vq)$ (see \cref{eqn:leading_nonsmooth}) is not integrable over $\vq$ in $\Omega^*_\text{low}$ and thus
the exchange energy would diverge as $N_\vk \rightarrow \infty$ in the TDL. 

In this low-dimensional model, similar to \cref{eqn:madelung}, the Madelung constant is defined as 
\begin{equation*}
        \xi
        = 
        \dfrac{|\Omega^*|}{(2\pi)^3N_\vk}\sum_{\vq\in\mathcal{K}_\vq}
        \xsum_{\vG \in \mathbb{L}^*}\dfrac{4\pi e^{-\varepsilon|\vq+\vG|^2}}{|\vq + \vG|^2}
        - \dfrac{1}{(2\pi)^3}\int_{\mathbb{R}^3}\ud \vq\dfrac{4\pi e^{-\varepsilon|\vq|^2}}{|\vq|^2}
        -\frac{4\pi \varepsilon}{| \Omega |N_\vk}
        + \xsum_{\mathbf{R} \in \mathbb{L}_{\mathcal{K}_\vq}} \frac{\operatorname{erfc}\left(\varepsilon^{-1/2}|\mathbf{R}|/2\right)}{|\mathbf{R}|}.
\end{equation*}
However, the $\Gamma$-centered mesh $\mathcal{K}_\vq$ is now sampled in  $\Omega_\text{low}^*$, and the real-space lattice $\mathbb{L}_{\mathcal{K}_\vq}$ associated with $\vq+\vG$ is defined accordingly. 
For example, for an $m\times 1\times 1$ and an $m\times m \times 1$ MP meshes $\mathcal{K}_\vq$ for a quasi-1D and a quasi-2D systems, respectively, 
the lattice $\mathbb{L}_{\mathcal{K}_\vq}$ is defined as  (recall $\mathbb{L} = \{c_1\va_1 + c_2\va_2 + c_3\va_3: c_1,c_2,c_3\in\mathbb{Z}\}$ for the unit cell)
\begin{equation}\label{eqn:LKq_low}
        \begin{split}
                \mathbb{L}_{\mathcal{K}_\vq} &= \{c_1m\va_1 + c_2\va_2 + c_3\va_3: c_1,c_2,c_3\in \mathbb{Z}\},
                \\
                \mathbb{L}_{\mathcal{K}_\vq} &= \{c_1m\va_1 + c_2m\va_2 + c_3\va_3: c_1,c_2,c_3\in \mathbb{Z}\}.
        \end{split}
\end{equation}
It is worth noting that, with $N_\vk\rightarrow \infty$ and $\mathcal{K}_\vq \rightarrow \Omega_\text{low}^*$, the Madelung constant does not scale as $\Or(N_\vk^{-\frac13})$ anymore, but instead diverges to infinity, and the finite-size correction is no longer optional.

In this model, the exchange energy with a finite mesh $\mathcal{K}$ in $\Omega^*_\text{low}$ is computed as 
\begin{align}
        E_\text{x,low}(N_\vk)
        & = -\dfrac{1}{|\Omega^*_\text{low}|^2}\mathcal{Q}_{\Omega_\text{low}^*\times \Omega_\text{low}^*}(\textstyle\sum_{ij}\wt{F}_\text{X}^{ij}, \mathcal{K}\times\mathcal{K}_\vq)  + N_\text{occ}\xi
        \nonumber\\
        & = -\dfrac{1}{|\Omega^*_\text{low}|^2}\mathcal{Q}_{\Omega_\text{low}^*\times \Omega_\text{low}^*}(\textstyle\sum_{ij}\wt{F}_\text{X}^{ij} - N_\text{occ}h_\varepsilon, \mathcal{K}\times\mathcal{K}_\vq)  
        \nonumber\\
         & \quad   +N_\text{occ}
        \left(
        - \dfrac{1}{(2\pi)^3}\int_{\mathbb{R}^3}\ud \vq\dfrac{4\pi e^{-\varepsilon|\vq|^2}}{|\vq|^2}
        -\frac{4\pi \varepsilon}{| \Omega |N_\vk}
        + \xsum_{\mathbf{R} \in \mathbb{L}_{\mathcal{K}_\vq}} \frac{\operatorname{erfc}\left(\varepsilon^{-1/2}|\mathbf{R}|/2\right)}{|\mathbf{R}|}
        \right)
        \label{eqn:exchange_low_1}
\end{align}
where $h_\varepsilon$ is the auxiliary function defined in \cref{eqn:h_eps} that connects the Madelung constant correction with the singularity subtraction method in \cref{thm:exchange_madelung_error}. 
In the TDL, the exchange energy converges to 
\begin{align*}
        E_\text{x,low}^\text{TDL}
        & = 
        -\dfrac{\mathcal{I}_{\Omega_\text{low}^*\times \Omega_\text{low}^*}(\textstyle\sum_{ij}\wt{F}_\text{X}^{ij} -N_\text{occ} h_\varepsilon)}{|\Omega^*_\text{low}|^2}
        +
        N_\text{occ}
        \left(
        - \dfrac{1}{(2\pi)^3}\int_{\mathbb{R}^3}\ud \vq\dfrac{4\pi e^{-\varepsilon|\vq|^2}}{|\vq|^2}
        + \xsum_{\mathbf{R} \in \mathbb{L}_\text{low}} \frac{\operatorname{erfc}\left(\varepsilon^{-1/2}|\mathbf{R}|/2\right)}{|\mathbf{R}|}
        \right),
\end{align*}
where $\mathbb{L}_\text{low}$ denotes the lattice vectors in $\mathbb{L}$ that is perpendicular  to the extended directions, e.g., for the two $\mathbb{L}_{\mathcal{K}_\vq}$ in \cref{eqn:LKq_low}, 
\[
\mathbb{L}_\text{low} = \{c_2\va_2+c_3\va_3: c_2,c_3\in\mathbb{Z}\} \quad \text{in~quasi-1D} \quad \text{and} \quad \mathbb{L}_\text{low} = \{c_3\va_3: c_3\in\mathbb{Z}\} \quad \text{in~quasi-2D}. 
\]

We remark that this is only one  way of defining the exchange energy for low-dimensional systems and other models can lead to different definitions.
The physical reason for such ambiguity is that the electrostatic interaction of a periodic array of charged particles is not well defined without additional constraints~\cite{LeeuwPerramSmith1980}.
Mathematically, as demonstrated in \cref{thm:exchange_madelung_error}, $N_\text{occ}h_\varepsilon$ removes the leading singular term $\frac{4\pi N_\text{occ}}{|\Omega||\vq|^2}$ in $\textstyle\sum_{ij}\wt{F}_\text{X}^{ij}$. The difference is still non-smooth but scales as $\frac{\Or(|\vq|^2)}{|\vq|^2}$  near $\vq = \bm{0}$, and thus $\mathcal{I}_{\Omega_\text{low}^*\times \Omega_\text{low}^*}(\textstyle\sum_{ij}\wt{F}_\text{X}^{ij} - N_\text{occ}h_\varepsilon)$ is finite and $E_\text{x,low}^\text{TDL}$ is well-defined. 

Then the quadrature error of this model exchange energy calculation for quasi-1D and quasi-2D systems satisfies 
\[
E_\text{x,low}^\text{TDL} - E_\text{x,low}(N_\vk)
= 
-\dfrac{1}{|\Omega^*|^2}
\mathcal{E}_{\Omega_\text{low}^*\times \Omega_\text{low}^*}(\textstyle\sum_{ij}\wt{F}_\text{X}^{ij} - N_\text{occ}h_\varepsilon, \mathcal{K}\times\mathcal{K}_\vq)
+ \Or\left(N_\vk^{-1} \right)
= \Or(N_\vk^{-1}).
\]
We note that since the Madelung constant $\xi$ does not vary with respect to parameter $\varepsilon$, the definitions of both $E_\text{x,low}^\text{TDL}$ 
and $E_\text{x,low}(N_\vk)$ also do not depend on $\varepsilon$. 

\begin{rem}[Alternative correction scheme for the exchange energy in low-dimensional systems]\label{rem:alternative_exchange_correction}
For the low-dimensional model with a shifted Ewald kernel, we note that some minor modifications need to be added to the singularity-subtraction-based correction in \cref{eqn:exchange_correction2}  to make the calculation converge to the same TDL energy $E_\textup{x,low}^\textup{TDL}$, i.e., 
\begin{align}
        E_\textup{x, low}^\textup{corrected, 2} (N_\vk)
        & = -\dfrac{1}{|\Omega^*_\textup{low}|^2}\mathcal{Q}_{\Omega_\textup{low}^*\times \Omega_\textup{low}^*}(\textstyle\sum_{ij}\wt{F}_\textup{x}^{ij} - N_\textup{occ}h_\varepsilon, \mathcal{K}\times\mathcal{K}_\vq)  
        \label{eqn:exchange_correction2_low}\\
        & \quad   + N_\textup{occ}
        \left(
        - \dfrac{1}{(2\pi)^3}\int_{\mathbb{R}^3}\ud \vq\dfrac{4\pi e^{-\varepsilon|\vq|^2}}{|\vq|^2}
        + \xsum_{\mathbf{R} \in \mathbb{L}_\textup{low}} \frac{\operatorname{erfc}\left(\varepsilon^{-1/2}|\mathbf{R}|/2\right)}{|\mathbf{R}|}
        \right). 
        \nonumber
\end{align}
Compared to $E_\textup{x,low}(N_\vk)$ in \cref{eqn:exchange_low_1}, $E_\textup{x, low}^\textup{corrected, 2} (N_\vk)$ drops the  term $ -\frac{4\pi \varepsilon}{| \Omega |N_\vk}$, and changes the real-space lattice $\mathbb{L}_{\mathcal{K}_\vq}$ to $\mathbb{L}_\textup{low}$.
Both exchange energy calculations converge to $E_\textup{x,low}^\textup{TDL}$ with $\Or(N_\vk^{-1})$ error with any fixed $\varepsilon$.
However, $E_\textup{x, low}(N_\vk)$ and the Madelung constant $\xi$ are only well defined with $\Gamma$-centered mesh $\mathcal{K}_\vq$, 
while $E_\textup{x, low}^\textup{corrected, 2} (N_\vk)$ is applicable for any MP mesh $\mathcal{K}_\vq$ in $\Omega_\textup{low}^*$ that is closed under inversion. 
In the staggered mesh method for computing the exchange energy (see \Cref{subsec:staggered_ex}) for low-dimensional systems, $E_\textup{x, low}^\textup{corrected, 2} (N_\vk)$ need to be used since the involved $\mathcal{K}_\vq$ does not contain the $\Gamma$ point.
\end{rem}

The Madelung constant correction to ERIs in \cref{eqn:eri_correction} is not invoked in the MP2 energy calculation. 
Assuming the orbitals and orbital energies are exact (if the orbital energies are obtained from the Hartree-Fock calculations, then the finite-size corrections should be applied to occupied orbital energies according to \cref{rem:orbital_madelung}), the analysis of the MP2 energy remains mostly the same simply with $\Omega^*$ changed to $\Omega_\text{low}^*$ and the quadrature error, now written as, 
\[
E_\text{MP2, low}^\text{TDL} - E_\text{MP2, low}(N_\vk)
=
\dfrac{1}{|\Omega^*_\text{low}|^3}
\mathcal{E}_{(\Omega_\text{low}^*)^{\times 3}}\left(\sum_{ijab} F_\text{MP2,d}^{ijab}(\vk_i, \vk_j, \vk_a) + F_\text{MP2,x}^{ijab}(\vk_i, \vk_j, \vk_a), (\mathcal{K})^{\times 3}\right), 
\]
still scales as $\Or(m^{-d}) = \Or(N_\vk^{-1})$ for both quasi-1D and quasi-2D systems.

\end{appendices}

\end{document}